\documentclass[journal]{IEEEtran}

\ifCLASSINFOpdf
  \usepackage[pdftex]{graphicx}
  \usepackage{epstopdf}
  \DeclareGraphicsExtensions{.eps, .pdf, .jpeg}
\else
\fi

\ifCLASSINFOpdf
	\usepackage[pdftex]{graphicx}
\else
	\usepackage[dvips]{graphicx}
\fi

\usepackage{amsmath}
\usepackage{amsthm}
\usepackage{mathtools}
\usepackage{cite}
\usepackage{amsfonts}
\usepackage{array}
\usepackage{booktabs}
\usepackage{colortbl}
\usepackage{mdwmath}
\usepackage{mdwtab}
\usepackage{eqparbox}
\usepackage[ruled,vlined]{algorithm2e}

\usepackage{amssymb}
\usepackage{amsthm}

\newtheorem{theorem}{Theorem}
\newtheorem{proposition}{Proposition}
\newtheorem{lemma}{Lemma}
\newtheorem{corollary}{Corollary}

\newtheorem{definition}{Definition}
\newtheorem{remark}{Remark}
\newtheorem{assumption}{Assumption}

\usepackage{csquotes}
\usepackage[dvipsnames]{xcolor}
\usepackage{color}
\usepackage{color, soul}
\usepackage{xr}

\usepackage{array}
\usepackage{booktabs}
\usepackage{colortbl}
\usepackage{mdwmath}
\usepackage{mdwtab}

\newcommand{\figSize}[1]{\def\@figSize{#1}}
\figSize{2.5in} 

\begin{document}
%
\newif\ifTwoColumns
\TwoColumnstrue 
%
\ifTwoColumns
\title{\huge On Muting Mobile Terminals for Uplink Interference Mitigation in HetNets -- System-Level Analysis via Stochastic Geometry}
\else
\title{\Large On Muting Mobile Terminals for Uplink Interference Mitigation in HetNets -- System-Level Analysis via Stochastic Geometry}
\fi
%
%
%
\ifTwoColumns
\author{\normalsize F.~J.~Martin-Vega, M.~C.~Aguayo-Torres, G.~Gomez and M.~Di~Renzo
\thanks{F.~J.~Martin-Vega, M. C. Aguayo-Torres and G. Gomez are with the Departamento
de Ingenier\'ia de Comunicaciones, Universidad de M\'alaga, M\'alaga 29071, Spain (e-mail: fjmvega@ic.uma.es, aguayo@ic.uma.es, ggomez@ic.uma.es).}
\thanks{Marco Di Renzo is with the Laboratoire des Signaux et Syst\`emes, CNRS, CentraleSupelec, Univ Paris-Sud, Universit\'e Paris-Saclay, 91192 Gif-sur-Yvette Cedex, France (e-mail: marco.direnzo@l2s.
centralesupelec.fr).}
\thanks{A small part of this manuscript has been accepted for presentation in the conference ICC 2017.}
\thanks{This work has been submitted to the IEEE for possible publication.
Copyright may be transferred without notice, after which this version may
no longer be accessible.}}%
\else
\author{\normalsize F.~J.~Martin-Vega, M.~C.~Aguayo-Torres, G.~Gomez and M.~Di~Renzo \vspace{-1.25cm}
\thanks{F.~J.~Martin-Vega, M. C. Aguayo-Torres and G. Gomez are with the Departamento
de Ingenier\'ia de Comunicaciones, Universidad de M\'alaga, M\'alaga 29071, Spain (e-mail: fjmvega@ic.uma.es, aguayo@ic.uma.es, ggomez@ic.uma.es).}
\thanks{Marco Di Renzo is with the Laboratoire des Signaux et Syst\`emes, CNRS, CentraleSupelec, Univ Paris-Sud, Universit\'e Paris-Saclay, 91192 Gif-sur-Yvette Cedex, France (e-mail: marco.direnzo@l2s.
centralesupelec.fr).}
\thanks{A small part of this manuscript has been accepted for presentation in the conference ICC 2017.}}%
\fi
%

\markboth{Draft for IEEE Journal,~Vol.~xx, No.~xx, Month~2017}%
{Shell \MakeLowercase{\textit{et al.}}: Bare Demo of IEEEtran.cls for Journals}

\maketitle

\begin{abstract}
We investigate the performance of a scheduling algorithm where the Mobile Terminals (MTs) may be turned off if they cause a level of interference greater than a given threshold. This approach, which is referred to as Interference Aware Muting (IAM), may be regarded as an interference-aware scheme that is aimed to reduce the level of interference. We analyze its performance with the aid of stochastic geometry and compare it against other interference-unaware and interference-aware schemes, where the level of interference is kept under control in the power control scheme itself rather than in the scheduling process. IAM is studied in terms of average transmit power, mean and variance of the interference, coverage probability, Spectral Efficiency (SE), and Binary Rate (BR), which accounts for the amount of resources allocated to the typical MT. Simplified expressions of SE and BR for adaptive modulation and coding schemes are proposed, which better characterize practical communication systems. Our system-level analysis unveils that IAM increases the BR and reduces the mean and variance of the interference. It is proved that an operating regime exists, where the performance of IAM is independent of the cell association criterion, which simplifies the joint design of uplink and downlink transmissions.
\end{abstract}



\ifTwoColumns

\else
\fi
\section{Introduction}
\label{sec:Introduction}
\IEEEPARstart{I}{nterference} awareness can be exploited at both the physical and Medium Access Control (MAC) layers to boost the performance of mobile networks. It is especially useful in the Uplink (UL) of Heterogeneous Cellular Networks (HCNs) for interference mitigation and performance enhancement. 
%
In current HCNs, the Mobile Terminals (MTs) are associated with the same Base Station (BS) in the UL and Downlink (DL) \cite{Elshaer14}. The cell association is performed based on DL pilot signals and the serving BS is chosen based on a given criterion, e.g., the highest average received power in the DL. {\color{black} In the UL, the same BS is used \cite{Elshaer14} which leads to a situation where MTs are associated with distant BSs.} In this context, the use of Fractional Power Control (FPC) accentuates the detrimental effect of the MTs that cause strong interference to neighboring BSs.
%

%
%
%

\subsection{UL Analysis: State-of-the-Art}
\label{sec:Related Work} 
{\color{black} The complex aforementioned interactions between power control and association in the UL require accurate mathematical frameworks to gain insights about the performance trends and limits of existing and future networks. }
Unfortunately, the mathematical analysis of the UL of HCNs is more involved than the analysis of the DL for two main reasons: i) due to the use of power control, the transmit power of the MTs depends on the distance to their serving BSs and ii) even though the locations of BSs and MTs are drawn from two independent Poisson Point Processes (PPPs), the locations of the interfering MTs scheduled in a given orthogonal Resource Block (RB) do not follow a PPP. These two peculiarities as compared to 
to the DL make the mathematical analysis of the UL intractable without resorting to approximations \cite{Singh15}. 
{\color{black} In \cite{Novlan13} it is studied the case of homogeneous cellular networks with FPC. To avoid such a mathematical intractability, it is assumed that the MTs that are scheduled in a given RB form a Voronoi tessellation and a single BS is available in each Voronoi cell. However, such an approach does not consider HCNs. 
The case of the UL of HCNs is accurately modeled in recent works like \cite{ElSawy14, Singh15, DiRenzo16b, Martin-Vega16}, where it is considered the spatial correlation between the location of the probe BS and those of the interfering MTs. 
In \cite{ElSawy14}, is studied a framework to model HCNs with a truncated channel inversion power control under smallest path loss association. In this work it is considered an homogeneous PPP as a generative process for the locations of interfering MTs, but then, the spatial correlation is added by means of an indicator function that discards interfering MTs' locations based on their received powers. 
The case of UL and DL with decoupled access is considered in \cite{Singh15}. The association is based on maximum weighted received powers and FPC is considered in the UL. Here, to account for the spatial correlation a non-homogeneous PPP is considered to model the locations of interfering MTs. 
A framework for the UL of HCNs with multi-antena BSs is stuided in \cite{DiRenzo16b}. In this work it is considered FPC under a generalized association criteria and two extreme detection techniques in terms of complexity and performance: Maximum Ratio Combining (MRC) and Optimum Combining (OC). It is demonstrated that OC, which can be regarded as an interference-aware detection technique for multi-antenna receivers, greatly outperforms MRC when MTs use aggressive power control, i.e., when the interference is high. 
The spatial correlation is imposed by means of a conditional thinning that takes into account the generalized cell association procedure. 
Interference-awareness is also studied in \cite{Martin-Vega16}, which considers HCNs with single-antenna BSs. In this work, it is studied a power control mechanism \cite{Zhang12}, which is referred to as 
Interference Aware Fractional Power Control (IAFPC). This approach consists of introducing a maximum interference level, $i_0$, that the transmission of each MT is allowed to cause to its most interfered BS. 
In simple terms, the MTs adjust their transmit power in order to cause a maximum interference level of $i_0$ to their most interfered BS. 
%

In the present paper, we investigate another option for interference mitigation in the UL and compare it with previously reported schemes. The approach consists of exploiting interference-awareness when scheduling the transmission of the MTs, rather than in the power control scheme itself (IAFPC) or in the detection process of the receiver (OC). As a result, interference management is conducted at the MAC layer rather than at the physical layer. The considered approach is referred to as Interference Aware Muting (IAM) and consists of turning off, i.e., muting, the MTs whose interference towards the most interfered BS is above a given threshold. 
The main difference between IAFPC and IAM can be summarized as follows. In IAFPC, all the MTs are active and adjust their transmit power for interference mitigation. In IAM, on the other hand, the transmit power of the MTs does not account for any interference constraints but some MTs may not be allowed to transmit if they produce too much interference. As a result, IAM has the potential of reducing the aggregate interference in the UL and of enabling the active MTs to better use the available resources, i.e., the transmission bandwidth. 
On the other hand, it reduces the fairness of allocating the resources among the MTs, since some of them may be turned off. 
{\color{black} Nevertheless, thanks to mobility and shadowing, muted MTs are only inactive for a given period of time. Hence, from the perspective of MTs the question to answer is whether this muting increases its achievable Binary Rate (BR), taking into account both the active and inactive periods.} 
 The main objective of the present paper is to quantify the advantages and the limitations of IAM and compare it against the IAFPC scheme.
}


\subsection{Technical Contribution}
\label{sec:Contributions}
In this paper, we overcome this mathematical intractability by using an approach similar to \cite{DiRenzo16b} and \cite{Martin-Vega16}, which is referred to as \textit{conditional thinning}. In simple terms, the locations of the active MTs are assumed to be drawn from a PPP but spatial constraints (correlations) are introduced, which account for the location of the serving BS, for the location of the most interfered BS, and for the maximum level of interference allowed. 
Based on these modeling assumptions, which are validated against extensive Monte Carlo simulations, we provide the following contributions.
\begin{itemize}

\item We study IAM scheme in terms of average transmit power of the MTs, mean and variance of the interference. The mathematical analysis reveals that IAM is capable of reducing the three latter performance metrics {\color{black} compared with IAFPC}, which results in several advantages for practical implementations. 
Reducing the variance of the interference, e.g., is beneficial for better estimating the SINR and, thus, for reducing the error probability of practical decoding schemes, e.g., turbo decoding, \cite{Shin02}, and for making easier the selection of the most appropriate Modulation and Coding Scheme (MCS) to use in LTE systems \cite{Zhang12}.

\item To make our study and conclusions directly applicable to current communication systems that are based on Adaptive Modulation and Coding (AMC) transmission, we provide tractable expressions of SE and BR based on practical MCSs that are compliant with the LTE standard and whose parameters are obtained from a link-level simulator \cite{wimo14, Martin-Vega13}. 

\item With the aid of the proposed mathematical frameworks, we compare IAFPC and IAM schemes in terms of SE and BR, which provide different information on their strengths and weaknesses. The SE provides information on how well the MTs exploit the available resources (e.g., bandwidth) that are shared among the MTs served by the same BS, whereas the BR accounts for the specific fraction of resources that is allocated to each MT served by a given BS. While the IAFPC scheme is superior in terms of SE, the IAM scheme is superior in terms of RB. This implies that IAM provides service to fewer users, which get better performance compared with IAFPC. To characterize this trade-off, we investigate the fairness of both schemes, which is defined as the probability that a randomly chosen MT gets access to the resources, and provide a tractable frameworks for its analysis.

\item In light of the emerging UL-DL decoupling principle, we develop the mathematical frameworks for a General Cell Association (GCA) criterion, whose association weights may be appropriately optimized for performance enhancement. By direct inspection of the mathematical frameworks, we prove that three operating regimes can be identified as a function of the interference threshold $i_0$: i) the first, where the performance is independent of $i_0$, ii) the second, where the performance depends on $i_0$ but it does not depend on the cell association, and iii) the third, where the performance depends on $i_0$ and the cell association. Of particular interest in this paper is the second regime, which highlights that UL-DL decoupling may not be an issue for some system setups, which in turn simplifies the design of HCNs.

\item As for the relevant case study for the UL where the serving BS of the typical MT is identified based on the Smallest Path-Loss Association (SPLA) criterion with channel-inversion power control \cite{Elshaer14}, we provide simple and closed-form frameworks for relevant performance indicators and prove that two operating regimes exist: i) the first, where the performance depends on $i_0$ (interference-aware) and ii) the second, where the performance is independent of $i_0$ (interference-unaware). We prove, in addition, that i) the scaling law of the average transmit power of the MTs, the mean interference and the probability that a MT gets access to the resources is a polynomial function of $i_0$ whose exponent depends on the path-loss exponent, ii) the distance towards the serving BS gets smaller as $i_0$ increases; and iii) the CCDF of the SINR is independent of the density of BSs.

\end{itemize}

{\color{black} To the best of authors’ knowledge, all these contributions are new in the literature and are not included in previous works. For instance, the muting mechanism introduces further correlations that do not exist in \cite{Martin-Vega16} and need to be taken into account. This muting differentiates the whole analysis. New metrics like the BR, which accounts for the amount of resources allocated by the scheduler, are obtained and it is also introduced a new framework to compute the SE and BR with AMC, which is closer to real systems than Shannon formula. Finally, a lot of closed-form expressions and remarks are obtained which provide important insights about system performance, fairness and cell association. }

The remainder of this paper is organized as follows. Section \ref{sec:System model} introduces the system model and the approach for system-level analysis. In Sections \ref{sec:General Association} and \ref{sec:Smallest Path Loss Association}, the analysis of IAM is presented for GCA and SPLA criteria, respectively. The BR of AMC schemes is analyzed and discussed in Section \ref{sec:Binary Rate and Spectral Efficiency}. In Section \ref{sec:Numerical Results}, IAM and IAFPC schemes are compared against each other via numerical simulations and the main findings and performance trends derived in the paper are substantiated with the aid of Monte Carlo simulations. Finally, Section \ref{sec:Conclusion} concludes this paper.

\textit{Notation}: A summary of the main symbols and functions used throughout the present paper is provided in Table \ref{tab:Symbols} for the convenience of the readers.


\section{System Model}
\label{sec:System model}

\begin{table*}
\renewcommand{\arraystretch}{1.0}
\caption{Summary of main symbols and functions used throughout the paper.}
\label{tab:Symbols}
\ifTwoColumns
\small
\else
\scriptsize
\fi
\centering
\begin{tabular}{ l l }
\toprule
{Symbol/function} & {Definition}  \\
\toprule
${_2 F_1} (\cdot,\cdot,\cdot,\cdot)$  & Gauss Hypergeometric function \\
\hline
${\cal K} = \left\{ {1,2} \right\}$ & Tier set: tier $1$ is related to macro BSs and tier $2$ is related to small cell BSs \\
\hline
$\tilde j = \left\{ {k \in {\cal K}:k \ne j} \right\}$ &  Complementary tier, i.e. $\tilde 1 = 2$ and
$\tilde 2 = 1$ \\
\hline
$\Phi ^{(j)} ,\lambda ^{(j)}$ & PPP and its density related to the locations of macro ($j=1$) and small cell BSs ($j=2$) \\
\hline
$\lambda _{{\rm{MT}}}$ & Density of the PPP of MTs' positions \\
\hline
$\Phi,\lambda$ & PPP and its density related to the locations of all BSs \\
\hline
$t^{(j)}$ & Association weight for tier $j$ \\
\hline
$i_0 ,p_0 ,\epsilon,p_{\max }$ & Interference threshold, target receive power, partial compensation factor, and maximum transmit power \\
\hline
$\tau ,\alpha$ & Path loss slope and path loss exponent \\
\hline
${\rm{MT}}_{\rm{0}} ,{\rm{MT}}_{\rm{i}}$ & Position of the probe MT and position of a generic MT, e.g. an interfering MT \\
\hline
$\Psi ^{(k)}$ & PPP of interfering MTs's locations \\
\hline
$R_{x,(q)}^{(j)}$ & Distance (including shadowing) between location $x$ and the $q$th nearest BS from tier $j$ \\
\hline
$R_{{\rm{MT}}_{\rm{i}} } ,U_{{\rm{MT}}_{\rm{i}} } ,D_{{\rm{MT}}_{\rm{i}} }$ & Distances (including shadowing) between MTi and its serving BS, its most interfered BS and the probe BS \\
\hline
$H_{{\rm{MT}}_{\rm{i}} }$ & Power gain of the multi-path fading which is exponentially distributed \\
\hline
$p_{{\rm{MT}}} \left( r \right) = p_0 \left( {\tau r} \right)^{\alpha \epsilon }$ & Transmit power for a given distance towards the serving BS for active MTs. Muted MTs has $0$ transmit power \\
\hline
$\sigma _n^2 ,I$ & Noise power
and aggregate interference according to \textbf{Assumption \ref{assumpt: Interference}} \\
\hline
${\cal X}_{{\rm{MT}}_i }^{(j)}$ & Event defined as: \textit{${\rm MT_i}$ is associated with tier $j$} \\
\hline
${\cal Q}_{{\rm{MT}}_i }^{(m)}$ & Event defined as: \textit{the most interfered BS of ${\rm MT_i}$ belongs to tier $m$} \\
\hline
${\cal X}_{{\rm{MT}}_i }^{(j,m)}$ & Event defined as: \textit{${\rm MT_i}$ is associated with tier $j$ and the most interfered BS of ${\rm MT_i}$ belongs to tier $m$} \\
\hline
${\cal A}_{{\rm{MT}}_i }$ & Event defined as: \textit{${\rm MT_i}$ is active, i.e., non-muted} \\
\hline
${\cal O}_{{\rm{MT}}_i }^{(j,k)}$ & Event defined as: \textit{the interfering ${\rm MT_i}$ of tier $k$ receives higher weighted average power} \\ & \textit{from its serving BS than from the probe BS that belong to tier $j$} \\
\hline
${\cal Z}_{{\rm{MT}}_i }$ & Event defined as: \textit{the interfering ${\rm MT_i}$ causes a level of interference less than $i_0$ to the probe BS} \\
\bottomrule
\end{tabular}
\end{table*}


We consider the UL of a HCN composed of two tiers, $j\in\mathcal{K}=\{1,2\}$, e.g., macro and small-cell BSs, which are spatially distributed according to two independent PPPs, $\Phi^{(j)}$, of intensities $\lambda^{(j)}$. Each transmitted signal goes through an independent multi-path fading channel with Rayleigh fading and log-normal shadowing. 
The path-loss is modeled by using a path-loss slope $\tau$ and a path-loss exponent $\alpha>2$\footnote{ {\color{black}The proposed framework can be generalized to account for a bounded path loss model; however, an unbounded path loss model has been used for the sake of mathematical tractability.}}.
The cell association among MTs and BSs is based on the weighted average received power criterion, similar to \cite{Singh15}, where the association weights are denoted by $t^{(j)}$ for tier $j\in\mathcal{K}$. Hence, the $i$th MT is associated with the $n$th BS of tier $j$ if the MT is in the weighted Voronoi cell of $\mathrm{BS}^{(j)}_n$ with respect to $\Phi  = \bigcup_{j \in {\cal K}} {{\Phi ^{\left( j \right)}}}$. With these assumptions, shadowing can be modeled as a random displacement \cite{Haenggi13} of $\Phi^{(j)}$ \cite{Dhillon14, Martin-Vega16}.

For ease of writing, we introduce the event $\mathcal{X}^{(j)}_{\mathrm{MT}_i}$ as follows.
\begin{definition}
The event $\mathcal{X}^{(j)}_{\mathrm{MT}_i}$ is defined as ``$\mathrm{MT}_i$ \textit{is associated with tier} $j$''.
\end{definition}

In mathematical terms, therefore, the association criterion can be formulated as follows:
%
{ 
\begin{equation}
\label{eq:Xj}
\mathcal{X}^{(j)}_{\mathrm{MT}_i} =
	\left\{ {t^{(j)}}{\left( \tau R^{(j)}_{\mathrm{MT}_i,(1)}
	 \right)^{-\alpha}}
	>  {t^{(\tilde{j})}}{\left( \tau R^{(\tilde{j})}_{\mathrm{MT}_i,(1)}
	\right)^{-\alpha}}   \right\}
\end{equation}
}
\noindent where $(\tau R_\mathrm{MT})^{-\alpha}$ is the path-loss at a distance%
\footnote{Throughout this paper, all the distances implicitly include shadowing.}
 $R_\mathrm{MT}$ from the transmitter, $\tilde{j} = \left\{k \in \mathcal{K}: k \neq j \right\}$ is the complementary tier of $j$, i.e., $\tilde{1}=2$ and $\tilde{2}=1$, ${R}^{(\tilde{j})}_{x,(q)}$ is the distance from $x$ to the $q$th nearest BS of tier $\tilde{j}$, i.e., ${R}^{(\tilde{j})}_{x,(1)}$ is the distance to the nearest BS. The association weights $t^{(1)}$ and $t^{(2)}$ allow us to model the GCA criterion, which encompasses the SPLA criterion for $t^{(1)}=t^{(2)}$.

Throughout this paper, the analysis is performed for the probe or typical MT, i.e., for a randomly chosen MT, which is denoted by $\mathrm{MT_0}$. Its serving BS is referred to as the probe BS.

\subsection{Scheduling}
\label{sec:Scheduling}
We consider full-frequency reuse, where all the BSs share the same bandwidth. Each BS has available a bandwidth of $b_w$ Hz that is shared among the MTs that are in its Voronoi cell. In practice, $b_w$ is divided in orthogonal RBs and each scheduled MT in each cell transmits in one (or several) of these RBs. Thus, no intra-cell interference is available. This implies that a single MT per BS can interfere with the probe MT. The set of active interfering MTs of tier $k$ that are scheduled for transmission in a given RB is denoted by $\Psi^{(k)}$. For tractability, we assume that the number of RBs is large enough to be regarded as a continuous resource by the scheduler.

Based on these assumptions, the scheduling process of every BS consists of two steps:
\begin{enumerate}
\item To determine the set of active MTs. The active transmitters are the MTs that, simultaneously, cause less interference than $i_0$ to any BSs and that transmit with less power than $p_{\mathrm{max}}$. The MTs that do not fulfill these two constraints are turned off (muted).
\item Resource allocation. Once the active MTs in each cell are identified, the bandwidth of each BS is equally divided among the active MTs associated with it. Let $N^\mathcal{A}_{\mathrm{BS}^{(j)}_n}$ be the number of MTs associated with BS $\mathrm{BS}^{(j)}_n$. Each of them is allocated a bandwidth $b_w/N^\mathcal{A}_{\mathrm{BS}^{(j)}_n}$%
\footnote{Although, in practice, the bandwidth is divided in RBs, we assume that it can be treated as a continuous resource and hence that it can be equally divided among the active MTs. This is assumed in \cite{Singh15} as well.}
 Hz.
\end{enumerate}

This scheduling process characterizes the IAM scheme and makes it different from the IAFPC scheme in \cite{Martin-Vega16}. In \cite{Martin-Vega16}, all the MTs are active and power control is responsible for controlling the level of interference, by making sure that the interference level at any BS is less than $i_0$.

To better understand the implications of interference awareness on turning off (muting) some MTs, we analyze the case study $i_0 \to \infty$ as well, which is referred to as Interference-Unaware Muting (IUM)\footnote{In the present paper, IUM and Interference Unaware FPC (IUFPC) schemes are similar but slightly different. IUM is referred to a setup where $i_0 \to \infty$ and $p_{{\rm{max}}} < \infty$. IUFPC is referred to a setup where $i_0 \to \infty$ and $p_{{\rm{max}}} \to \infty$. As for IUM, only the constraint on the maximum transmit power exists. As for IUFPC, there is no constraint on either the maximum transmit power or the maximum interference.}. 
%
For ease of writing, we introduce some definitions that are useful for mathematical analysis.
\begin{definition}
The event $\mathcal{Q}^{(m)}_{\mathrm{MT}_i}$ is defined as ``\textit{the most interfered BS of MT$_i$ belongs to tier m}''.
\end{definition}
\begin{definition}
The event $\mathcal{X}^{(j,m)}_{\mathrm{MT}_i} = \mathcal{X}^{(j)}_{\mathrm{MT}_i} \cap \mathcal{Q}^{(m)}_{\mathrm{MT}_i}$ is defined as ``\textit{MT$_i$ is associated with tier j and the most interfered BS of MT$_i$ belongs to tier m}''.
\end{definition}

In mathematical terms, $\mathcal{X}^{(j,m)}_{\mathrm{MT}_i}$ can be formulated as follows:
\ifTwoColumns
\begin{align}
\label{eq:Xjm}
	&\mathcal{X}^{(j,m)}_{\mathrm{MT}_i} = \mathcal{X}^{(j)}_{\mathrm{MT}_i} \cap
	\overbrace{ \left\{ R^{(j)}_{\mathrm{MT}_i,(2)} > R^{(m)}_{\mathrm{MT}_i,(1)} \right\}}
	^{\mathcal{Q}^{(m)}_{\mathrm{MT}_i}},\, \mathrm{if} j \neq m \nonumber \\	
	&\mathcal{X}^{(j,m)}_{\mathrm{MT}_i} = \mathcal{X}^{(j)}_{\mathrm{MT}_i} \cap
	\overbrace{ \left\{ R^{(j)}_{\mathrm{MT}_i,(2)} < R^{(\tilde{j})}_{\mathrm{MT}_i,(1)}
	\right\} }
	^{\mathcal{Q}^{(m)}_{\mathrm{MT}_i}} ,\, \mathrm{if} j = m
\end{align}
\else
\begin{align}
\label{eq:Xjm}
	&\mathcal{X}^{(j,m)}_{\mathrm{MT}_i} = \mathcal{X}^{(j)}_{\mathrm{MT}_i} \cap
	\overbrace{ \left\{ R^{(j)}_{\mathrm{MT}_i,(2)} > R^{(m)}_{\mathrm{MT}_i,(1)} \right\}}
	^{\mathcal{Q}^{(m)}_{\mathrm{MT}_i}}\, \mathrm{if} \; j \neq m;\,
	\mathcal{X}^{(j,m)}_{\mathrm{MT}_i} = \mathcal{X}^{(j)}_{\mathrm{MT}_i} \cap
	\overbrace{ \left\{ R^{(j)}_{\mathrm{MT}_i,(2)} < R^{(\tilde{j})}_{\mathrm{MT}_i,(1)}
	\right\} }
	^{\mathcal{Q}^{(m)}_{\mathrm{MT}_i}} \, \mathrm{if} \; j = m
\end{align}
\fi

According to IAM, the MTs that either cause higher interference than $i_0$ or transmit with higher power than $p_\mathrm{max}$ are kept silent. The set of active MTs is defined as follows.
\begin{definition}
The event $\mathcal{A}_{\mathrm{MT}_i}$ is defined as ``\textit{MT$_i$ is active}''.
\end{definition}

In mathematical terms, $\mathcal{A}_{\mathrm{MT}_i}$ can be formulated as follows:
\ifTwoColumns
\begin{align}
\label{eq:A_MTi}
	\mathcal{A}_{\mathrm{MT}_i} &=
	\left\{ \left( p_{\rm MT} \left( R_{\mathrm{MT}_i} \right)
	< i_0 \left( \tau U_{\mathrm{MT}_i} \right)^{\alpha} \right. \right.
	\nonumber \\
	&\cap \left. \left. p_{\rm MT} \left( R_{\mathrm{MT}_i} \right)
	< p_\mathrm{max} \right)  \right\}
\end{align}
\else
\begin{align}
\label{eq:A_MTi}
	\mathcal{A}_{\mathrm{MT}_i} =
	\left\{  \left( p_{\rm MT} \left(  R_{\mathrm{MT}_i} \right)
	< i_0 \left( \tau U_{\mathrm{MT}_i} \right)^{\alpha} \right)
	\cap \left( p_{\rm MT} \left(  R_{\mathrm{MT}_i} \right)
	< p_\mathrm{max}  \right) \right\}
\end{align}
\fi
%
%
\noindent where $p_{\rm MT}(r)$, $p_0$ and $\epsilon$ are related to power control and they are described in Table \ref{tab:Symbols},  
$R_{\mathrm{MT}_i}$ is the distance between $\mathrm{MT_i}$ and its serving BS, and $U_{\mathrm{MT}_i}$ is the distance between $\mathrm{MT_i}$ and its most interfered BS. If the probe MT is associated with tier $j$, i.e., the event $\mathcal{X}^{(j)}_{\mathrm{MT}_i}$ is true, then $R_{\mathrm{MT}_i} = R^{(j)}_{\mathrm{MT}_i,(1)}$. The distance $U_{\mathrm{MT}_i}$ depends, on the other hand, on the event $\mathcal{X}^{(j,m)}_{\mathrm{MT}_i}$. Accordingly, $U_{\mathrm{MT}_i} = R^{(m)}_{\mathrm{MT}_i,(1)}$ if $j \neq m$ and $U_{\mathrm{MT}_i} = R^{(j)}_{\mathrm{MT}_i,(2)}$ if $j = m$.
{\color{black} The aim of event $\mathcal{A}_{\rm MT_i}$ is to capture the spatial correlation between the position of a given MT, its serving BS and it most interfered BS, which follows from the muting process. }

As far as IAM is concerned, fractional power control is applied at the physical layer and is interference-unaware, i.e., the transmit power of the MTs that are not turned off depends only on path-loss and shadowing and it can be expressed as $p_\mathrm{MT} \left( R_{\mathrm{MT}_0} \right)$. If the MTs are muted, on the other hand, their transmit power is equal to zero. This implies that their associated ${\rm SINR}$, ${\rm BR}$, etc. are, by definition, equal to zero as well.

%
%

\subsection{SINR}
\label{sec:SINR}
The SINR of the typical active MT that is measured at the probe BS can be formulated as:
\begin{equation}
\label{eq:SINR}
	\mathrm{SINR}_{\mathrm{MT}_0} = \frac{ H_{\mathrm{MT}_0} \left(\tau R_{\mathrm{MT}_0} \right)^{-\alpha}
	p_\mathrm{MT} \left(R_{\mathrm{MT}_0} \right) }{I + \sigma_n^2 }
\end{equation}
where $H_{\mathrm{MT}_0}$ is the channel gain, $R_{\mathrm{MT}_0}$ is the distance from the serving BS, $p_\mathrm{MT} \left(R_{\mathrm{MT}_0} \right)$ is the transmit power, $I$ is the other-cell interference, and $\sigma_n^2$ is the noise power.

In the UL, as discussed in Section \ref{sec:Introduction}, the set of interfering MTs does not constitute a PPP, even though the MTs and BSs are distributed according to a PPP. Further details can be found in \cite{DiRenzo16b} and \cite{Martin-Vega16}. This makes the mathematical analysis intractable. In the present paper, the distinctive scheduling process of IAM negatively affects the mathematical tractability of the problem at hand even further. To make the analysis tractable, some approximations for modeling the set of active MTs are needed. In \cite{DiRenzo16b} and \cite{Martin-Vega16}, it is shown that a tractable approximation consists of assuming that the set of active MTs can still be modeled as a PPP, provided that appropriate spatial constraints on the locations of the MTs are introduced. Stated differently, the set of active MTs is modeled as a spatially-thinned PPP or equivalently as a non-homogeneous PPP.

{\color{black} Before introducing the approach to model interfering MTs' locations, the following events need to be defined:}
%
%
\begin{definition}
The event $\mathcal{O}^{(j,k)}_{\mathrm{MT}_i}$ is defined as ``\textit{the interfering MT$_i$ of tier $k$ receives higher weighted average power from its serving BS than from the probe BS that belongs to tier $j$}''.
\end{definition}

In mathematical terms, $\mathcal{O}^{(j,k)}_{\mathrm{MT}_i}$ can be formulated as follows:
\begin{equation}
\label{eq:O_MTi}
	\mathcal{O}^{(j,k)}_{\mathrm{MT}_i} = \left\{ t^{(k)}
	\left( \tau R_{\mathrm{MT}_i} \right)^{-\alpha} >
	t^{(j)} \left( \tau D_{\mathrm{MT}_i} \right)^{-\alpha} \right\}
\end{equation}

\begin{definition}
The event $\mathcal{Z}_{\mathrm{MT}_i}$ is defined as ``\textit{the interfering MT$_i$ causes a level of interference less than $i_0$ to the probe BS}''.
\end{definition}

In mathematical terms, $\mathcal{Z}_{\mathrm{MT}_i}$ can be formulated as follows:
\begin{equation}
\label{eq:ZMTi}
\mathcal{Z}_{\mathrm{MT}_i} =
    \left\{ p_{\mathrm{MT}} \left( R_{\mathrm{MT}_i}\right) \left( \tau D_{\mathrm{MT}_i} \right)^{-\alpha}
     < i_0 \right\}
\end{equation}

Hence, inspired by \cite{DiRenzo16b} and \cite{Martin-Vega16}, our mathematical framework is based on the following approximation.
\begin{assumption}
\label{assumpt: Interference}
The other-cell interference of the typical \textit{active} MT is approximated as \cite{Martin-Vega16}:
\begin{equation}
\label{eq:I}
	I \approx \sum_{k \in \mathcal{K}} \sum_{\mathrm{MT}_i \in \Psi^{(k)}}
	\frac{H_{\mathrm{MT}_i} p_{\mathrm{MT}} \left( R_{\mathrm{MT}_i} \right)}
	{ \left(\tau D_{\mathrm{MT}_i} \right)^{\alpha} }	
	\mathbf{1} \left( \mathcal{O}^{(j,k)}_{\mathrm{MT}_i} \right)
	\mathbf{1} \left( \mathcal{Z}_{\mathrm{MT}_i} \right)
\end{equation}
\noindent where $\Psi^{(k)}$ is a PPP of intensity $\lambda^{(k)}$ whose points constitute the locations of the interfering MTs that are scheduled for transmission in the same RB as that of the typical MT, the events $\mathcal{O}^{(j,k)}_{\mathrm{MT}_i}$ and $\mathcal{Z}_{\mathrm{MT}_i}$ take into account the necessary spatial constraints imposed by the cell association criterion and the maximum interference and power constraints, respectively, $R_{\mathrm{MT}_i}$ and $D_{\mathrm{MT}_i}$ are the distances from ${\mathrm{MT}_i}$ to its own serving BS and to the probe BS, respectively.
\end{assumption}

More specifically, i) the event $\mathcal{O}^{(j,k)}_{\mathrm{MT}_i}$ is necessary to account for the spatial correlation that exists between the locations of the probe BS, the interfering MTs and their serving BSs, since the interfering MTs must lie outside the Voronoi cell of the probe BS by definition of cell association, and ii) the event $\mathcal{Z}_{\mathrm{MT}_i}$ is necessary to account for the fact that the interfering MTs need to cause less interference than $i_0$ according to the IAM scheduling process. 

The next two sections provide mathematical expressions of the CCDF of the SINR and of the mean and variance of the other-cell interference for GCA and SPLA cell association criteria respectively. 
%
%

\section{General Cell Association Criterion}
\label{sec:General Association}
We start introducing some enabling results for proving the main theorems of this section.
\begin{proposition}
\label{prop:IAM Xjm}
The probability that the typical MT is active and is associated with tier $j$ is:
\ifTwoColumns
\begin{align}
\label{eq:Pr(Xjm and A)}
\Pr \left( {{\cal A}_{{\rm{MT}}_{\rm{0}} } } \right) & = \sum\limits_{j \in {\cal K}} {\int\limits_{v > 0} {{\bf{1}}\left( {v < \frac{1}{\tau }\left( {\frac{{p_{\max } }}{{p_0 }}} \right)^{\frac{1}{{\alpha }}} } \right)} }
\nonumber \\
& \times \left( {\nu ^{(j)} \left( v \right) + \eta ^{(j)} \left( v \right)} \right){\rm{d}}v
\end{align}
\else
\begin{align}
\label{eq:Pr(Xjm and A)}
\Pr \left( {{\cal A}_{{\rm{MT}}_{\rm{0}} } } \right) = \sum\limits_{j \in {\cal K}} {\int\limits_{v > 0} {{\bf{1}}\left( {v < \frac{1}{\tau }\left( {\frac{{p_{\max } }}{{p_0 }}} \right)^{\frac{1}{{\alpha }}} } \right)} } \left( {\nu ^{(j)} \left( v \right) + \eta ^{(j)} \left( v \right)} \right){\rm{d}}v
\end{align}
\fi

\noindent where $\nu^{(j)}(v)$ and $\eta^{(j)}(v)$ are defined in (\ref{eq:nu}) and (\ref{eq:kappa}), respectively.

\begin{figure*}[t]
\ifTwoColumns
\normalsize
\else
\small
\fi
\begin{align}
\label{eq:nu}
	& \nu^{(j)} (v) = 2 \pi v \lambda^{(j)} \Bigg( \mathrm{e}^{-\pi
	\lambda^{(j)} v^2 } \times
    \left( \mathrm{e}^{-\pi \lambda^{(\tilde{j})} \mathrm{max}^2 \left( \left( \frac{p_0}{i_0} \right)
    ^{\frac{1}{\alpha}} \frac{(\tau v)^\epsilon}{\tau},
    \left( \frac{t^{(\tilde{j})}}{t^{(j)}} \right)^{\frac{1}{\alpha}} v \right)}
    - \mathrm{e}^{-\pi \lambda^{(\tilde{j})} v^2}  \right) \times \nonumber \\
    & \mathbf{1} \left( v > \mathrm{max} \left( \left( \frac{p_0}{i_0} \right)
    ^{\frac{1}{\alpha}} \frac{(\tau v)^\epsilon}{\tau},
    \left( \frac{t^{(\tilde{j})}}{t^{(j)}} \right)^{\frac{1}{\alpha}} v \right) \right)
    + \frac{\lambda^{(\tilde{j})}}{\lambda^{(j)} + \lambda^{(\tilde{j})}}
    \mathrm{e}^{-\pi \left( \lambda^{(j)} + \lambda^{(\tilde{j})} \right)
    \mathrm{max}^2 \left(\left( \frac{p_0}{i_0} \right)
    ^{\frac{1}{\alpha}} \frac{(\tau v)^\epsilon}{\tau},
    \left( \frac{t^{(\tilde{j})}}{t^{(j)}} \right)^{\frac{1}{\alpha}} v, v \right) }   \Bigg)
\end{align}
\hrulefill
\begin{align}
\label{eq:kappa}
	& \eta^{(j)} (v) = 2 \pi v \lambda^{(j)} \Bigg( \mathrm{e}^{-\pi
	\lambda^{(\tilde{j})} \left( \frac{t^{(\tilde{j})}}{t^{(j)}}
	\right)^{\frac{2}{\alpha}} v^2 } \times
    \left( \mathrm{e}^{-\pi \lambda^{(j)} \mathrm{max}^2 \left( \left( \frac{p_0}{i_0} \right)
    ^{\frac{1}{\alpha}} \frac{(\tau v)^\epsilon}{\tau}, v \right)}
    - \mathrm{e}^{-\pi \lambda^{(j)}  \left( \frac{t^{(\tilde{j})}}{t^{(j)}}
	\right)^{\frac{2}{\alpha}} v^2 }   \right) \times \nonumber \\
    & \mathbf{1} \left( v >  \left( \frac{t^{(j)}}{t^{(\tilde{j})}}
	\right)^{\frac{1}{\alpha}} \mathrm{max} \left( \left( \frac{p_0}{i_0} \right)
    ^{\frac{1}{\alpha}} \frac{(\tau v)^\epsilon}{\tau}, v \right) \right)
    + \frac{\lambda^{(j)}}{\lambda^{(j)} + \lambda^{(\tilde{j})}}
    \mathrm{e}^{-\pi \left( \lambda^{(j)} + \lambda^{(\tilde{j})} \right)
    \mathrm{max}^2 \left( \left( \frac{p_0}{i_0} \right)
    ^{\frac{1}{\alpha}} \frac{(\tau v)^\epsilon}{\tau},
    v, \left( \frac{t^{(\tilde{j})}}{t^{(j)}} \right)^{\frac{1}{\alpha}} v \right) }   \Bigg)
\end{align}
\hrulefill
\end{figure*}
\end{proposition}

\begin{proof}
See Appendix \ref{app:Proof of Proposition IAM Xjm}.
\end{proof}

\textbf{Proposition \ref{prop:IAM Xjm}} is useful for understanding and quantifying the fairness of the IAM scheme. The higher $\Pr \left( {{\cal A}_{{\rm{MT}}_{\rm{0}} } } \right)$ is, in fact, the higher the probability that a randomly chosen MT is served in a given RB and, thus, the higher the fairness that it gets access to the available resources is\footnote{The system fairness can be defined in different ways. In \cite{Lin15}, e.g., it is defined based on the proportionally fair criterion and is obtained by computing the logarithm of the average rate. Our framework could be generalized for analyzing the system fairness based on this definition, but this study is outside the scope of the current paper and is postponed to future research.}.

\begin{lemma}
\label{prop:IAM f_RMT_0 cond XjmA}
The Probability Density Function (PDF) of the distance between the typical MT and its serving BS by conditioning on the event $\mathcal{X}^{(j,m)}_{\mathrm{MT}_0} \cap \mathcal{A}_{\mathrm{MT}_0}$ can be formulated as follows:
\ifTwoColumns
\begin{align}
\label{eq:f_RMT_0 cond XjmA IAM}
    f_{R_{\mathrm{MT}_{0}}} & \left( v|\mathcal{X}_{\mathrm{MT}_{0}}^{(j,m)},
    \mathcal{A}_{\mathrm{MT}_{0}} \right) = \nonumber \\
    &\begin{cases}
    \frac{\nu ^{(j)}\left( v \right)}{\Pr \left( \mathcal{X}_{
    \mathrm{MT}_{0}}^{(j,m)},\mathcal{A}_{\mathrm{MT}_{0}}^{(j,m)} \right)},
    v<\frac{1}{\tau }\left( \frac{p_{\mathrm{max} }}{p_{0}}
    \right)^{\frac{1}{\alpha \epsilon }}, &\quad \mathrm{if  }j\ne m \\
    \frac{\eta ^{(j)}\left( v \right)}{\Pr \left( \mathcal{X}_{
    \mathrm{MT}_{0}}^{(j,j)},\mathcal{A}_{\mathrm{MT}_{0}}^{(j,j)} \right)},
    v<\frac{1}{\tau }\left( \frac{p_{\mathrm{max} }}{p_{0}} \right)^{
    \frac{1}{\alpha \epsilon }}, &\quad \mathrm{if  }j=m
    \end{cases}
\end{align}
\else
\begin{align}
\label{eq:f_RMT_0 cond XjmA IAM}
    f_{R_{\mathrm{MT}_{0}}} & \left( v|\mathcal{X}_{\mathrm{MT}_{0}}^{(j,m)},
    \mathcal{A}_{\mathrm{MT}_{0}} \right) =
    \begin{cases}
    \frac{\nu ^{(j)}\left( v \right)}{\Pr \left( \mathcal{X}_{
    \mathrm{MT}_{0}}^{(j,m)},\mathcal{A}_{\mathrm{MT}_{0}}^{(j,m)} \right)} &\quad \mathrm{for  } \;
    0<v<\frac{1}{\tau }\left( \frac{p_{\mathrm{max} }}{p_{0}}
    \right)^{\frac{1}{\alpha \epsilon }} \; \mathrm{and  }  \; j\ne m \\
    \frac{\eta ^{(j)}\left( v \right)}{\Pr \left( \mathcal{X}_{
    \mathrm{MT}_{0}}^{(j,j)},\mathcal{A}_{\mathrm{MT}_{0}}^{(j,j)} \right)} &\quad \mathrm{for  } \;
    0<v<\frac{1}{\tau }\left( \frac{p_{\mathrm{max} }}{p_{0}} \right)^{
    \frac{1}{\alpha \epsilon }}\; \mathrm{and  }  \; j=m
    \end{cases}
\end{align}
\fi
\noindent where $\nu^{(j)}\left( v \right)$ and $\eta^{(j)}\left( v \right)$ are defined in (\ref{eq:nu}) and (\ref{eq:kappa}), respectively.
\end{lemma}
\begin{proof}
The Cumulative Distribution Function (CDF) of the distance between the typical MT and its serving BS by conditioning on the MT being active and on $\mathcal{X}_{\text{MT}_{\text{0}}}^{(j,m)}$ is obtained by using steps similar to Appendix \ref{app:Proof of Proposition IAM Xjm}. The PDF is obtained from the CDF by computing the derivative.
\end{proof}

In the UL, an important performance metric to study is the average transmit power of the typical MT, which is related to its power consumption. Since some MTs may be turned off in the IAM scheme, this implies that some MTs may transmit zero power, which results in reducing their power consumption. The following proposition provides the average transmit power of the typical MT, by taking into account that the \textit{typical} MT may be a MT that is turned off as it does not fulfill either the maximum power constraint or the maximum interference constraint.
\begin{proposition}
\label{prop:IAM E(P_MT)}
The average transmit power of the typical MT can be formulated as follows:
\ifTwoColumns
\begin{align}
\label{eq:IAM E(P_MT)}
    \hspace{-0.25cm}\mathbb{E} & \left[ P_{\mathrm{MT}_{0}} \right]=\sum\limits_{j\in
    \mathcal{K}}{\sum\limits_{m\in \mathcal{K}}{\Pr \left(
    \mathcal{X}_{\mathrm{MT}_{0}}^{(j,m)},\mathcal{A}_{\mathrm{MT}_{0}} \right)}}
    \int\limits_{p=0}^{\infty } \frac{p}{\tau p_{0}\alpha \epsilon }
    \nonumber \\
    & \times \left( \frac{p}{p_{0}} \right)^{\frac{1}{\alpha \epsilon }-1}
    f_{R_{\mathrm{MT}_{0}}}
    \left( \frac{1}{\tau }\left( \frac{p}{p_{0}} \right)^{\frac{1}{
    \alpha \epsilon }}|\mathcal{X}_{
    \mathrm{MT}_{0}}^{(j,m)},\mathcal{A}_{\mathrm{MT}_{0}} \right) \mathrm{d}p
\end{align}
\else
\begin{align}
\label{eq:IAM E(P_MT)}
    \hspace{-0.5cm} \mathbb{E}  \left[ P_{\mathrm{MT}_{0}} \right]=\sum\limits_{j\in
    \mathcal{K}}{\sum\limits_{m\in \mathcal{K}}{\Pr \left(
    \mathcal{X}_{\mathrm{MT}_{0}}^{(j,m)},\mathcal{A}_{\mathrm{MT}_{0}} \right)}}
    \int\limits_{0}^{\infty } \frac{p}{\tau p_{0}\alpha \epsilon }
      \left( \frac{p}{p_{0}} \right)^{\frac{1}{\alpha \epsilon }-1}
    f_{R_{\mathrm{MT}_{0}}}
    \left( \frac{1}{\tau }\left( \frac{p}{p_{0}} \right)^{\frac{1}{
    \alpha \epsilon }}|\mathcal{X}_{
    \mathrm{MT}_{0}}^{(j,m)},\mathcal{A}_{\mathrm{MT}_{0}} \right) \mathrm{d}p
\end{align}
\fi
\noindent where $f_{R_{\mathrm{MT}_{0}}}\left(v |\mathcal{X}_{
    \mathrm{MT}_{0}}^{(j,m)},\mathcal{A}_{\mathrm{MT}_{0}}\right)$ is in (\ref{eq:f_RMT_0 cond XjmA IAM}) and ${\Pr \left(
    \mathcal{X}_{\mathrm{MT}_{0}}^{(j,m)},\mathcal{A}_{\mathrm{MT}_{0}} \right)}$ is defined in Appendix A.
\end{proposition}

\begin{proof}
It follows by computing the average transmit power by conditioning on the events $\mathcal{A}_{\mathrm{MT}_{0}}$ and $\mathcal{X}_{\text{MT}_{\text{0}}}^{(j,m)}$. The final result is obtained from the total probability theorem.
\end{proof}

\begin{remark}[Exact analysis]
\label{rem:Exact analysis}
The previous propositions and lemmas are exact, since they do not depend on the set of active interfering MTs but depend only on the locations of the BSs, which constitute a PPP, and on the typical MT. In other words, \textbf{Assumption \ref{assumpt: Interference}} is not applied.
\end{remark}

The next lemma provides the Laplace transform of the other-cell interference based on its mathematical formulation in (\ref{eq:I}), which exploits \textbf{Assumption \ref{assumpt: Interference}}.
\begin{lemma}
\label{prop:IAM LI}
Assume that the typical MT is associated with a BS of tier $j$. The Laplace transform of the (conditional) interference in (\ref{eq:I}) can be formulated as follows:
\begin{align}
\label{eq:IAM LI(s|Xj)}
    \mathcal{L}_{I}\left( s|\mathcal{X}_{\mathrm{MT}_{0}}^{(j)} \right)
    = \exp \left( \beta^{(j)} (s)  \right)
\end{align}
\noindent where $\beta^{(j)} (s)$ is defined as follows:
\ifTwoColumns
\begin{align}
\label{eq:IAM theta(s)}
    \beta^{(j)} (s)
    & = -\sum\limits_{k\in \mathcal{K}} {2\pi \lambda ^{(k)}
    \sum\limits_{n\in \mathcal{K}}{\Pr \left( \mathcal{Q}_{
    \mathrm{MT}_{i}}^{(n)}|\mathcal{X}_{\mathrm{MT}_{i}}^{(k)},\mathcal{A}_{
    \mathrm{MT}_{i}} \right)}}
    \nonumber \\
    & \int\limits_{0}^{\infty }
    {f_{R_{\mathrm{MT}_{i}}}\left( r|\mathcal{X}_{\mathrm{MT}_{i}}^{(k,n)}
    ,\mathcal{A}_{\mathrm{MT}_{i}} \right)} \chi \left( s,r \right)\mathrm{d}r,
\end{align}
\else
\begin{align}
\label{eq:IAM theta(s)}
    \beta^{(j)} (s)
     = -\sum\limits_{k\in \mathcal{K}} {2\pi \lambda ^{(k)}
    \sum\limits_{n\in \mathcal{K}}{\Pr \left( \mathcal{Q}_{
    \mathrm{MT}_{i}}^{(n)}|\mathcal{X}_{\mathrm{MT}_{i}}^{(k)},\mathcal{A}_{
    \mathrm{MT}_{i}} \right)}}
    & \int\limits_{0}^{\infty }
    {f_{R_{\mathrm{MT}_{i}}}\left( r|\mathcal{X}_{\mathrm{MT}_{i}}^{(k,n)}
    ,\mathcal{A}_{\mathrm{MT}_{i}} \right)} \chi \left( s,r \right)\mathrm{d}r,
\end{align}
\fi

\noindent $f_{R_{\mathrm{MT}_{i}}} \left( r|\mathcal{X}_{\mathrm{MT}_{i}}^{(k,n)},\mathcal{A}_{\mathrm{MT}_{i}} \right)$ is the PDF of the distance between the $i$th interfering MT and its serving BS, which is provided in Lemma \ref{prop:IAM f_RMT_0 cond XjmA}, $\chi(s,r)$ is defined as follows:
\ifTwoColumns
\begin{align}
\label{eq:chi(s,r)}
    \chi( s, & r )=\frac{p_{0}s\left( \tau r \right)^{\alpha \epsilon }
    \tau ^{-\alpha }}{\alpha -2}\mathrm{max}^{2-\alpha }\left( \left(
    \frac{t^{(j)}}{t^{(k)}} \right)^{\frac{1}{\alpha }}r,\left(
    \frac{p_{0}}{i_{0}} \right)^{\frac{1}{\alpha }}\frac{\left( \tau r \right)
    ^{\epsilon }}{\tau } \right) \nonumber \\
    & _{2}\mathrm{F}_{1}\left( 1,
    \frac{\alpha -2}{\alpha },2-\frac{2}{\alpha },-p_{0}s\left( \tau r \right)
    ^{\alpha \epsilon }\tau ^{-\alpha } \right.
    \nonumber \\
    & \left. \mathrm{max}^{-\alpha }\left( \left(
    \frac{t^{(j)}}{t^{(k)}} \right)^{\frac{1}{\alpha }}r,\left(
    \frac{p_{0}}{i_{0}} \right)^{\frac{1}{\alpha }}\frac{\left( \tau r \right)
    ^{\epsilon }}{\tau } \right) \right)
\end{align}
\else
\begin{align}
\label{eq:chi(s,r)}
    \chi( s, r ) &=\frac{p_{0}s\left( \tau r \right)^{\alpha \epsilon }
    \tau ^{-\alpha }}{\alpha -2}\mathrm{max}^{2-\alpha }\left( \left(
    \frac{t^{(j)}}{t^{(k)}} \right)^{\frac{1}{\alpha }}r,\left(
    \frac{p_{0}}{i_{0}} \right)^{\frac{1}{\alpha }}\frac{\left( \tau r \right)
    ^{\epsilon }}{\tau } \right) \nonumber \\
    & \times _{2}\mathrm{F}_{1}\left( 1,
    \frac{\alpha -2}{\alpha },2-\frac{2}{\alpha },-p_{0}s\left( \tau r \right)
    ^{\alpha \epsilon }\tau ^{-\alpha } \right.
     \left. \mathrm{max}^{-\alpha }\left( \left(
    \frac{t^{(j)}}{t^{(k)}} \right)^{\frac{1}{\alpha }}r,\left(
    \frac{p_{0}}{i_{0}} \right)^{\frac{1}{\alpha }}\frac{\left( \tau r \right)
    ^{\epsilon }}{\tau } \right) \right)
\end{align}
\fi
\noindent and $\Pr \left( \mathcal{Q}_{\mathrm{MT}_{i}}^{(n)}|\mathcal{X}_{\mathrm{MT}_{i}}^{(k)},\mathcal{A}_{\mathrm{MT}_{i}} \right)$ is defined as follows:
\ifTwoColumns
\begin{align}
\label{eq:IAM LI(s|Xj) Part 2}
    \Pr \left( \mathcal{Q}_{\mathrm{MT}_{i}}^{(n)}|\mathcal{X}_{\mathrm{MT}_{i}}^{(k)}
    ,\mathcal{A}_{\mathrm{MT}_{i}} \right)
    &= \frac{\Pr \left(
    \mathcal{X}_{\mathrm{MT}_{i}}^{(k,n)},\mathcal{A}_{\mathrm{MT}_{i}} \right)}
    {\Pr \left( \mathcal{X}_{\mathrm{MT}_{i}}^{(k)},\mathcal{A}_{\mathrm{MT}_{i}}
    \right)} \nonumber \\
    &=\frac{\Pr \left( \mathcal{X}_{\mathrm{MT}_{i}}^{(k,n)}
    ,\mathcal{A}_{\mathrm{MT}_{i}} \right)}{\sum\limits_{q\in \mathcal{K}}
    {\Pr \left( \mathcal{X}_{\mathrm{MT}_{i}}^{(k,q)},\mathcal{A}_{\mathrm{MT}_{i}} \right)}}
\end{align}
\else
\begin{align}
\label{eq:IAM LI(s|Xj) Part 2}
    \Pr \left( \mathcal{Q}_{\mathrm{MT}_{i}}^{(n)}|\mathcal{X}_{\mathrm{MT}_{i}}^{(k)}
    ,\mathcal{A}_{\mathrm{MT}_{i}} \right)
    = \frac{\Pr \left(
    \mathcal{X}_{\mathrm{MT}_{i}}^{(k,n)},\mathcal{A}_{\mathrm{MT}_{i}} \right)}
    {\Pr \left( \mathcal{X}_{\mathrm{MT}_{i}}^{(k)},\mathcal{A}_{\mathrm{MT}_{i}}
    \right)}
    =\frac{\Pr \left( \mathcal{X}_{\mathrm{MT}_{i}}^{(k,n)}
    ,\mathcal{A}_{\mathrm{MT}_{i}} \right)}{\sum\nolimits_{q\in \mathcal{K}}
    {\Pr \left( \mathcal{X}_{\mathrm{MT}_{i}}^{(k,q)},\mathcal{A}_{\mathrm{MT}_{i}} \right)}}
\end{align}
\fi
\end{lemma}
\begin{proof}
See Appendix \ref{app:IAM LI}.
\end{proof}

From the Laplace transform in \eqref{eq:IAM LI(s|Xj)}, the moments of the interference can be obtained as shown in the next proposition. Of particular interest is the variance of the interference, since its affects the performance of AMC schemes \cite{Zhang12}: the smaller the variance is, the more robust and accurate the estimation of the SINR is, which makes easier the choice of the best MCS to use.
\begin{proposition}
\label{prop:IAM E(I) var(I)}
The mean and variance of the interference can be formulated as follows:
\ifTwoColumns
\begin{align}
\label{eq:IAM var(I)}
    \mathbb{E}\left[ I \right] &= -\sum\limits_{j\in \mathcal{K}}
    {\Pr \left(\mathcal{X}_{\mathrm{MT}_{0}}^{(j)} \right)\beta '^{(j)}\left( 0 \right)}
    \\
    \operatorname{var}\left( I \right) &=-\sum\limits_{j\in \mathcal{K}}{\Pr \left(
    \mathcal{X}_{\mathrm{MT}_{0}}^{(j)} \right)} \nonumber \\
    & \left( \beta ''^{(j)}\left( 0 \right)+\left( \beta '^{(j)}\left( 0 \right)
    \right)^{2}-\left( \mathbb{E}\left[ I \right] \right)^{2} \right)
\end{align}
\else
\begin{align}
\label{eq:IAM var(I)}
 &   \mathbb{E}\left[ I \right] = -\sum\limits_{j\in \mathcal{K}}
    {\Pr \left(\mathcal{X}_{\mathrm{MT}_{0}}^{(j)} \right)\beta '^{(j)}\left( 0 \right)}
    %
    ; \, \operatorname{var}\left( I \right) =-\sum\limits_{j\in \mathcal{K}}{\Pr \left(
    \mathcal{X}_{\mathrm{MT}_{0}}^{(j)} \right)}
     \left( \beta ''^{(j)}\left( 0 \right)+\left( \beta '^{(j)}\left( 0 \right)
    \right)^{2}-\left( \mathbb{E}\left[ I \right] \right)^{2} \right)
\end{align}
\fi
\noindent where the following definitions hold:
\ifTwoColumns
\begin{align}
  \beta'^{(j)}\left( 0 \right) &= -\sum\limits_{k\in \mathcal{K}}{2\pi
  \lambda ^{(k)}}\sum\limits_{n\in \mathcal{K}}{\Pr \left(
  \mathcal{Q}_{\mathrm{MT}_{i}}^{(n)}|\mathcal{X}_{\mathrm{MT}_{i}}^{(k)}
  ,\mathcal{A}_{\mathrm{MT}_{i}} \right) } \nonumber \\
  & \int\limits_{r=0}^{\infty }{f_{R_{\mathrm{MT}_{i}}}\left( r
  |\mathcal{X}_{\mathrm{MT}_{i}}^{(k,n)},\mathcal{A}_{\mathrm{MT}_{i}} \right)}
  \frac{p_{0}\left( \tau r \right)^{\alpha \epsilon } \tau^{-\alpha} }{\alpha -2} \nonumber \\
  & \mathrm{max} ^{2-\alpha }\left( \left( \frac{t^{(j)}}{t^{(k)}} \right)
  ^{\frac{1}{\alpha }}r,\left( \frac{p_{0}}{i_{0}} \right)^{\frac{1}{\alpha }}
  \frac{\left( \tau r \right)^{\epsilon }}{\tau } \right)\mathrm{d}r
\end{align}
\begin{align}
  \beta''^{(j)}\left( 0 \right) &= -\sum\limits_{k\in \mathcal{K}}{2\pi
  \lambda ^{(k)}}\sum\limits_{n\in \mathcal{K}}{\Pr \left( \mathcal{Q}_{
  \mathrm{MT}_{i}}^{(n)}|\mathcal{X}_{\mathrm{MT}_{i}}^{(k)},\mathcal{A}_{\mathrm{MT}_{i}}
  \right)\times }
  \nonumber \\
  & \int\limits_{r=0}^{\infty }{f_{R_{\mathrm{MT}_{i}}}
  \left( r|\mathcal{X}_{\mathrm{MT}_{i}}^{(k,n)},\mathcal{A}_{
  \mathrm{MT}_{i}} \right)}\frac{p_{0}^{2}\left( \tau r \right)^{2\alpha
  \epsilon } \tau^{-2\alpha} }{1-\alpha}
  \nonumber \\
  & \mathrm{max} ^{2\left( 1-\alpha  \right)}\left( \left( \frac{t^{(j)}}{t^{(k)}}
  \right)^{\frac{1}{\alpha }}r,\left( \frac{p_{0}}{i_{0}} \right)^{\frac{1}{
  \alpha }}\frac{\left( \tau r \right)^{\epsilon }}{\tau } \right)\mathrm{d}r
\end{align}
\else
\begin{align}
 & \beta'^{(j)}\left( 0 \right) = -\sum\limits_{k\in \mathcal{K}}{2\pi
  \lambda ^{(k)}}\sum\limits_{n\in \mathcal{K}}{\Pr \left(
  \mathcal{Q}_{\mathrm{MT}_{i}}^{(n)}|\mathcal{X}_{\mathrm{MT}_{i}}^{(k)}
  ,\mathcal{A}_{\mathrm{MT}_{i}} \right) }
  \nonumber \\
  & \times \int\limits_{0}^{\infty }{f_{R_{\mathrm{MT}_{i}}}\left( r
  |\mathcal{X}_{\mathrm{MT}_{i}}^{(k,n)},\mathcal{A}_{\mathrm{MT}_{i}} \right)}
  \frac{p_{0}\left( \tau r \right)^{\alpha \epsilon } \tau^{-\alpha} }{\alpha -2}
   \mathrm{max} ^{2-\alpha }\left( \left( \frac{t^{(j)}}{t^{(k)}} \right)
  ^{\frac{1}{\alpha }}r,\left( \frac{p_{0}}{i_{0}} \right)^{\frac{1}{\alpha }}
  \frac{\left( \tau r \right)^{\epsilon }}{\tau } \right)\mathrm{d}r \\
 & \beta''^{(j)}\left( 0 \right) = -\sum\limits_{k\in \mathcal{K}}{2\pi
  \lambda ^{(k)}}\sum\limits_{n\in \mathcal{K}}{\Pr \left( \mathcal{Q}_{
  \mathrm{MT}_{i}}^{(n)}|\mathcal{X}_{\mathrm{MT}_{i}}^{(k)},\mathcal{A}_{\mathrm{MT}_{i}}
  \right) }
  \nonumber \\
  & \times \int\limits_{0}^{\infty }{f_{R_{\mathrm{MT}_{i}}}
  \left( r|\mathcal{X}_{\mathrm{MT}_{i}}^{(k,n)},\mathcal{A}_{
  \mathrm{MT}_{i}} \right)}\frac{p_{0}^{2}\left( \tau r \right)^{2\alpha
  \epsilon } \tau^{-2\alpha} }{1-\alpha}
   \mathrm{max} ^{2\left( 1-\alpha  \right)}\left( \left( \frac{t^{(j)}}{t^{(k)}}
  \right)^{\frac{1}{\alpha }}r,\left( \frac{p_{0}}{i_{0}} \right)^{\frac{1}{
  \alpha }}\frac{\left( \tau r \right)^{\epsilon }}{\tau } \right)\mathrm{d}r
\end{align}
\fi

\end{proposition}

\begin{proof}
It directly follows from the first and second derivative of \eqref{eq:IAM LI(s|Xj)} evaluated at $s=0$.
\end{proof}

\begin{remark}[Impact of $i_0$]
\label{rem:IA vs Non IA}
By inspection of \textbf{Propositions \ref{prop:IAM E(P_MT)} {\textnormal{and}} \ref{prop:IAM E(I) var(I)}}, we evince that the average transmit power, the mean and variance of the interference decrease by decreasing $i_0$. Since the interference-unaware setup is obtained by setting $i_0 \to \infty$, this implies that IAM is beneficial in terms of reducing the power consumption of the MTs and of implementing AMC schemes. The system fairness may, however, be negatively affected if $i_0$ decreases, as more MTs are muted.
\end{remark}

The next theorem provides a tractable expression of the coverage probability of HCNs.
\begin{theorem}
\label{prop:IAM F_SINR}
The CCDF of the SINR of the typical MT can be formulated as follows:
\ifTwoColumns
\begin{align}
\label{eq:IAM ccdf SINR}
  \bar{F}_{\mathrm{SINR}} & \left( \gamma  \right)=\sum\limits_{j\in \mathcal{K}}{\sum\limits_{m\in \mathcal{K}}{\Pr \left( \mathcal{X}_{\mathrm{MT}_{0}}^{(j,m)},\mathcal{A}_{\mathrm{MT}_{0}} \right)}} \nonumber \\
 & \int\limits_{v=0}^{\infty }{f_{R_{\mathrm{MT}_{0}}}\left( v|\mathcal{X}_{\mathrm{MT}_{0}}^{(j,m)},\mathcal{A}_{\mathrm{MT}_{0}} \right)} \nonumber \\
 & \mathrm{e}^{-\gamma \sigma _{n}^{2}\left( \tau v \right)^{\alpha \left( 1-\epsilon  \right)}p_{0}^{-1}}\mathcal{L}_{I}\left( \gamma \left( \tau v \right)^{\alpha \left( 1-\epsilon  \right)}p_{0}^{-1}|\mathcal{X}_{\mathrm{MT}_{0}}^{(j)} \right)\mathrm{d}v
\end{align}
\else
\begin{align}
\label{eq:IAM ccdf SINR}
  \bar{F}_{\mathrm{SINR}} & \left( \gamma  \right)=\sum\limits_{j\in \mathcal{K}}{\sum\limits_{m\in \mathcal{K}}{\Pr \left( \mathcal{X}_{\mathrm{MT}_{0}}^{(j,m)},\mathcal{A}_{\mathrm{MT}_{0}} \right)}} \nonumber \\
 & \times \int\nolimits_{0}^{\infty }{f_{R_{\mathrm{MT}_{0}}}\left( v|\mathcal{X}_{\mathrm{MT}_{0}}^{(j,m)},\mathcal{A}_{\mathrm{MT}_{0}} \right)}
  \mathrm{e}^{-\gamma \sigma _{n}^{2}\left( \tau v \right)^{\alpha \left( 1-\epsilon  \right)}p_{0}^{-1}}\mathcal{L}_{I}\left( \gamma \left( \tau v \right)^{\alpha \left( 1-\epsilon  \right)}p_{0}^{-1}|\mathcal{X}_{\mathrm{MT}_{0}}^{(j)} \right)\mathrm{d}v
\end{align}
\fi
\end{theorem}

\begin{proof}
With the aid of the total probability theorem, we have:
\ifTwoColumns
\begin{align}
\label{eq:IAM ccdf SINR Part 1}
& \bar{F}_{\mathrm{SINR}}\left( \gamma  \right)
= \bar{F}_{\mathrm{SINR}}\left( \gamma |\mathcal{A}_{\mathrm{MT}_{0}} \right)\Pr \left( \mathcal{A}_{\mathrm{MT}_{0}} \right)+0\times \Pr \left( \overline{\mathcal{A}_{\mathrm{MT}_{0}}} \right)
\nonumber \\
& \; =
\sum\limits_{j\in \mathcal{K}}{\sum\limits_{m\in \mathcal{K}}{\Pr \left( \mathcal{X}_{\mathrm{MT}_{0}}^{(j,m)},\mathcal{A}_{\mathrm{MT}_{0}} \right)}}
 \bar{F}_{\mathrm{SINR}}\left( \gamma |\mathcal{X}_{\mathrm{MT}_{0}}^{(j,m)},\mathcal{A}_{\mathrm{MT}_{0}} \right)
\nonumber \\
& \;  =\sum\limits_{j\in \mathcal{K}}{\sum\limits_{m\in \mathcal{K}}{\Pr \left( \mathcal{X}_{\mathrm{MT}_{0}}^{(j,m)},\mathcal{A}_{\mathrm{MT}_{0}} \right)}}
\mathbb{E}_{R_{\mathrm{MT}_{0}}} \mathbb{E}_{I} \Bigg[
\nonumber \\ & \;
\Pr \left( H_{\mathrm{MT}_{0}}>\frac{\gamma}{p_0} \left( I+\sigma _{n}^{2} \right)\left( \tau R_{\mathrm{MT}_{0}} \right)^{\alpha \left( 1-\epsilon
\right)}
  |\mathcal{X}_{\mathrm{MT}_{0}}^{(j,m)},\mathcal{A}_{\mathrm{MT}_{0}} \right) \Bigg]
\end{align}
\else
\begin{align}
\label{eq:IAM ccdf SINR Part 1}
\bar{F}_{\mathrm{SINR}}&\left( \gamma  \right)
= \bar{F}_{\mathrm{SINR}}\left( \gamma |\mathcal{A}_{\mathrm{MT}_{0}} \right)\Pr \left( \mathcal{A}_{\mathrm{MT}_{0}} \right)+0\times \Pr \left( \overline{\mathcal{A}_{\mathrm{MT}_{0}}} \right)
\nonumber \\
& =
\sum\limits_{j\in \mathcal{K}}{\sum\limits_{m\in \mathcal{K}}{\Pr \left( \mathcal{X}_{\mathrm{MT}_{0}}^{(j,m)},\mathcal{A}_{\mathrm{MT}_{0}} \right)}}
 \bar{F}_{\mathrm{SINR}}\left( \gamma |\mathcal{X}_{\mathrm{MT}_{0}}^{(j,m)},\mathcal{A}_{\mathrm{MT}_{0}} \right) =\sum\limits_{j\in \mathcal{K}}{\sum\limits_{m\in \mathcal{K}}{\Pr \left( \mathcal{X}_{\mathrm{MT}_{0}}^{(j,m)},\mathcal{A}_{\mathrm{MT}_{0}} \right)}}
\nonumber \\
& \times \mathbb{E}_{R_{\mathrm{MT}_{0}}}
 \mathbb{E}_{I}\left[ \Pr \left( H_{\mathrm{MT}_{0}}>\frac{\gamma}{p_0} \left( I+\sigma _{n}^{2} \right)\left( \tau R_{\mathrm{MT}_{0}} \right)^{\alpha \left( 1-\epsilon
\right)}
  |\mathcal{X}_{\mathrm{MT}_{0}}^{(j,m)},\mathcal{A}_{\mathrm{MT}_{0}} \right) \right]
\end{align}
\fi

\noindent The proof follows by computing the two remaining expectations.
\end{proof}


\begin{corollary}
Assume $\epsilon = 1$, i.e., the active MTs apply a power control scheme based on full channel inversion. The CCDF in \textbf{Theorem \ref{prop:IAM F_SINR}} simplifies as follows:
\label{cor:F_SINR}
\ifTwoColumns
\begin{align}
\label{eq:IAM ccdf SINR epsilon=1}
\bar{F}_{\mathrm{SINR}}\left( \gamma  \right) &=\sum\limits_{j\in \mathcal{K}}{\sum\limits_{m\in \mathcal{K}}{\Pr \left( \mathcal{X}_{\mathrm{MT}_{0}}^{(j,m)},\mathcal{A}_{\mathrm{MT}_{0}} \right)}} \nonumber \\
 & \mathrm{e}^{-\gamma \sigma _{n}^{2}/p_{0}}\mathcal{L}_{I}\left( \gamma /p_{0}|\mathcal{X}_{\mathrm{MT}_{0}}^{(j)} \right)
\end{align}
\else
\begin{align}
\label{eq:IAM ccdf SINR epsilon=1}
\bar{F}_{\mathrm{SINR}}\left( \gamma  \right) =\sum\limits_{j\in \mathcal{K}}{\sum\limits_{m\in \mathcal{K}}{\Pr \left( \mathcal{X}_{\mathrm{MT}_{0}}^{(j,m)},\mathcal{A}_{\mathrm{MT}_{0}} \right)}}
  \mathrm{e}^{-\gamma \sigma _{n}^{2}/p_{0}}\mathcal{L}_{I}\left( \gamma /p_{0}|\mathcal{X}_{\mathrm{MT}_{0}}^{(j)} \right)
\end{align}
\fi
\end{corollary}

\begin{proof}
It follows from (\ref{eq:IAM ccdf SINR}) by setting $\epsilon=1$ and some algebra.
\end{proof}

\begin{remark}[Operating regimes as a function of $i_0$]
\label{rem:Operation regimes of IAM under GA}
By direct inspection of \textbf{Corollary \ref{cor:F_SINR}}, three operating regimes as a function of $i_0$ can be identified: i) interference-unaware, where the CCDF of the SINR is independent of $i_0$. This occurs if $i_0 > p_0$ and $p_0/i_0 < \min \left(t^{1}/t^{(2)}, t^{(2)}/t^{(1)} \right)$, ii) interference-aware and cell association independent, where the CCDF of the SINR depends on $i_0$ but does not depend on the cell association weights $t^{(1)}$ and $t^{(2)}$. This occurs if $i_0 < p_0$ and $p_0/i_0 > \max \left(t^{(1)}/t^{(2)}, t^{(2)}/t^{(1)} \right)$, iii) interference-aware and cell association dependent, where the CCDF of the SINR depends on $i_0$ and $t^{(\tilde j)}/t^{(j)}, \, \forall j\in \mathcal{K}$. This occurs if the conditions above are not satisfied. The same operating regimes can be identified from \textbf{Propositions \ref{prop:IAM f_RMT_0 cond XjmA}} and \textbf{\ref{prop:IAM E(P_MT)}}.
\end{remark}

\begin{proof}
It follows by direct inspection of $\Pr \left( {{\cal X}_{{\rm{MT}}_{\rm{0}} }^{(j,m)} ,{\cal A}_{{\rm{MT}}_{\rm{0}} } } \right)$, $\nu^{(j)}(v)$ and $\eta^{(j)}(v)$.
\end{proof}

The second operating regime, i.e., the performance is independent of the cell association weights, is of particular interest for making the design of HCNs easier: it implies that, for some system parameters, optimizing the DL results in optimizing the UL as well.

It is worth mentioning, in addition, that the conditions that identify the three operating regimes in \textbf{Remark \ref{rem:Operation regimes of IAM under GA}} can be conveniently formulated in dB as well, which provides further information for system design. More precisely, regime i) emerges if $i_0 > p_0$ dB and $t^{(1)}/t^{(2)} \in [-i_0/p_0,i_0/p_0]$ dB and regime ii) emerges if $i_0 < p_0$ dB and $t^{(1)}/t^{(2)} \in [-p_0/i_0,p_0/i_0]$ dB.

\section{Smallest Path-Loss Association}
\label{sec:Smallest Path Loss Association}
In this section, tractable mathematical frameworks under the SPLA scheme are provided. In this case, the condition $t^{(1)}=t^{(2)}$ holds and simplified formulas can be obtained. Under the assumption that the path-loss exponents of all the tiers of BSs are the same, in fact, multi-tier HCNs reduce to an equivalent single-tier cellular network of intensity $\lambda  = \sum_{j \in {\cal K}} {{\lambda ^{\left( j \right)}}}$ \cite{Singh15}.

\begin{proposition}
\label{cor:SPLA PrA}
The probability that the typical MT is active can be formulated as follows:
\begin{equation}
\label{eq:SPLA PrA}
\Pr \left( {{{\cal A}_{{\rm{M}}{{\rm{T}}_{\rm{0}}}}}} \right) = \int\nolimits_{{r_1} = 0}^{{1 \over \tau }{{\left( {{{{p_{\max }}} \over {{p_0}}}} \right)}^{{1 \over \alpha }}}} {2\pi \lambda {r_1}} {e^{ - \pi \lambda {{\max }^2}\left( {{r_1},{{\left( {{{{p_0}} \over {{i_0}}}} \right)}^{{1 \over \alpha }}}{{{{\left( {\tau {r_1}} \right)}^\epsilon}} \over \tau }} \right)}}d{r_1}
\end{equation}
\end{proposition}
\begin{proof}
The proof is similar to that of \textbf{Proposition \ref{prop:IAM Xjm}}. The difference is that only the joint PDF of the distance of nearest and second nearest BSs needs to be used (see Appendix A).
\end{proof}

\begin{corollary}
If $\epsilon=1$, $\Pr \left( {{{\cal A}_{{\rm{M}}{{\rm{T}}_{\rm{0}}}}}} \right)$ in \eqref{eq:SPLA PrA} simplifies as follows:

\ifTwoColumns
\begin{equation}
\label{eq:SPLA PrA epsilon1}
\Pr \left( {{{\cal A}_{{\rm{M}}{{\rm{T}}_{\rm{0}}}}}} \right) = {{1 - {{\rm{e}}^{ - {\pi  \over {{\tau ^2}}}{{\left( {{{{p_{\max }}} \over {{p_0}}}} \right)}^{{2 \over \alpha }}}\lambda \max \left( {1,{{\left( {{{{p_0}} \over {{i_0}}}} \right)}^{{2 \over \alpha }}}} \right)}}} \over {\max \left( {1,{{\left( {{{{p_0}} \over {{i_0}}}} \right)}^{{2 \over \alpha }}}} \right)}}
\end{equation}
\else
\begin{equation}
\label{eq:SPLA PrA epsilon1}
\Pr \left( {{{\cal A}_{{\rm{M}}{{\rm{T}}_{\rm{0}}}}}} \right) = 
{\left(1 - {{\rm{e}}^{ - {\pi  \over {{\tau ^2}}}{{\left( {{{{p_{\max }}} \over {{p_0}}}} \right)}^{{2 \over \alpha }}}\lambda \max \left( {1,{{\left( {{{{p_0}} 
\over {{i_0}}}} \right)}^{{2 \over \alpha }}}} \right)}}\right)} 
\Bigg/{\max \left( {1,{{\left( {{{{p_0}} 
\over {{i_0}}}} \right)}^{{2 \over \alpha }}}} \right)}
\end{equation}
\fi
\end{corollary}
\begin{proof}
It directly follows from \eqref{eq:SPLA PrA} by setting $\epsilon=1$ and computing the integral.
\end{proof}
\begin{remark}[{Operating regimes as a function of $i_0$}]
\label{rem: Interference unaware regime}
From (\ref{eq:SPLA PrA epsilon1}), two operating regimes can be identified: i) interference-unaware, i.e., $\Pr \left( {{{\cal A}_{{\rm{M}}{{\rm{T}}_{\rm{0}}}}}} \right)$ is independent of $i_0$, which occurs if $i_0 > p_0$ and ii) interference-aware, i.e., $\Pr \left( {{{\cal A}_{{\rm{M}}{{\rm{T}}_{\rm{0}}}}}} \right)$ depends on $i_0$, which occurs if $i_0 < p_0$.
\end{remark}

\begin{remark}[{Unlimited transmit power of the MTs}]
\label{rem: power law PrA}
Assume $p_\mathrm{max} \to \infty$, i.e., the MTs have no maximum transmit power constraint. From (\ref{eq:SPLA PrA epsilon1}), the following holds: i) under the interference-unaware regime ($i_0 > p_0$), $\Pr \left( {{{\cal A}_{{\rm{M}}{{\rm{T}}_{\rm{0}}}}}} \right) \to 1$, and ii) under the interference-aware regime ($i_0 < p_0$), $\Pr \left( {{{\cal A}_{{\rm{M}}{{\rm{T}}_{\rm{0}}}}}} \right) = {\left( {{{{i_0}} / {{p_0}}}} \right)^{{2 \over \alpha }}}$. In both regimes, $\Pr \left( {{{\cal A}_{{\rm{M}}{{\rm{T}}_{\rm{0}}}}}} \right)$ is independent of the density of BSs $\lambda$.
\end{remark}

\begin{lemma}
\label{cor: SPLA pdf dist serving}
The PDF of the distance between the typical MT and its serving BS is as follows:
\ifTwoColumns
\begin{align}
\label{eq: SPLA pdf dist serving}
& {f_{{R_{{\rm{M}}{{\rm{T}}_{\rm{0}}}}}}}\left( {v|{{\cal A}_{{\rm{M}}{{\rm{T}}_{\rm{0}}}}}} \right) = {{2\pi \lambda v{{\rm{e}}^{ - \pi \lambda {{\max }^2}\left( {v,{{\left( {{{{p_0}} \over {{i_0}}}} \right)}^{{1 \over \alpha }}}{{{{\left( {\tau v} \right)}^\epsilon }} \over \tau }} \right)}}} \over {\Pr \left( {{{\cal A}_{{\rm{M}}{{\rm{T}}_{\rm{0}}}}}} \right)}}
\nonumber \\ & \quad \times
{\bf 1} \left( 0< v < {1 \over \tau }{\left( {{{{p_{\max }}} \over {{p_0}}}} \right)^{{1 \over \alpha }}} \right)
\end{align}
\else
\begin{align}
\label{eq: SPLA pdf dist serving}
{f_{{R_{{\rm{M}}{{\rm{T}}_{\rm{0}}}}}}}\left( {v|{{\cal A}_{{\rm{M}}{{\rm{T}}_{\rm{0}}}}}} \right) = {{2\pi \lambda v{{\rm{e}}^{ - \pi \lambda {{\max }^2}\left( {v,{{\left( {{{{p_0}} \over {{i_0}}}} \right)}^{{1 \over \alpha }}}{{{{\left( {\tau v} \right)}^\epsilon }} \over \tau }} \right)}}} \over {\Pr \left( {{{\cal A}_{{\rm{M}}{{\rm{T}}_{\rm{0}}}}}} \right)}}
{\bf 1} \left( 0< v < {1 \over \tau }{\left( {{{{p_{\max }}} \over {{p_0}}}} \right)^{{1 \over \alpha }}} \right)
\end{align}
\fi
\end{lemma}

\begin{proof}
The proof is similar to that of \textbf{Lemma \ref{prop:IAM f_RMT_0 cond XjmA}}. The difference is that only the joint PDF of the distance of nearest and second nearest BSs needs to be used (see Appendix A).
\end{proof}

\begin{remark}[{Interference-awareness is equivalent to network densification if $p_{\max} \to \infty$}]
\label{rem: Equivalent BS density}
If the system operates in the interference-aware regime ($i_0 < p_0$) and $p_{\max} \to \infty$, (\ref{eq: SPLA pdf dist serving}) reduces to:
\begin{equation}
\label{eq: SPLA pdf dist serving IA case}
{f_{{R_{{\rm{M}}{{\rm{T}}_{\rm{0}}}}}}}\left( {v|{{\cal A}_{{\rm{M}}{{\rm{T}}_{\rm{0}}}}}} \right) = 2\pi \lambda {\left( {{{{p_0}} \over {{i_0}}}} \right)^{{2 \over \alpha }}}v{{\rm{e}}^{ - \pi \lambda {{\left( {{{{p_0}} \over {{i_0}}}} \right)}^{{2 \over \alpha }}}{v^2}}}
\end{equation}

\noindent This implies that IAM's impact is equivalent to increasing the density of BSs from $\lambda$ to $ \lambda {\left( {{p_0}/{i_0}} \right)^{{2 \over \alpha }}}$, since the PDF of the distance from the nearest BS in Poisson cellular networks is $2\pi \lambda v{{\rm{e}}^{ - \pi \lambda {v^2}}}$. Hence, the distance between probe MT and probe BS is reduced, resulting in better performance.
\end{remark}

\begin{proposition}
\label{cor:SPLA avP}
If $\epsilon=1$, the average transmit power of the typical MT is as follows:
\ifTwoColumns
\begin{align}
\label{eq:SPLA avP}
& \mathbb{E}\left[ {p\left( {{R_{{\rm{M}}{{\rm{T}}_{\rm{0}}}}}} \right)} \right] = {{{p_0}{\tau ^\alpha }} \over {{{\left( {\pi \lambda } \right)}^{{\alpha  \over 2}}}\max \left( {1,{{\left( {{{{p_0}} \over {{i_0}}}} \right)}^{{2 \over \alpha } + 1}}} \right)}}
\Bigg( \Gamma \left( {1 + {\alpha  \over 2}} \right)
\nonumber \\ & \quad
- \Gamma \left( {{{2 + \alpha } \over 2},{\lambda \pi  \over {{\tau ^2}}}{{\left( {{{{p_{\max }}} \over {{p_0}}}} \right)}^{{2 \over \alpha }}} \max \left( {1,{{\left( {{{{p_0}} \over {{i_0}}}} \right)}^{{2 \over \alpha }}}} \right)} \right) \Bigg)
\end{align}
\else
\begin{align}
\label{eq:SPLA avP}
& \mathbb{E}\left[ {p_{{{\rm{M}}{{\rm{T}}}}}\left( {{R_{{\rm{M}}{{\rm{T}}_{\rm{0}}}}}} \right)} \right] 
  = {{{p_0}{\tau ^\alpha \Bigg( \Gamma \left( {1 + {\alpha  \over 2}} \right)
- \Gamma \left( {{{2 + \alpha } \over 2},{\lambda \pi  \over {{\tau ^2}}}{{\left( {{{{p_{\max }}} \over {{p_0}}}} \right)}^{{2 \over \alpha }}} \max \left( {1,{{\left( {{{{p_0}} \over {{i_0}}}} \right)}^{{2 \over \alpha }}}} \right)} \right) \Bigg) }} \over {{{\left( {\pi \lambda } \right)}^{{\alpha  \over 2}}}\max \left( {1,{{\left( {{{{p_0}} \over {{i_0}}}} \right)}^{{2 \over \alpha } + 1}}} \right)}} 
\end{align}
\fi
\end{proposition}

\begin{proof}
If follows from \textbf{Proposition \ref{prop:IAM E(P_MT)}}, by setting $\epsilon=1$ and computing the integral.
\end{proof}

\begin{remark}[Impact of interference-awareness]
\label{rem: power law for the average power}
If $p_\mathrm{max} \to \infty$ and $i_0 < p_0$ (interference-aware regime), (\ref{eq:SPLA avP}) simplifies as follows:
\begin{equation}
\label{eq:SPLA avP IA regime}
\mathbb{E}\left[ {p_{{{\rm{M}}{{\rm{T}}}}}\left( {{R_{{\rm{M}}{{\rm{T}}_{\rm{0}}}}}} \right)} \right] = {{{\tau ^\alpha }\Gamma \left( {1 + {\alpha  \over 2}} \right)} \over {{{\left( {\pi \lambda } \right)}^{{\alpha  \over 2}}}{p_0}^{{2 \over \alpha }}}}{i_0}^{^{{2 \over \alpha } + 1}}
\end{equation}
\noindent which implies that the average power consumption of the MTs scales polynomially with exponent $2/\alpha+1$, as a function of the maximum interference constraint $i_0$.
\end{remark}

\begin{lemma}
\label{lem: SPLA LI 1}
Assume $\epsilon=1$. The Laplace transform of the aggregate interference can be formulated as ${{\cal L}_I}\left( s \right) = \exp \left( {\beta (s)} \right)$, where $\beta (s) =  - 2\pi \lambda \theta \mu \left( s \right)$ and the following holds:
\ifTwoColumns
\begin{align}
\label{eq: SPLA theta}
& \theta  = \Bigg( 1 - \left( {1 + {\pi  \over {{\tau ^2}}}\lambda {{\left( {{{{p_{\max }}} \over {{p_0}}}} \right)}^{{2 \over \alpha }}}} \right)
\nonumber \\ & \; \times
{e^{ - {\pi  \over {{\tau ^2}}}\lambda {{\left( {{{{p_{\max }}} \over {{p_0}}}} \right)}^{{2 \over \alpha }}}\max \left( {1,{{\left( {{{{p_0}} \over {{i_0}}}} \right)}^{{2 \over \alpha }}}} \right)}} \Bigg)
{\left( {\pi \lambda \max \left( {1,{{\left( {{{{p_0}} \over {{i_0}}}} \right)}^{{2 \over \alpha }}}} \right)} \right)^{ - 1}}
\end{align}
\else
\begin{align}
\label{eq: SPLA theta}
\hspace{-0.5cm} \theta  = \Bigg( 1 - \left( {1 + {\pi  \over {{\tau ^2}}}\lambda {{\left( {{{{p_{\max }}} \over {{p_0}}}} \right)}^{{2 \over \alpha }}}} \right)
{e^{ - {\pi  \over {{\tau ^2}}}\lambda {{\left( {{{{p_{\max }}} \over {{p_0}}}} \right)}^{{2 \over \alpha }}}\max \left( {1,{{\left( {{{{p_0}} \over {{i_0}}}} \right)}^{{2 \over \alpha }}}} \right)}} \Bigg)
{\left( {\pi \lambda \max \left( {1,{{\left( {{{{p_0}} \over {{i_0}}}} \right)}^{{2 \over \alpha }}}} \right)} \right)^{ - 1}}
\end{align}
\fi
\ifTwoColumns
\begin{align}
\label{eq: SPLA mu}
& \mu \left( s \right) = {{{p_0}s} \over {\alpha  - 2}}{\rm{ma}}{{\rm{x}}^{2 - \alpha }}\left( {1,{{\left( {{{{p_0}} \over {{i_0}}}} \right)}^{{1 \over \alpha }}}} \right)
\nonumber \\ & \; \times
{_2}{{\rm{F}}_1}\left( {1,{{\alpha  - 2} \over \alpha },2 - {2 \over \alpha }, - {p_0}s} \right.\left. {{\rm{ma}}{{\rm{x}}^{ - \alpha }}\left( {1,{{\left( {{{{p_0}} \over {{i_0}}}} \right)}^{{1 \over \alpha }}}} \right)} \right)
\end{align}
\else
\begin{equation}
\label{eq: SPLA mu}
\hspace{-0.00cm} \mu \left( s \right) = {{{p_0}s} \over {\alpha  - 2}}{\rm{ma}}{{\rm{x}}^{2 - \alpha }}\left( {1,{{\left( {{{{p_0}} \over {{i_0}}}} \right)}^{{1 \over \alpha }}}} \right){_2}{{\rm{F}}_1}\left( {1,{{\alpha  - 2} \over \alpha },2 - {2 \over \alpha }, - {p_0}s} \right.\left. {{\rm{ma}}{{\rm{x}}^{ - \alpha }}\left( {1,{{\left( {{{{p_0}} \over {{i_0}}}} \right)}^{{1 \over \alpha }}}} \right)} \right)
\end{equation}
\fi
\end{lemma}

\begin{proof}
The proof follows from $\chi (s,r)$ in (\ref{eq:chi(s,r)}), by setting  $t^{(1)}=t^{(2)}$ and formulating it as $\chi (s,r)=r^2 \mu (s)$. Hence, $\beta (s) =  - 2\pi \lambda \mu \left( s \right)\theta $, where $\theta  = \mathbb{E}\left[ {R_{{\rm{M}}{{\rm{T}}_{\rm{i}}}}^2|{A_{{\rm{M}}{{\rm{T}}_{\rm{0}}}}}} \right]$.
\end{proof}

\begin{proposition}
\label{cor:avI varI}
Assume $\epsilon=1$. The mean and variance of the interference can be expressed as:
\ifTwoColumns
\begin{align}
\label{eq: SPLA avI varI}
\mathbb{E}\left[ I \right] = 2\pi \lambda \theta {{{p_0}{\rm{ma}}{{\rm{x}}^{2 - \alpha }}\left( {1,{{\left( {{{{p_0}} \over {{i_0}}}} \right)}^{{1 \over \alpha }}}} \right)} \over {\alpha  - 2}}
\nonumber \\
{\mathop{\rm var}} \left( I \right) = 2\pi \lambda \theta {{p_0^2{\rm{ma}}{{\rm{x}}^{2 - 2\alpha }}\left( {1,{{\left( {{{{p_0}} \over {{i_0}}}} \right)}^{{1 \over \alpha }}}} \right)} \over {\alpha  - 1}}
\end{align}
\else
\begin{align}
\label{eq: SPLA avI varI}
\mathbb{E}\left[ I \right] = 2\pi \lambda \theta {{{p_0}{\rm{ma}}{{\rm{x}}^{2 - \alpha }}\left( {1,{{\left( {{{{p_0}} \over {{i_0}}}} \right)}^{{1 \over \alpha }}}} \right)} \over {\alpha  - 2}}
; \quad {\mathop{\rm var}} \left( I \right) = 2\pi \lambda \theta {{p_0^2{\rm{ma}}{{\rm{x}}^{2 - 2\alpha }}\left( {1,{{\left( {{{{p_0}} \over {{i_0}}}} \right)}^{{1 \over \alpha }}}} \right)} \over {\alpha  - 1}}
\end{align}
\fi
\end{proposition}

\begin{proof}
It follows from {\textbf{Lemma \ref{lem: SPLA LI 1}}} evaluating the derivatives of the Laplace transform at zero.
\end{proof}

\begin{remark}[Trends of mean and variance of the interference as a function of $i_0$]
\label{rem: power law for the avI}
Assume ${p_{\max }} \to \infty $ and consider the interference-aware regime, i.e., $i_0 < p_0$. Then, \eqref{eq: SPLA theta} simplifies to $\theta  = {1 \over {\pi \lambda }}{\left( {{{{i_0}} \over {{p_0}}}} \right)^{{4 \over \alpha }}}$ and the mean and variance of the interference can be formulated as follows:
\ifTwoColumns
\begin{align}
\label{eq: SPLA avI varI pmax inf}
\mathbb{E}\left[ I \right] = {2 \over {\alpha  - 2}}{ p_0^{ - {2 \over \alpha }}}{i_0^{{{\alpha  + 2} \over \alpha }}}
; \, 
{\mathop{\rm var}} \left( I \right) = {2 \over {\alpha  - 1}}
p_0^{ - {2 \over \alpha }}{i_0^{{{2\left( {\alpha  + 1} \right)} \over \alpha }}}
\end{align}

\else
\begin{align}
\label{eq: SPLA avI varI pmax inf}
\mathbb{E}\left[ I \right] = {2 \over {\alpha  - 2}}{\left( {{p_0}} \right)^{ - {2 \over \alpha }}}{\left( {{i_0}} \right)^{{{\alpha  + 2} \over \alpha }}}
; \quad \quad
{\mathop{\rm var}} \left( I \right) = {2 \over {\alpha  - 1}}{\left( {{p_0}} \right)^{ - {2 \over \alpha }}}{\left( {{i_0}} \right)^{{{2\left( {\alpha  + 1} \right)} \over \alpha }}}
\end{align}

\fi

\noindent which implies that the mean and variance of the interference scale polynomially with exponents $\alpha+2/\alpha$ and $2\left( {\alpha  + 1} \right)/\alpha$ as a function of $i_0$, respectively, {\color{black} and they do not depend on the BSs' density}.
\end{remark}

Finally, the following theorem provides the coverage probability under the SPLA criterion.
\begin{theorem}
\label{cor: SPLA ccdf of the SINR}
Assume $\epsilon=1$, $p_{\rm max} \to \infty$ and that the system operates in the interference-aware regime ($i_0 < p_0$). The CCDF of the SINR can be formulated as follows:
\ifTwoColumns
\begin{align}
\label{eq: SPLA ccdf of the SINR}
& {\bar F_{{\rm{SINR}}}}\left( {\gamma |{{\cal A}_{{\rm{M}}{{\rm{T}}_0}}}} \right) = \exp \Bigg(  - {{\gamma \sigma _n^2} \over {{p_0}}} - 2{\gamma  \over {\alpha  - 2}}{{\left( {{{{i_0}} \over {{p_0}}}} \right)}^{{{\alpha  + 2} \over \alpha }}}
\nonumber \\ & \; \times
{_2}{{\rm{F}}_1}\left( {1,{{\alpha  - 2} \over \alpha },2 - {2 \over \alpha }, - \gamma \left( {{{{i_0}} \over {{p_0}}}} \right)} \right) \Bigg)
\end{align}
\else
\begin{align}
\label{eq: SPLA ccdf of the SINR}
{\bar F_{{\rm{SINR}}}}\left( {\gamma |{{\cal A}_{{\rm{M}}{{\rm{T}}_0}}}} \right) = \exp \left( { - {{\gamma \sigma _n^2} \over {{p_0}}} - 2{\gamma  \over {\alpha  - 2}}{{\left( {{{{i_0}} \over {{p_0}}}} \right)}^{{{\alpha  + 2} \over \alpha }}}{_2}{{\rm{F}}_1}\left( {1,{{\alpha  - 2} \over \alpha },2 - {2 \over \alpha }, - \gamma \left( {{{{i_0}} \over {{p_0}}}} \right)} \right)} \right)
\end{align}
\fi
\end{theorem}

\begin{proof}
The proof follows from \textbf{Theorem \ref{prop:IAM F_SINR}} by setting $t^{(1)}=t^{(2)}$ and $\epsilon=1$, and from \textbf{Lemma \ref{lem: SPLA LI 1}} by letting $p_{\rm{max}} \to \infty$ and considering $i_0<p_0$.
\end{proof}

\begin{remark}[SINR invariance as a function of $\lambda$]
\label{rem:Scaling invariance}
From (\ref{eq: SPLA ccdf of the SINR}), we evince that the CCDF of the SINR is independent of $\lambda$, but it depends on the ratio $i_0/p_0$ and the path-loss exponent $\alpha$.
\end{remark}

{\color{black} Interestingly the SINR in such a setup is invariant with the BSs' density. Intuitively, this means that both the desired received power and the interference does not vary with the BSs' density. On the one hand, the desired power does not vary thanks to full channel inversion power control ($\epsilon=1$, $p_{\rm max}\to \infty$). On the other hand, although the distances towards nearest interfering MTs decrease with $\lambda$, their transmit power also decrease with $\lambda$, making received interference invariant with $\lambda$, as it can be observed from its moments in eq. (\ref{eq: SPLA avI varI pmax inf}). This density invariance has been also reported in \cite{Andrews11, Singh15} for the case of the SIR.}

\section{Spectral Efficiency and Binary Rate}
\label{sec:Binary Rate and Spectral Efficiency}
This section is focused on the analysis of SE and BR. Unlike the vast majority of papers on stochastic geometry modeling of HCNs that evaluate these key performance indicators based on the Shannon formula, we provide a mathematical formulation that is more useful for current cellular deployments based on practical AMC schemes and, thus, provides estimates of SE and BR that can be achieved at a finite target value of the Block Error Rate (BLER) rather than their theoretically achievable counterparts under the assumptions of unlimited decoding complexity and arbitrarily small BLER. We show, remarkably, that more tractable expressions of SE and BR can be provided, compared to those that can be obtained based on the Shannon definition. As mentioned in Section I, the BR accounts for the amount of bandwidth allocated to the typical MT by the scheduler and, thus, accounts for the BS's load, i.e., the number of MTs that need to be simultaneously served in the cell to which the typical MT belongs to. Accordingly, SE and BR provide different information on the advantages and limitations of transmission schemes and, as such, are both employed for assessing the performance of practical LTE systems \cite{Sesia09}.

SE and BR, however, are related to each other and, in mathematical terms, we have:
\begin{equation}
\label{eq:Shannon and AMC BR}
\mathrm{BR}_{\mathrm{MT}_{0}}=\frac{b_{w}}{N_{\mathcal{B}_{\mathrm{MT}_{0}}}^{\mathcal{A}}}\mathrm{SE}_{\mathrm{MT}_{0}}
\quad \left( \mathrm{bps} \right)
\end{equation}
where $b_w$ is the available bandwidth per BS and $N^{\mathcal{A}}_{\mathcal{B}_{\mathrm{MT}_{0}}}$ denotes the number of active MTs associated with the probe BS, which is commonly referred to as the cell load \cite{DiRenzo16a}.

As extensively discussed in, e.g., \cite{DiRenzo16a}, \cite{Singh14b} \cite{Ferenc07}, the distribution of $N^{\mathcal{A}}_{\mathcal{B}_{\mathrm{MT}_{0}}}$ is not available for cell association criteria that are not based on the shortest distance, and, thus approximations need to be used. For mathematical tractability, but without loosing in accuracy, we exploit the approximation in \cite{Singh14b} which, for the convenience of the readers, is reported in what follows.
\begin{assumption}
The Probability Mass Function (PMF) of the number of active MTs, $N_{\mathcal{B}_{\mathrm{MT}_{0}}}^{\mathcal{A}}$, associated with a BS of tier $j$ is approximated as follows:
\ifTwoColumns
\begin{align}
\label{eq:n^A_BMT0}
\Pr & \left( N_{\mathcal{B}_{\mathrm{MT}_{0}}}^{\mathcal{A}}=n|\mathcal{X}_{\mathrm{MT}_{0}}^{(j)},\mathcal{A}_{\mathrm{MT}_{0}} \right)
 \approx \frac{3.5^{3.5}}{\left( n-1 \right)!}\frac{\Gamma \left( n+3.5 \right)}{\Gamma \left( 3.5 \right)}
\nonumber \\
& \quad \left( \frac{\lambda _{\mathrm{MT}}\Pr \left( \mathcal{X}_{\mathrm{MT}_{0}}^{(j)},\mathcal{A}_{\mathrm{MT}_{0}} \right)}{\lambda ^{(j)}} \right)^{n-1}
\nonumber \\
& \quad  \left( 3.5+\frac{\lambda _{\mathrm{MT}}\Pr \left( \mathcal{X}_{\mathrm{MT}_{0}}^{(j)},\mathcal{A}_{\mathrm{MT}_{0}} \right)}{\lambda ^{(j)}} \right)
\end{align}
\else
\begin{align}
\label{eq:n^A_BMT0}
\Pr \left( N_{\mathcal{B}_{\mathrm{MT}_{0}}}^{\mathcal{A}}=n|\mathcal{X}_{\mathrm{MT}_{0}}^{(j)},\mathcal{A}_{\mathrm{MT}_{0}} \right)
& \approx \frac{3.5^{3.5}}{\left( n-1 \right)!}\frac{\Gamma \left( n+3.5 \right)}{\Gamma \left( 3.5 \right)}
\left( \frac{\lambda _{\mathrm{MT}} \cdot  p}{\lambda ^{(j)}} \right)^{n-1}
\left( 3.5+\frac{\lambda _{\mathrm{MT}} \cdot p}{\lambda ^{(j)}} \right)
\end{align}
\fi
\end{assumption}
\noindent where, for notational simplicity, the short-hand $p=\Pr(\mathcal{X}^{(j)}_{\mathrm{MT}_0},\mathcal{A}_{\mathrm{MT}_{0}})$ is used.

\subsection{Adaptive Modulation and Coding}
In modern cellular systems \cite{Sesia09}, AMC is aimed to adapt the MCS to be used to the channel conditions. This is needed for maximizing the BR while providing a BLER below a desired threshold BLER$_T$. In practice, AMC is implemented as follows. In the UL, the MTs transmit sounding reference signals that are used by the BSs for estimating the SINR. Based on these estimates, the BSs choose the MCS to use (usually identified by an index), which corresponds to a given Channel Quality Indicator (CQI), $i_\mathrm{CQI}\in [1,n_\mathrm{CQI}]$, that maximizes the SE while maintaining the BLER below BLER$_T$. The choice of the best MCS to use is made based on lookup tables that provide the SINR thresholds, $\gamma_{i_\mathrm{CQI}}$, associated to each value of CQI. Finally, the BSs inform each scheduled MT of the MCS index to use for its subsequent transmission. To reduce the reporting overhead associated with the CQIs, the LTE standard assumes that the number of bits used for reporting the CQI is equal to $4$, which implies $n_\mathrm{CQI}=15$.

Based on this working principle, the BR can be obtained from (\ref{eq:Shannon and AMC BR}) and the SE is as follows:
\begin{equation}
\label{eq:AMC SE}
\mathrm{SE}_{\mathrm{MT}_{0}}=\sum\limits_{i_{\mathrm{CQI}}=1}^{n_{\mathrm{CQI}}}{\mathrm{SE}_{i_{\mathrm{CQI}}}}\mathbf{1}\left( \mathrm{SINR}_{\mathrm{MT}_{0}}\in [\gamma _{i_{\mathrm{CQI}}},\gamma _{i_{\mathrm{CQI}+1}}) \right)
\end{equation}
where $\gamma _{1}<\cdots <\gamma _{n_{\text{CQI}}}$, $i_\mathrm{CQI}=0$ if no transmission,
$\bigcap\nolimits_{i_{\mathrm{CQI}}=1}^{n_{\mathrm{CQI}}}{[\gamma _{i_{\mathrm{CQI}}},\gamma _{i_{\mathrm{CQI}+1}})=\emptyset }$, $\gamma _{n_{\mathrm{CQI}+1}} \rightarrow \infty$.

Based on \eqref{eq:AMC SE}, the spatially-average SE can be obtained from the CCDF of the SINR provided in Sections III and IV for GCA and SPLA criteria, respectively. More precisely, we have:
\ifTwoColumns
\begin{align}
\label{eq:AMC E(SE)}
& \mathbb{E}\left[ \mathrm{SE_{\mathrm{MT}_{0}}} \right]=\sum\limits_{j\in \mathcal{K}}{\sum\limits_{m\in \mathcal{K}}{\Pr \left( \mathcal{X}_{\mathrm{MT}_{0}}^{(j,m)},\mathcal{A}_{\mathrm{MT}_{0}} \right)}}
\nonumber \\
& \quad  \sum\limits_{i_{\mathrm{CQI}}=1}^{n_{\mathrm{CQI}}}{\mathrm{SE}_{i_{\mathrm{CQI}}}}
\Bigg( \bar{F}_{\mathrm{SINR}}\left( \gamma _{i_{\mathrm{CQI}}}|\mathcal{X}_{\mathrm{MT}_{0}}^{(j,m)},\mathcal{A}_{\mathrm{MT}_{0}} \right)
\nonumber \\
& \quad -\bar{F}_{\mathrm{SINR}}\left( \gamma _{i_{\mathrm{CQI}}+1}|\mathcal{X}_{\mathrm{MT}_{0}}^{(j,m)},\mathcal{A}_{\mathrm{MT}_{0}} \right) \Bigg)
\end{align}
\else
\begin{align}
\label{eq:AMC E(SE)}
\mathbb{E}\left[ \mathrm{SE_{\mathrm{MT}_{0}}} \right]&=\sum\limits_{j\in \mathcal{K}}{\sum\limits_{m\in \mathcal{K}}{\Pr \left( \mathcal{X}_{\mathrm{MT}_{0}}^{(j,m)},\mathcal{A}_{\mathrm{MT}_{0}} \right)}}
\nonumber \\
& \hspace{-0.5cm} \times \sum\limits_{i_{\mathrm{CQI}}=1}^{n_{\mathrm{CQI}}}{\mathrm{SE}_{i_{\mathrm{CQI}}}}
\Bigg( \bar{F}_{\mathrm{SINR}}\left( \gamma _{i_{\mathrm{CQI}}}|\mathcal{X}_{\mathrm{MT}_{0}}^{(j,m)},\mathcal{A}_{\mathrm{MT}_{0}} \right)
-\bar{F}_{\mathrm{SINR}}\left( \gamma _{i_{\mathrm{CQI}}+1}|\mathcal{X}_{\mathrm{MT}_{0}}^{(j,m)},\mathcal{A}_{\mathrm{MT}_{0}} \right) \Bigg)
\end{align}
\fi

With similar arguments, the average BR of the probe MT can be written as follows:
\ifTwoColumns
\begin{align}
\label{eq:AMC E(BR)}
& \mathbb{E}\left[ \mathrm{BR_{\mathrm{MT}_{0}}} \right]=\sum\limits_{j\in \mathcal{K}}{\sum\limits_{m\in \mathcal{K}}{\sum\limits_{n>0}{\Pr \left( \mathcal{X}_{\mathrm{MT}_{0}}^{(j,m)},\mathcal{A}_{\mathrm{MT}_{0}} \right)}}}
\nonumber \\
& \quad \Pr \left( N_{\mathcal{B}_{\mathrm{MT}_{0}}}^{\mathcal{A}}=n|\mathcal{X}_{\mathrm{MT}_{0}}^{(j)},\mathcal{A}_{\mathrm{MT}_{0}} \right)
\nonumber \\
& \quad \sum\limits_{i_{\mathrm{CQI}}=1}^{n_{\mathrm{CQI}}}{\frac{b_{w}}{n}\mathrm{SE}_{i_{\mathrm{CQI}}}}\Bigg(\bar{F}_{\mathrm{SINR}}\left( \gamma _{i_{\mathrm{CQI}}}|\mathcal{X}_{\mathrm{MT}_{0}}^{(j,m)},\mathcal{A}_{\mathrm{MT}_{0}} \right)
\nonumber  \\
& \quad -\bar{F}_{\mathrm{SINR}}\left( \gamma _{i_{\mathrm{CQI}}+1}|\mathcal{X}_{\mathrm{MT}_{0}}^{(j,m)},\mathcal{A}_{\mathrm{MT}_{0}} \right) \Bigg)
%
\nonumber \\ & \quad
 \overset{\text{(a)}}{\mathop{=}} \sum\limits_{j \in {\cal K}} \sum\limits_{m \in {\cal K}} \sum\limits_{i_{{\rm{CQI}}}  = 1}^{n_{{\rm{CQI}}} } \Pr \left( {{\cal X}_{{\rm{MT}}_0 }^{(j,m)} ,{\cal A}_{{\rm{MT}}_0 } } \right)
 \nonumber \\ & \quad \times
{\rm{SE}}_{i_{{\rm{CQI}}} }
 \frac{{3.5^{3.5} b_w \left( {3.5\lambda ^{(j)}  + \lambda _{{\rm{MT}}} p} \right)\left( {1 - \left( {1 - \frac{{\lambda _{{\rm{MT}}} p}}{{\lambda ^{(j)} }}} \right)^{3.5} } \right)}}{{\lambda _{{\rm{MT}}} p\left( {1 - \frac{{\lambda _{{\rm{MT}}} p}}{{\lambda ^{(j)} }}} \right)^{3.5} }}
\nonumber \\ & \quad \times
 \Bigg(\bar F_{{\rm{SINR}}} \left( {\gamma _{i_{{\rm{CQI}}} } |{\cal X}_{{\rm{MT}}_0 }^{(j,m)} ,{\cal A}_{{\rm{MT}}_0 } } \right)
\nonumber \\ & \quad
 - \bar F_{{\rm{SINR}}} \left( {\gamma _{i_{{\rm{CQI}}}  + 1} |{\cal X}_{{\rm{MT}}_0 }^{(j,m)} ,{\cal A}_{{\rm{MT}}_0 } } \right)\Bigg)
\end{align}
\else
\begin{align}
\label{eq:AMC E(BR)}
& \mathbb{E}\left[ \mathrm{BR_{\mathrm{MT}_{0}}} \right]=\sum\limits_{j\in \mathcal{K}}{\sum\limits_{m\in \mathcal{K}}{\sum\limits_{n>0}{\Pr \left( \mathcal{X}_{\mathrm{MT}_{0}}^{(j,m)},\mathcal{A}_{\mathrm{MT}_{0}} \right)}}}
\Pr \left( N_{\mathcal{B}_{\mathrm{MT}_{0}}}^{\mathcal{A}}=n|\mathcal{X}_{\mathrm{MT}_{0}}^{(j)},\mathcal{A}_{\mathrm{MT}_{0}} \right)
\nonumber \\
&  \quad \times \sum\limits_{i_{\mathrm{CQI}}=1}^{n_{\mathrm{CQI}}}{\frac{b_{w}}{n}\mathrm{SE}_{i_{\mathrm{CQI}}}}\Bigg(\bar{F}_{\mathrm{SINR}}\left( \gamma _{i_{\mathrm{CQI}}}|\mathcal{X}_{\mathrm{MT}_{0}}^{(j,m)},\mathcal{A}_{\mathrm{MT}_{0}} \right)
-\bar{F}_{\mathrm{SINR}}\left( \gamma _{i_{\mathrm{CQI}}+1}|\mathcal{X}_{\mathrm{MT}_{0}}^{(j,m)},\mathcal{A}_{\mathrm{MT}_{0}} \right) \Bigg)
\nonumber \\
& \quad \overset{\text{(a)}}{\mathop{=}} \sum\limits_{j \in {\cal K}} {\sum\limits_{m \in {\cal K}} {\sum\limits_{i_{{\rm{CQI}}}  = 1}^{n_{{\rm{CQI}}} } {\Pr \left( {{\cal X}_{{\rm{MT}}_0 }^{(j,m)} ,{\cal A}_{{\rm{MT}}_0 } } \right){\rm{SE}}_{i_{{\rm{CQI}}} } } } }
 \frac{{3.5^{3.5} b_w \left( {3.5\lambda ^{(j)}  + \lambda _{{\rm{MT}}} p} \right)\left( {1 - \left( {1 - \frac{{\lambda _{{\rm{MT}}} p}}{{\lambda ^{(j)} }}} \right)^{3.5} } \right)}}{{\lambda _{{\rm{MT}}} p\left( {1 - \frac{{\lambda _{{\rm{MT}}} p}}{{\lambda ^{(j)} }}} \right)^{3.5} }}
\nonumber \\
& \quad \quad \quad \quad \quad \quad \quad \times
  \left(\bar F_{{\rm{SINR}}} \left( {\gamma _{i_{{\rm{CQI}}} } |{\cal X}_{{\rm{MT}}_0 }^{(j,m)} ,{\cal A}_{{\rm{MT}}_0 } } \right) - \bar F_{{\rm{SINR}}} \left( {\gamma _{i_{{\rm{CQI}}}  + 1} |{\cal X}_{{\rm{MT}}_0 }^{(j,m)} ,{\cal A}_{{\rm{MT}}_0 } } \right)\right)
\end{align}
\fi
\noindent where (a) is obtained by computing the summation over $n=N_{\mathcal{B}_{\mathrm{MT}_{0}}}^{\mathcal{A}}$ in closed-form with the aid of the PMF in (\ref{eq:n^A_BMT0}). 

The mathematical expressions of SE and BR of AMC schemes are easier to compute than the corresponding formulas obtained from the Shannon definition of SE, since the latter definition requires an extra integral to be computed \cite{DiRenzo16a}. This is remarkable, since the SE and BR in \eqref{eq:AMC E(SE)} and \eqref{eq:AMC E(BR)} account for feedback's overhead and limited-complexity receivers.

In the present paper, as a sensible case study, we consider the range of CQI values and a target BLER equal to $10\%$, as recommended by LTE specifications \cite{Sesia09}. The SINR thresholds $\gamma _{i_{\mathrm{CQI}}}$ are obtained from link-level simulations conducted with an accurate LTE simulator \cite{wimo14,Martin-Vega13}. More precisely, the considered simulator assumes MTs of limited computational complexity, where decoding is performed by using a 1-tap zero forcing equalizer and a turbo decoder based on the soft output Viterbi algorithm. Numerical illustrations are reported in Section \ref{sec:Numerical Results}. For completeness, Table \ref{tab:AMC} reports the input parameters that are needed for computing the SE and BR in \eqref{eq:AMC E(SE)} and \eqref{eq:AMC E(BR)}. It is worth emphasizing, however, that \eqref{eq:AMC E(SE)} and \eqref{eq:AMC E(BR)} are general enough for being used for analyzing different wireless standards and receiver implementations.

\begin{table*}
\renewcommand{\arraystretch}{1.0}
\ifTwoColumns
\caption{SINR thresholds and SE values obtained from the LTE link-level simulator in \cite{wimo14,Martin-Vega13}.}
\else
\caption{SINR thresholds and SE values obtained from the LTE link-level simulator in \cite{wimo14,Martin-Vega13}. \vspace{-0.5cm}}
\fi
\label{tab:AMC}
\ifTwoColumns
\else
\scriptsize
\fi
\centering
\begin{tabular}{ c c c c c c c c c c c c c c c c}
\toprule
$i_\mathrm{CQI}$ & 1 & 2 & 3 & 4 & 5 & 6 & 7 & 8 & 9 & 10 & 11 & 12 & 13 & 14 & 15 \\
\toprule
$\mathrm{SE}_{i_\mathrm{CQI}}$ [bps/Hz] & 0.15 & 0.23 & 0.38 & 0.60 & 0.88 & 1.18 & 1.48 & 1.91 &
2.41 & 2.73 & 3.32 & 3.90 & 4.52 & 5.11 & 5.55 \\
\hline
$\gamma_{i_\mathrm{CQI}}$ [dB] &
-3.65 & -1.60 & 0.00 & 2.25 & 3.75 & 4.75 & 9.00 & 10.50 &
12.35 & 15.40 & 17.18 & 18.85 & 20.70 & 24.0 & 25.0 \\
\bottomrule
\end{tabular}
\end{table*}

\section{Numerical Results}
\label{sec:Numerical Results}

\ifTwoColumns
\begin{table}
\renewcommand{\arraystretch}{1.1}
\ifTwoColumns
\caption{Simulation setup.}
\else
\caption{Simulation setup.  \vspace{-1.5cm}}
\fi
\label{tab:Simulation Parameters}
\scriptsize
\centering
\begin{tabular}{ c c c c }
\toprule
Parameter & Value & Parameter & Value \\
\toprule
$f_c$ (MHz) & $2 \times 10^{3}$ & $h_{\mathrm{BS}}$ (m) & $10$ \\
\hline
$b_w$ (MHz) & $9$ & $t^{(1)}/t^{(2)}$ (dB) & $\{9,0\}$ \\
\hline
$\lambda^{(1)}$ (points/m$^2$) & $2 \times 10^{-6}$ & $\lambda^{(2)}$ (points/m$^2$) & $4 \times 10^{-6}$ \\
\hline
$\lambda_{\mathrm{MT}}$ (points/m$^2$) & $80 \times 10^{-6}$ & $n_\mathrm{thermal}$ (dBm/Hz) & $-174$  \\
\hline
$n_\mathrm{F}$ (dB) & $9$ & $\sigma_s$ (dB) & $4$  \\
\hline
$p_0$ (dBm) & $-70$ & $p_\mathrm{max}$ (dBm) & $\{\infty,5 \}$  \\
\hline
$i_0$ (dBm) & $[-120,-60 ]$ & $\epsilon$ & $[0,1]$  \\
\bottomrule
\end{tabular}
\end{table}
\else
\begin{table}
\renewcommand{\arraystretch}{1.0}
\caption{Simulation parameters.  \vspace{-0.5cm}}
\label{tab:Simulation Parameters}
\scriptsize
\centering
\begin{tabular}{ c c c c c c c c}
\toprule
Parameter & Value & Parameter & Value & Parameter & Value & Parameter & Value\\
\toprule
$\{\tau,\alpha\}$ (MHz) & $\{2.6,3.8\}$ & $n_\mathrm{thermal}$ (dBm/Hz) & $-174$
& $i_0$ (dBm) & $[-120,-60 ]$ & $\epsilon$ & $[0,1]$ \\
\hline
$b_w$ (MHz) & $9$ & $t^{(1)}/t^{(2)}$ (dB) & $\{9,0\}$
& $p_0$ (dBm) & $-70$ & $p_\mathrm{max}$ (dBm) & $\{\infty,5 \}$\\
\hline
$\{\lambda^{(1)},\lambda^{(2)}\}$ (points/km$^2$) & $\{2 ,4 \}$ & $\lambda_{\mathrm{MT}}$ (points/km$^2$) & $80 $
& $n_\mathrm{F}$ (dB) & $9$ & $\sigma_s$ (dB) & $4$ \\
\bottomrule
\end{tabular}
\end{table}
\fi

In this section, we validate the mathematical frameworks and findings derived in the previous sections with the aid of Monte Carlo simulations, as well as compare the IAM scheme against IAFPC and IUFPC schemes. The following setup compliant with LTE specifications is considered. The bandwidth is equal to $10$ MHz, which implies $b_w=9$ MHz by excluding the guard bands. The noise power spectral density is $n_\mathrm{thermal}=-174$ dBm/Hz and the noise figure of the receiver is $n_\mathrm{F}=9$ dB. Both GCA and SPLA criteria are studied, and the association weights are, unless otherwise stated, $t^{(1)}/t^{(2)}=9$ dB and $t^{(1)}/t^{(2)}=0$ dB, respectively.
The case study $t^{(1)}/t^{(2)}=9$ dB is related to a cell association based on the average DL received power criterion, where the first tier of BSs (macro) has transmit power equal to $46$ dBm and the second tier of BSs (small-cell) has transmit power equal to $37$ dBm, which agrees with \cite[Annex A: Simulation Model]{3gpp36872}.
Other simulation parameters are provided in Table \ref{tab:Simulation Parameters}. As far as Monte Carlo simulations are concerned, they are obtained by considering $10^{4}$ realizations of channels and network topologies. In all the figures, analytical and Monte Carlo simulation results are represented with solid lines and markers, respectively.

\subsection{Average Transmit Power, Probability of Being Active, Mean \& Variance of the Interference}
\label{sec:E(P), E(I) and var(I)}
In this section, we analyze the average transmit power of the MTs, the probability that the typical MT is active, which provides information on the system fairness, and the mean and variance of the interference.
%
%
Figures \ref{fig:comparativa_i0_iafpc_pmax_avPMT}-\ref{fig:comparativa_i0_iafpc_pmax_varI} confirm the conclusions drawn in \textbf{Remark \ref{rem:Exact analysis}}, i.e., the mathematical frameworks of average transmit power and probability of being active are exact while those of mean and variance of the interference are approximations that exploit \textbf{Assumption \ref{assumpt: Interference}}. 
{\color{black} Such an assumption considers that the position of interfering MTs can be modeled as a conditionally thinned (i.e., non-homogeneous) PPP. 
The difference between such a non-homogeneous PPP, and the actual point process, which is on the other hand not tractable, explains also the difference between simulation and analytical results in all the metrics that depend on the interference (SINR, SE, BR).}
The conclusions drawn in \textbf{Remark \ref{rem:IA vs Non IA}} are confirmed as well: the mean and variance of the interference decrease by decreasing $i_0$, which provide important advantages for implementing AMC schemes. 
%
\ifTwoColumns
    \begin{figure}[t]
    \centering
    \includegraphics[width=\@figSize]{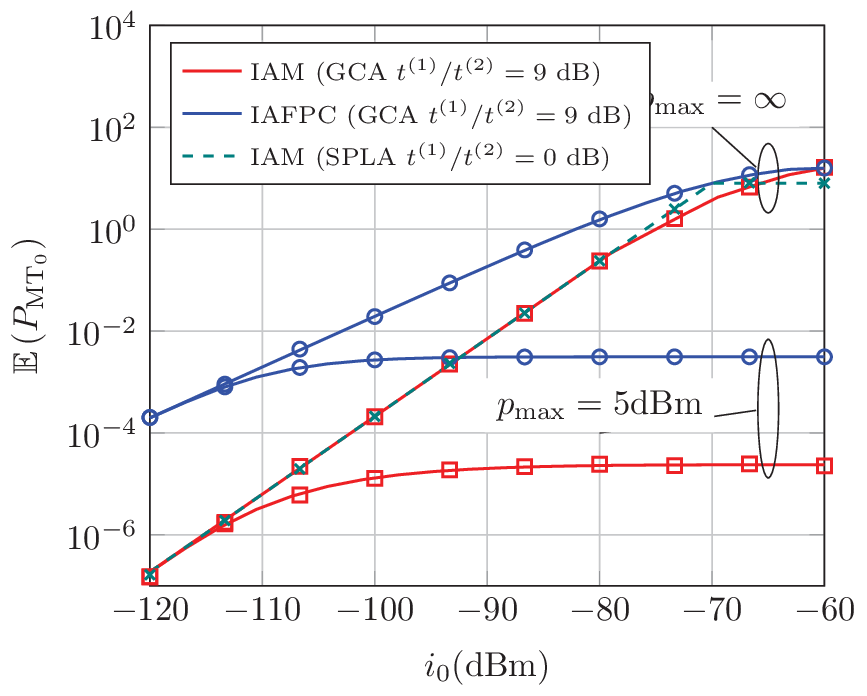}
    \caption{Average transmit power versus $i_0$ for IAM and IAFPC methods with $\epsilon=1$,
     $p_\mathrm{max} \rightarrow \infty$ and $p_\mathrm{max} = 5$ dBm.}
    \label{fig:comparativa_i0_iafpc_pmax_avPMT}
    \end{figure}

    \begin{figure}[t]
    \centering
    \includegraphics[width=\@figSize]{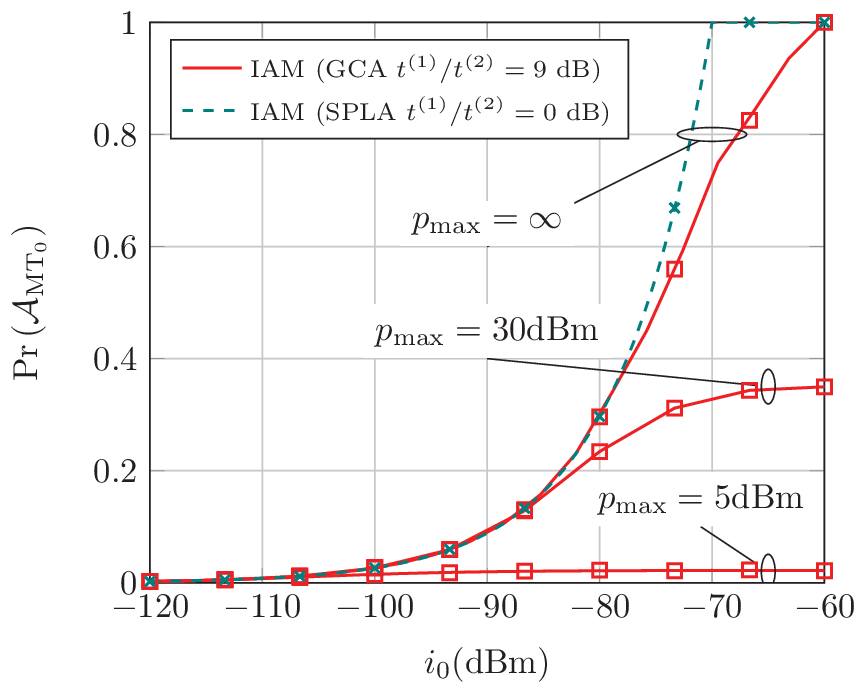}
    \caption{Probability of being active of the typical MT for IAM with $\epsilon=1$,
    $p_\mathrm{max} \rightarrow \infty$, $p_\mathrm{max} = 30$ dBm and
    $p_\mathrm{max} = 5$ dBm.}
    \label{fig:comparativa_i0_hia_pmax_PrA}
    \end{figure}

    \begin{figure}[t]
    \centering
    \includegraphics[width=\@figSize]{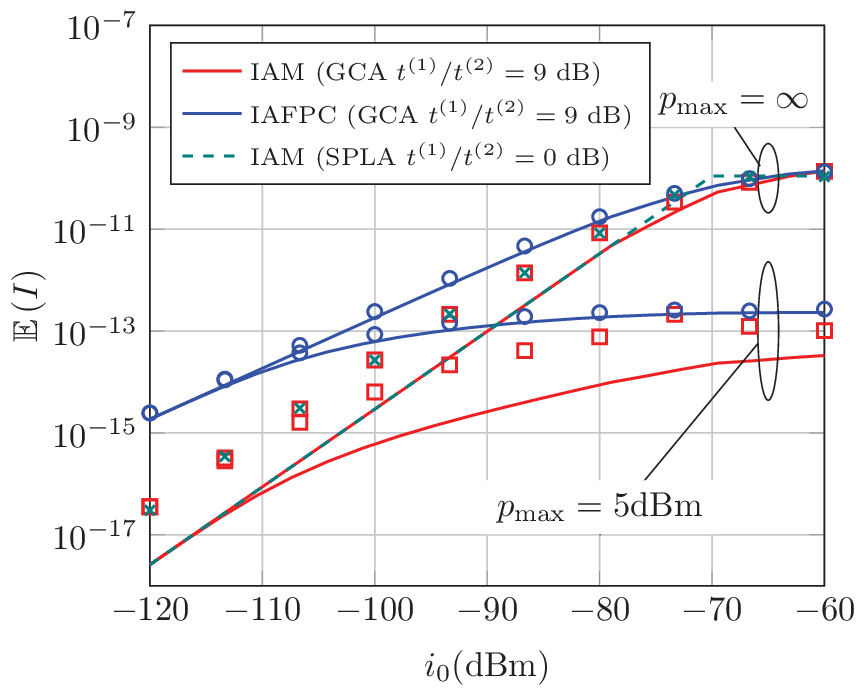}
    \caption{Mean of the interference versus $i_0$ for IAM and IAFPC methods with $\epsilon=1$,
     $p_\mathrm{max} \rightarrow \infty$ and $p_\mathrm{max} = 5$ dBm.}
    \label{fig:comparativa_i0_iafpc_pmax_avI}
    \end{figure}

\else
\begin{figure}[!tbp]
  \centering
  \begin{minipage}[b]{0.32\textwidth}
    \includegraphics[width=\textwidth]{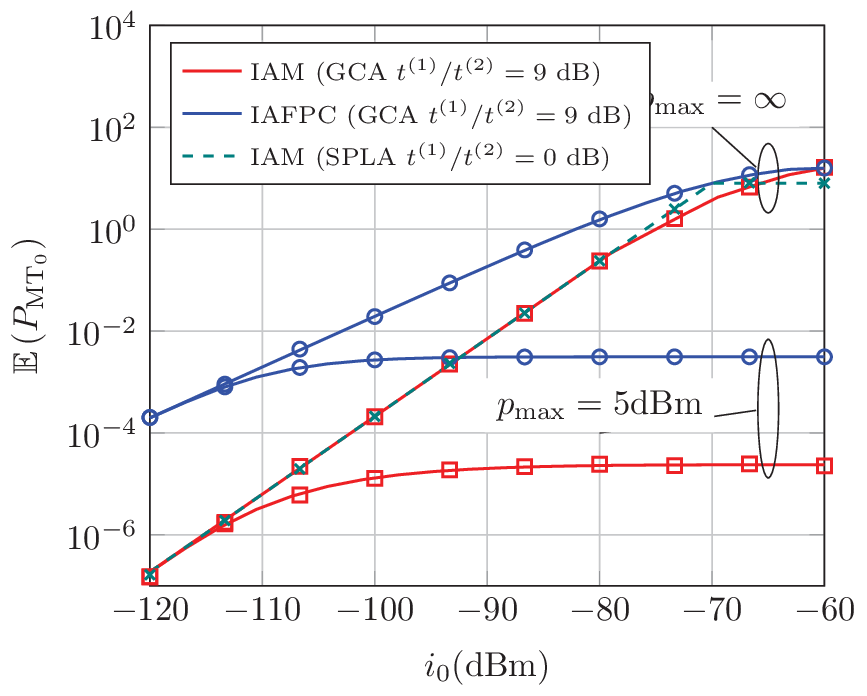}
    \caption{Average transmit power versus $i_0$ for IAM and IAFPC methods with $\epsilon=1$,
     $p_\mathrm{max} \rightarrow \infty$ and $p_\mathrm{max} = 5$ dBm.}
     \label{fig:comparativa_i0_iafpc_pmax_avPMT}
  \end{minipage}
  \hfill
  \begin{minipage}[b]{0.32\textwidth}
    \includegraphics[width=\textwidth]{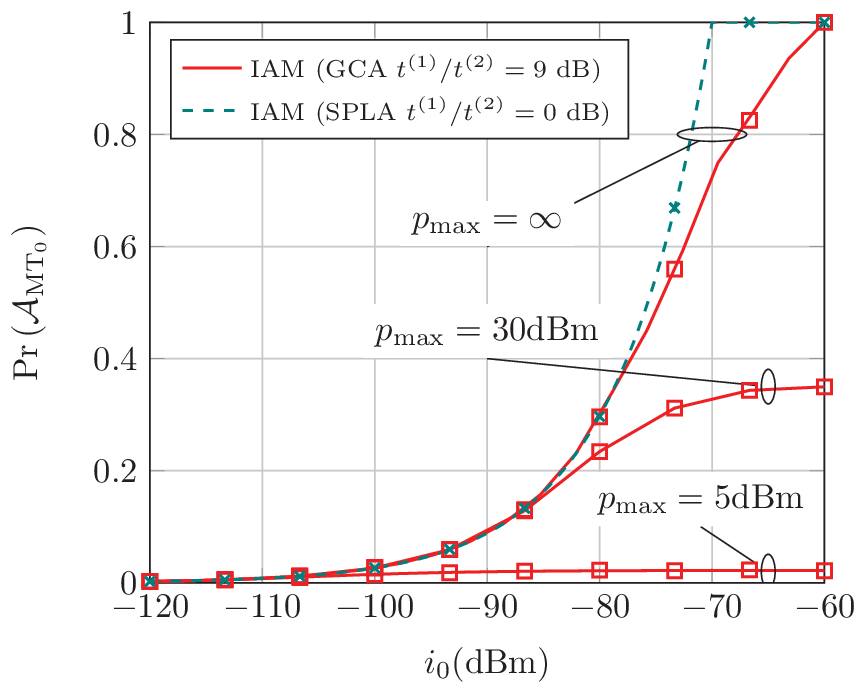}
    \caption{Probability of being active of the typical MT for IAM with $\epsilon=1$,
    $p_\mathrm{max} \rightarrow \infty$, $p_\mathrm{max} = 30$ dBm and
    $p_\mathrm{max} = 5$ dBm.}
    \label{fig:comparativa_i0_hia_pmax_PrA}
  \end{minipage}
  \hfill
  \begin{minipage}[b]{0.32\textwidth}
    \includegraphics[width=\textwidth]{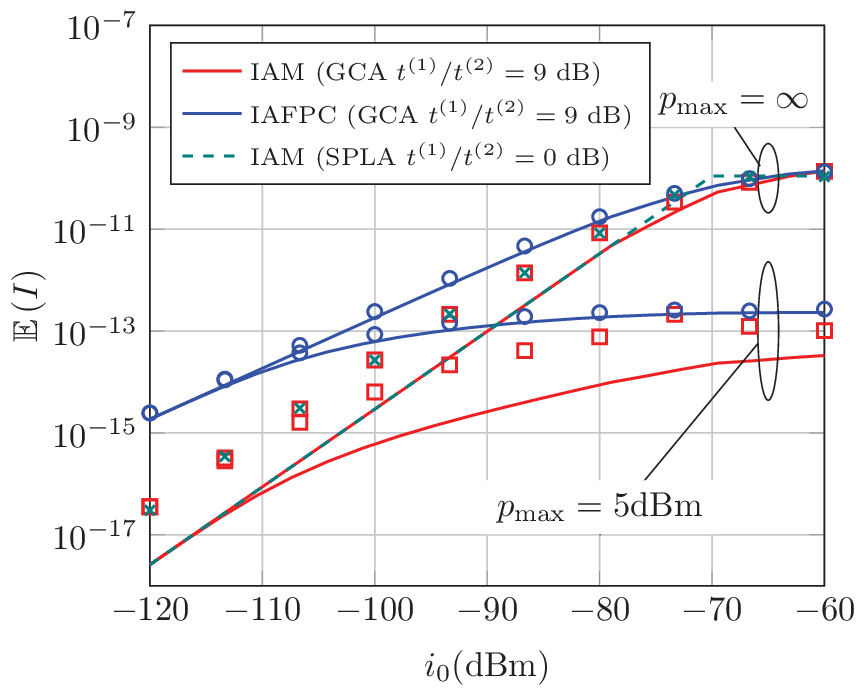}
    \caption{Mean of the interference versus $i_0$ for IAM and IAFPC methods with $\epsilon=1$,
     $p_\mathrm{max} \rightarrow \infty$ and $p_\mathrm{max} = 5$ dBm.}
    \label{fig:comparativa_i0_iafpc_pmax_avI}
  \end{minipage}
\end{figure}
\fi

In the figures, IAM and IAFPC are compared as well. We observe that IAM reduces the average transmit power and the mean and variance of the interference. 

Consider the SPLA criterion, which is illustrated with dashed lines in the figures. We observe that the findings in \textbf{Remark \ref{rem: Interference unaware regime}} are confirmed: the system is interference-aware and interference-unaware if $i_0<p_0$ and $i_0>p_0$, respectively. As expected, the crossing point occurs at $p_0=-70$ dBm based on the simulation parameters used. In addition, the scaling laws of average transmit power and average interference are in agreement with the findings in \textbf{Remark \ref{rem: power law for the average power}}, \textbf{Remark \ref{rem: power law for the avI}}.

All in all, the numerical illustrations reported in Figs. \ref{fig:comparativa_i0_iafpc_pmax_avPMT}-\ref{fig:comparativa_i0_iafpc_pmax_varI} confirm all the conclusions and performance trends discussed in the previous sections and highlight the advantages of IAM.

\ifTwoColumns
    \begin{figure}[t]
    \centering
    \includegraphics[width=\@figSize]{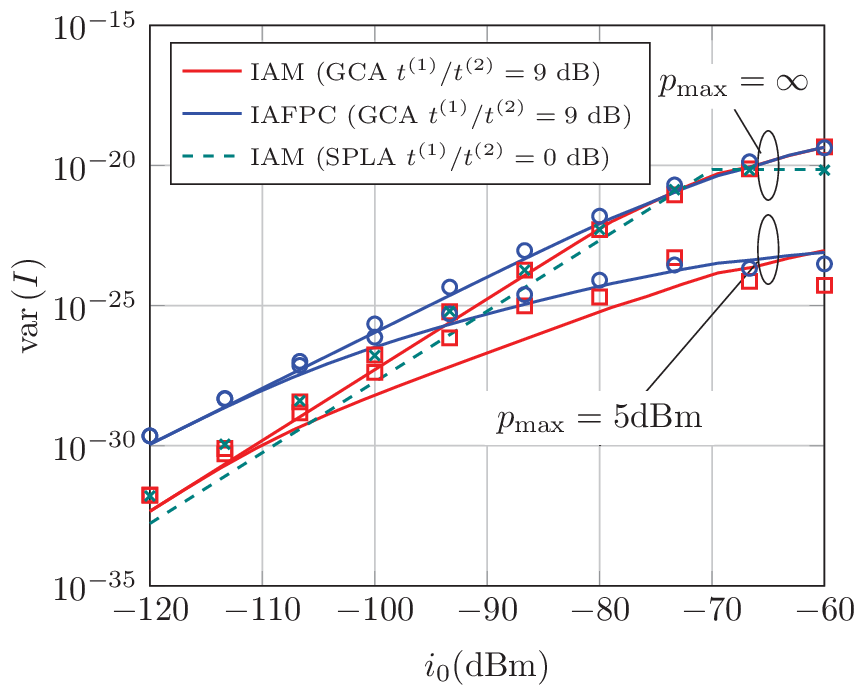}
    \caption{Variance of the interference versus $i_0$ for IAM and IAFPC schemes with $\epsilon=1$,
    $p_\mathrm{max} \rightarrow \infty$ and $p_\mathrm{max} = 5$ dBm.}
    \label{fig:comparativa_i0_iafpc_pmax_varI}	
    \end{figure}

    \begin{figure}[t]
    \centering
    \includegraphics[width=\@figSize]{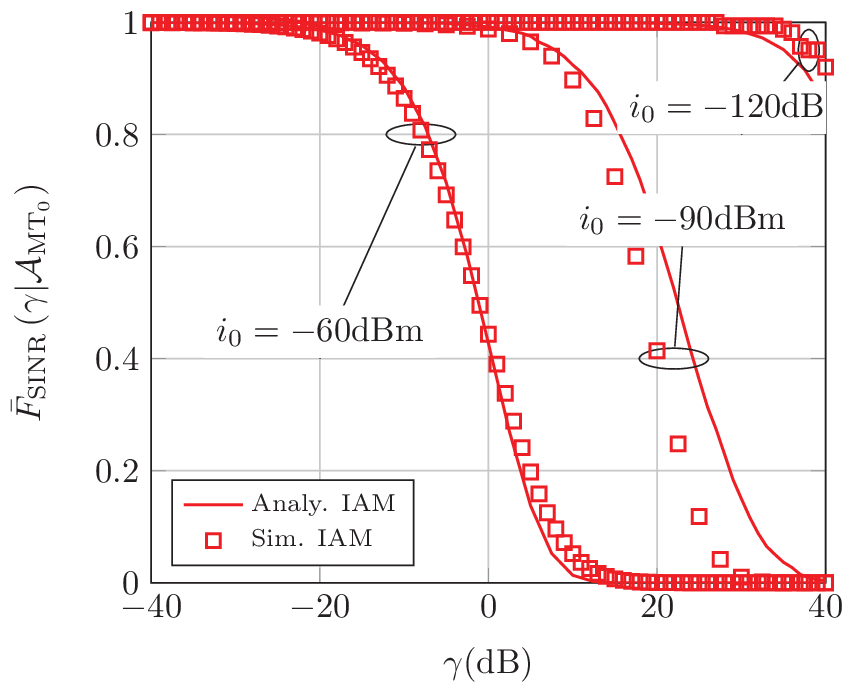}
    \caption{CCDF of the SINR for the typical MT conditioned on being active for IAM with $\epsilon=1$,
     $t^{(1)}/t^{(2)}=9$ dB, $p_{\max} \to \infty$ and $i_0=\{-120, -90, -60\}$ dBm.}
    \label{fig:comparativa_hia_active_epsilon1}
    \end{figure}

    \begin{figure}[t]
    \centering
    \includegraphics[width=\@figSize]{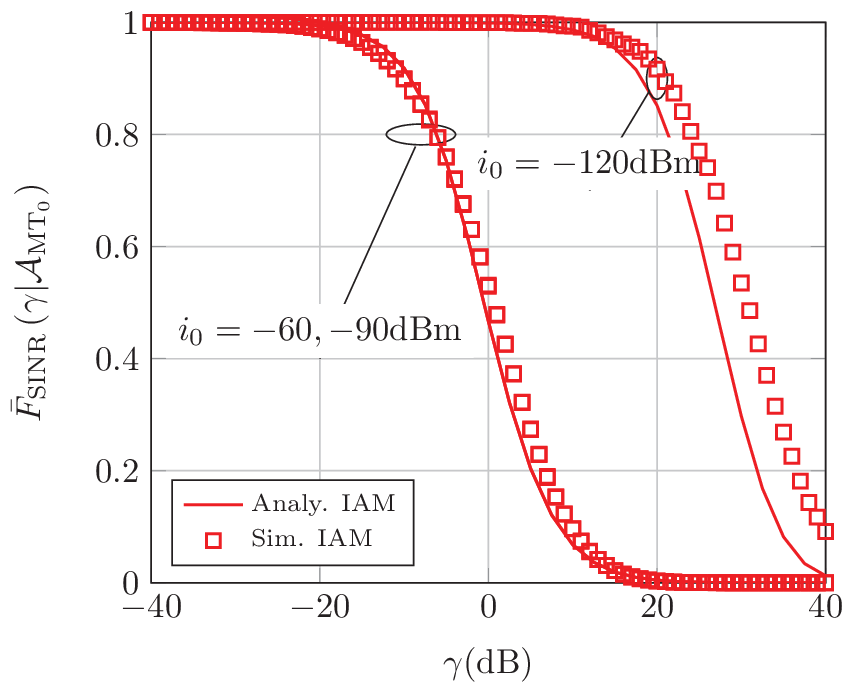}
    \caption{CCDF of the SINR for the typical MT conditioned on being active for IAM with $\epsilon=0.75$,
     $t^{(1)}/t^{(2)}=9$ dB, $p_{\max} \to \infty$ and $i_0=\{-120, -90, -60\}$ dBm.}
    \label{fig:comparativa_hia_active_epsilon0p75}	
    \end{figure}

\else
\begin{figure}[!tbp]
  \centering
  \begin{minipage}[b]{0.32\textwidth}
    \includegraphics[width=\textwidth]{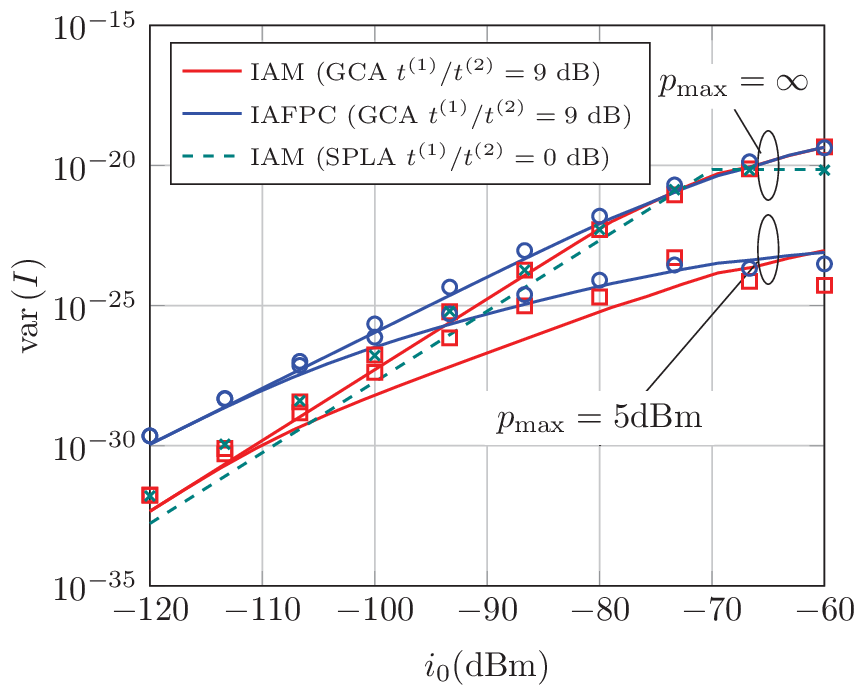}
    \caption{Variance of the interference versus $i_0$ for IAM and IAFPC schemes with $\epsilon=1$,
    $p_\mathrm{max} \rightarrow \infty$ and $p_\mathrm{max} = 5$ dBm.}
    \label{fig:comparativa_i0_iafpc_pmax_varI}	
  \end{minipage}
  \hfill
  \begin{minipage}[b]{0.32\textwidth}
    \includegraphics[width=\textwidth]{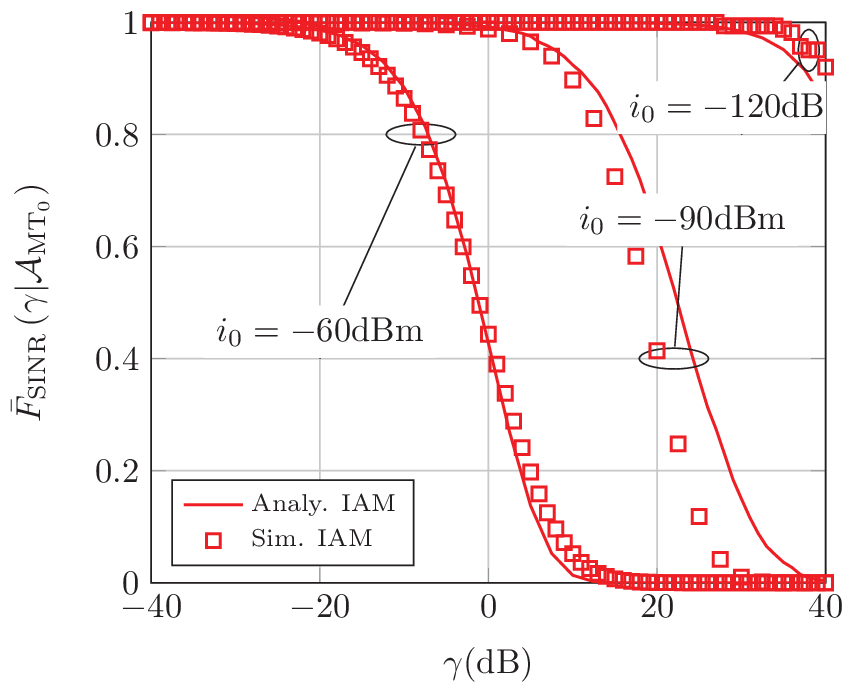}
    \caption{CCDF of the SINR for the typical MT conditioned on being active for IAM with $\epsilon=1$,
     $t^{(1)}/t^{(2)}=9$ dB, $p_{\max} \to \infty$ and $i_0=\{-120, -90, -60\}$ dBm.}
    \label{fig:comparativa_hia_active_epsilon1}
  \end{minipage}
  \hfill
  \begin{minipage}[b]{0.32\textwidth}
    \includegraphics[width=\textwidth]{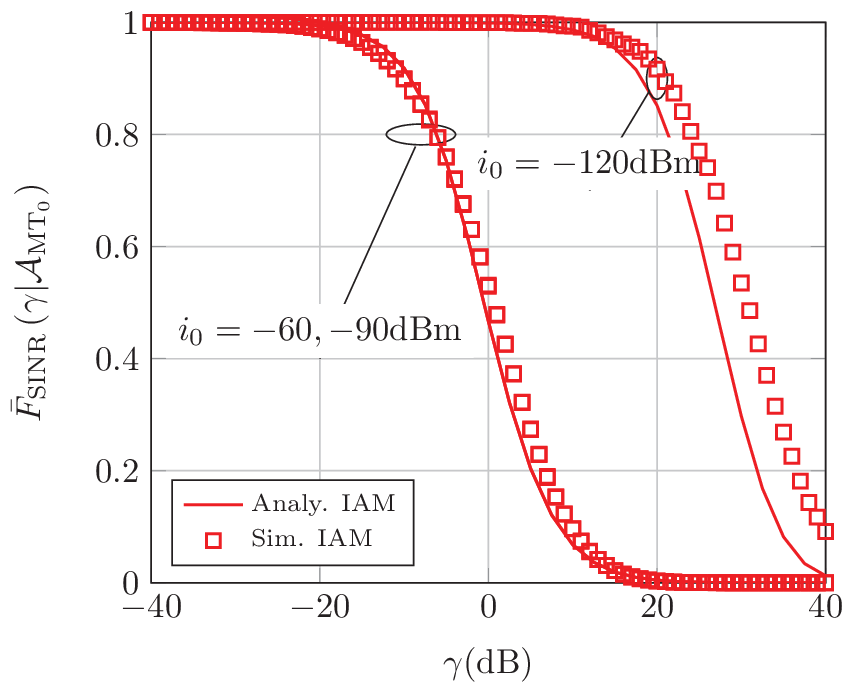}
    \caption{CCDF of the SINR for the typical MT conditioned on being active for IAM with $\epsilon=0.75$,
     $t^{(1)}/t^{(2)}=9$ dB, $p_{\max} \to \infty$ and $i_0=\{-120, -90, -60\}$ dBm.}
    \label{fig:comparativa_hia_active_epsilon0p75}	
  \end{minipage}
\end{figure}
\fi

%
%

\subsection{Complementary Cumulative Distribution Function of the SINR}
\label{sec:ccdf of the SINR}
In this section, we analyze the coverage probability (CCDF of the SINR) of the active MTs. 
The results are illustrated in Figs. \ref{fig:comparativa_hia_active_epsilon1} and \ref{fig:comparativa_hia_active_epsilon0p75} for $\epsilon=1$ and $\epsilon=0.75$, respectively, and by assuming $p_\mathrm{max} \rightarrow \infty$.

In both figures, we observe a good agreement between mathematical frameworks and Monte Carlo simulations. 
In particular, the figures confirm, once again, that the coverage probability of IAM increases as $i_0$ decreases. In Fig. \ref{fig:comparativa_hia_active_epsilon1}, for example, almost all the active MTs have a SINR greater than $20$ dB if $i_0=-120$ dBm. This good SINR is obtained because IAM keeps under control the interference by muting the MTs that create more interference. Based on Fig. \ref{fig:comparativa_i0_hia_pmax_PrA}, in fact, we note that only a small fraction of the MTs are allowed to be active for $i_0=-120$ dBm. The active MTs, however, better exploit the available bandwidth. Similar conclusions can be drawn for $\epsilon=0.75$ shown in Fig. \ref{fig:comparativa_hia_active_epsilon0p75}. The main difference is that, in this latter figure, IAM provides almost the same coverage probability for $i_0=-60$ dBm and $i_0=-90$ dBm. The reason is that the MTs transmit with less power if $\epsilon=0.75$ and, thus, there is almost no difference between the two interference constraints. 
This brings to our attention that the design of the UL of HCNs requires to jointly optimize $i_0$, $p_0$, $p_\mathrm{max}$ and $\epsilon$, in order to identify the desired operating regime that fulfills the requirements in terms of system fairness and interference mitigation. The proposed mathematical frameworks can be used to this end.

\subsection{Spectral Efficiency and Binary Rate}
\label{sec:Binary rate and spectral efficiency}
In this section, the average SE and average BR are analyzed, as well as the IAFPC%
\footnote{It is worth noting that the average SE of the IAFPC scheme based on the Shannon formula was analyzed in \cite{Martin-Vega16}. In the present paper, we focus our attention of the more practical definition provided in Section VI.} and IAM schemes are compared against each other for several system setups.
%
%
\ifTwoColumns
    \begin{figure}[t]
    \centering
    \includegraphics[width=\@figSize]{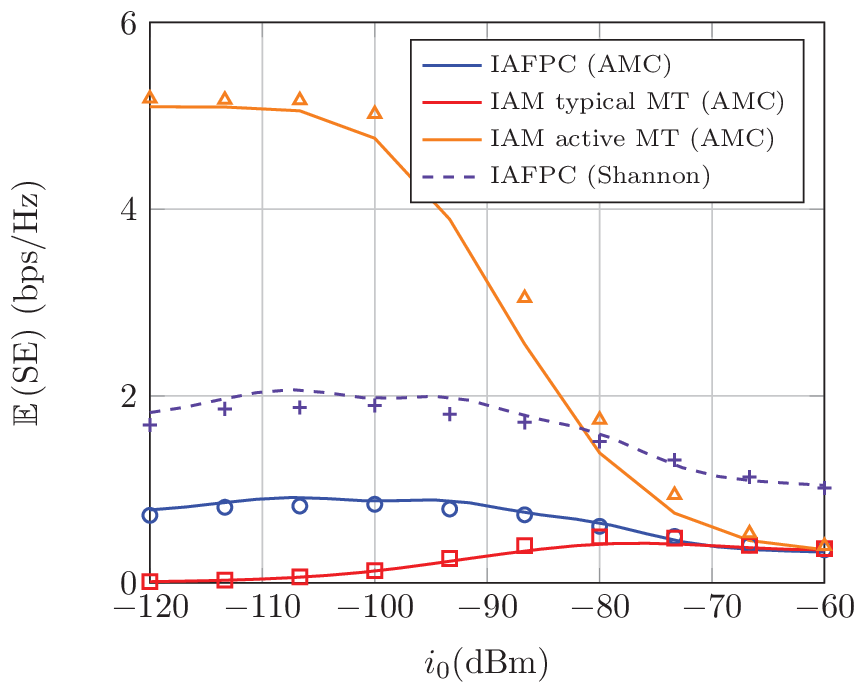}
    \caption{Comparison of average SE of IAFPC and IAM for $\epsilon=1$,
    $t^{(1)}/t^{(2)}=9$ dB and $p_{\mathrm{max}} \rightarrow \infty$.
    As for IAFPC, the average SE based on the Shannon formula is shown as well.}
    \label{fig:comparativa_i0_pmax_inf_avSE_AMC}
    \end{figure}

    \begin{figure}[t]
    \centering
    \includegraphics[width=\@figSize]{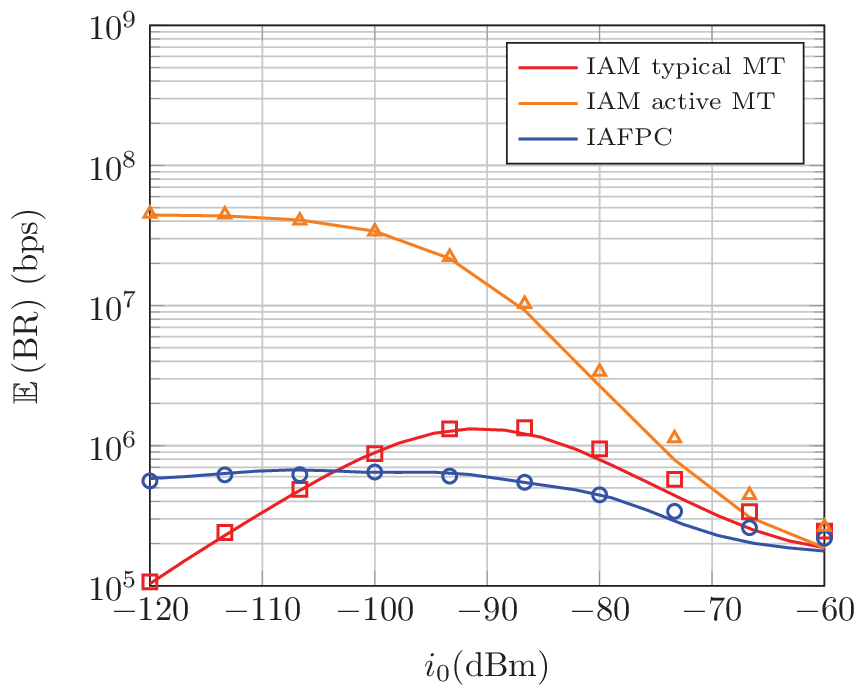}
    \caption{Comparison of average BR of IAFPC and IAM for $\epsilon=1$,
    $t^{(1)}/t^{(2)}=9$ dB and $p_{\mathrm{max}} \rightarrow \infty$.
    As for IAM, two cases are considered: the typical MT and the typical active MT.}
    \label{fig:comparativa_i0_pmax_inf_avBR_AMC}
    \end{figure}

    \begin{figure}[t]
    \centering
    \includegraphics[width=\@figSize]{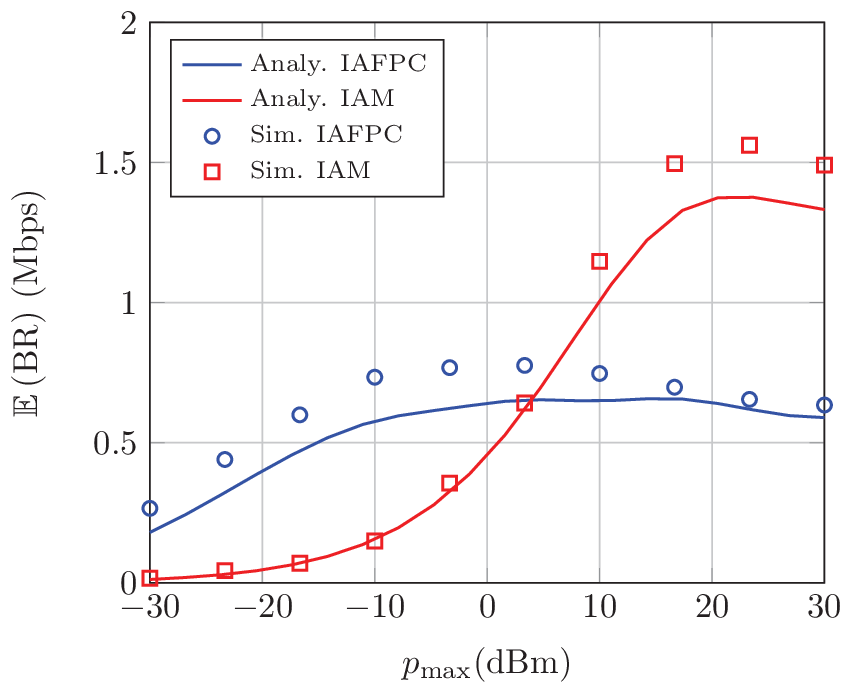}
    \caption{Comparison of average BR of IAFPC and IAM for $\epsilon=1$,
    $t^{(1)}/t^{(2)}=9$ dB and $i_{0} = -90$ dBm.}
    \label{fig:comparativa_pmax_i0_m90dBm_avBR_AMC}
    \end{figure}
\else
\begin{figure}[!tbp]
  \centering
  \begin{minipage}[b]{0.32\textwidth}
    \includegraphics[width=\textwidth]{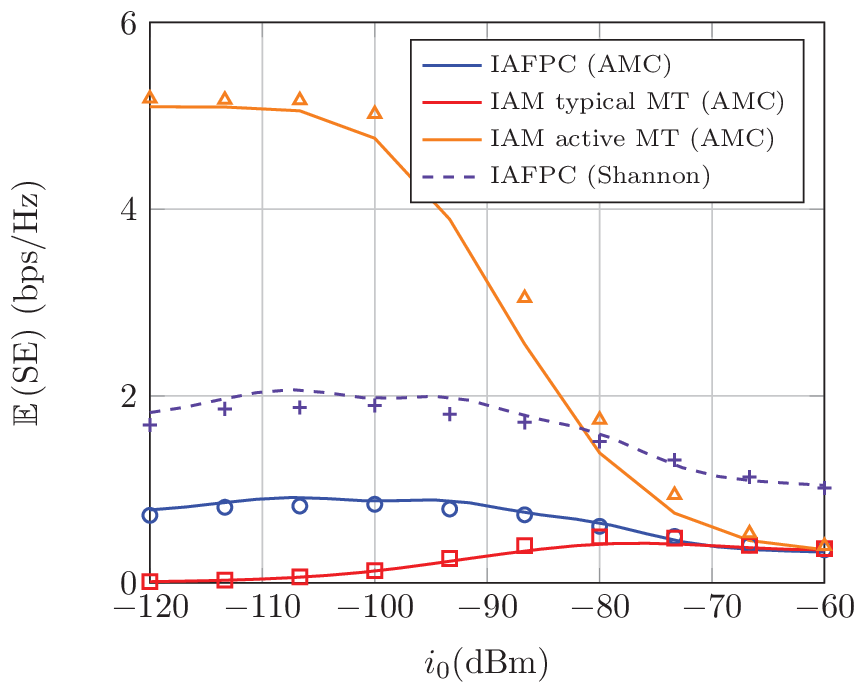}
    \caption{Comparison of average SE of IAFPC and IAM for $\epsilon=1$,
    $t^{(1)}/t^{(2)}=9$ dB and $p_{\mathrm{max}} \rightarrow \infty$.
    As for IAFPC, the average SE based on the Shannon formula is shown as well.}
    \label{fig:comparativa_i0_pmax_inf_avSE_AMC}
  \end{minipage}
  \hfill
  \begin{minipage}[b]{0.32\textwidth}
    \includegraphics[width=\textwidth]{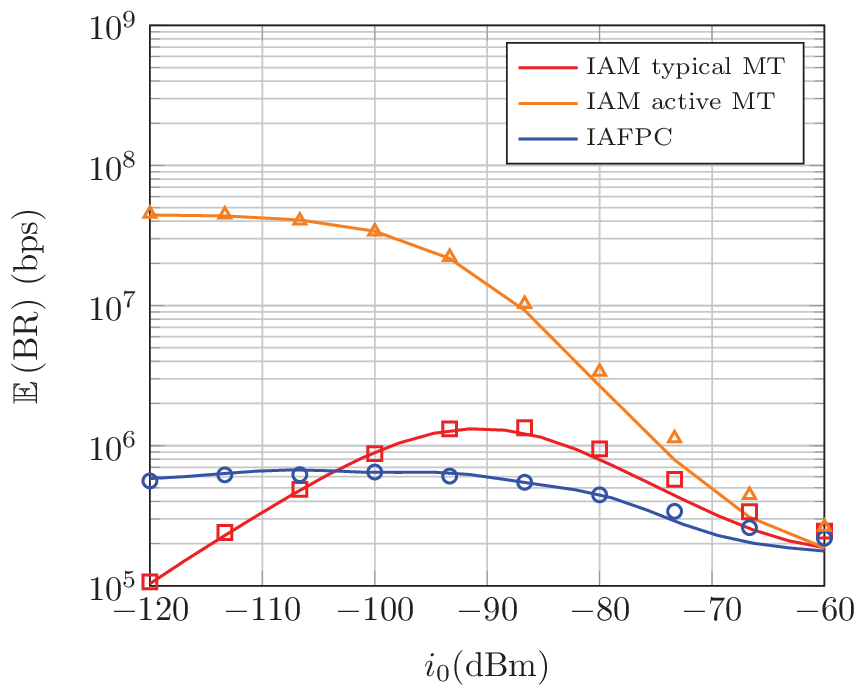}
    \caption{Comparison of average BR of IAFPC and IAM for $\epsilon=1$,
    $t^{(1)}/t^{(2)}=9$ dB and $p_{\mathrm{max}} \rightarrow \infty$.
    As for IAM, two cases are considered: the typical MT and the typical active MT.}
    \label{fig:comparativa_i0_pmax_inf_avBR_AMC}
  \end{minipage}
  \hfill
  \begin{minipage}[b]{0.32\textwidth}
    \includegraphics[width=\textwidth]{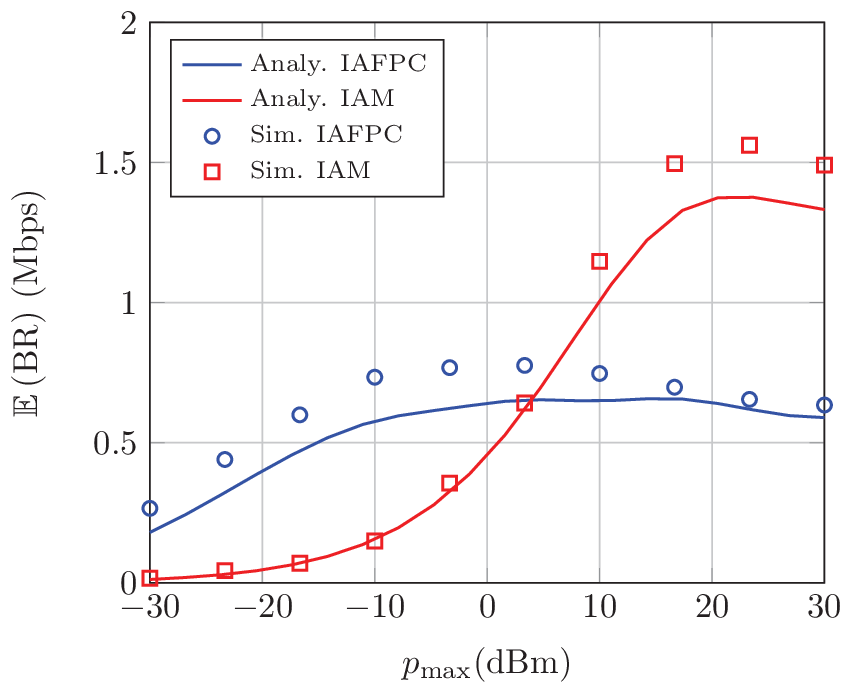}
    \caption{Comparison of average BR of IAFPC and IAM for $\epsilon=1$,
    $t^{(1)}/t^{(2)}=9$ dB and $i_{0} = -90$ dBm.}
    \label{fig:comparativa_pmax_i0_m90dBm_avBR_AMC}
  \end{minipage}
\end{figure}
\fi

In Fig. \ref{fig:comparativa_i0_pmax_inf_avSE_AMC}, the average SE of IAFPC and IAM schemes is analyzed and three conclusions can be drawn. By comparing the average SE of the IAPFC scheme based on the definition given in Section VI (i.e., for AMC schemes) and on the Shannon formula, we note, as expected, that the latter formula provides optimistic estimates of the average SE. By comparing the average SE of the IAM scheme for typical (active and muted) MTs and active (only) MTs, we note a different performance trend as a function of $i_0$. 
As for the active MTs, the average SE increases as $i_0$ decreases. As for the typical MTs, on the other hand, the average SE decreases as $i_0$ decreases. 
This is because the lower $i_0$ is the more MTs are turned off, which on average, contributes to reduce the SE of the typical MT. 
By comparing the average SE of IAPFC and IAM schemes, we evince that IAFPC outperforms IAM for all relevant values of the maximum interference constraint $i_0$, since all the MTs are active under the IAFPC scheme. The average SE of the active MTs under the IAM scheme is, however, much better than that of the IAFPC scheme, since the other-cell interference is reduced.

As discussed in Section I, however, the SE does not provide information on the amount of bandwidth that the scheduler allocates to each active MT. 

This trade-off is captured by the average BR, which is shown Fig. \ref{fig:comparativa_i0_pmax_inf_avBR_AMC}. As far as the average BR is concerned, in particular, we note that IAFPC and IAM schemes provide opposite trends compared to those evinced from the analysis of the average SE of the typical MT. More precisely, IAM provides a better average BR than IAFPC and there exists an optimal value of $i_0$ that maximizes it. This optimal value of $i_0$ emerges if the typical MT is considered, i.e., the MT may be either active or inactive. The figure, however, shows the average BR achieved only by the active MTs as well. In this case, we note that the MTs that satisfy both power and interference constraints achieve a very high throughput due to the reduce level of interference that is generated in this case. In a nutshell, IAM outperforms IAFPC in terms of average BR because the available bandwidth is shared among fewer MTs (only those active), which results in a higher throughput for each of them. Even though some MTs may be turned off in IAM, this may not necessarily be considered as a downside from the user's perspective: in high-mobility scenarios, for example, some MTs may prefer to be muted for some periods of time if their reward is achieving a higher throughput once they are allowed to transmit.
In Fig. \ref{fig:comparativa_pmax_i0_m90dBm_avBR_AMC}, we study the impact of $p_\mathrm{max}$ for a given maximum interference constraint $i_0$. We observe that $p_\mathrm{max}$ plays a critical role as well and highly affects the average BR. This figure confirms, once again, that both $p_\mathrm{max}$ and $i_0$ constraints need to be appropriately optimized in order for IAM to outperform IAFPC. 
%
In Fig. \ref{fig:avBR_vs_varI}, we illustrate the potential of IAM of reducing the variance of the interference compared with IUM, while still guaranteeing the same average BR. As discussed in the previous sections, this is beneficial for implementing AMC schemes. The figure shows a four-order magnitude reduction of the variance of the interference for the considered setup of parameters.

\ifTwoColumns
    \begin{figure}[t]
    \centering
    \includegraphics[width=\@figSize]{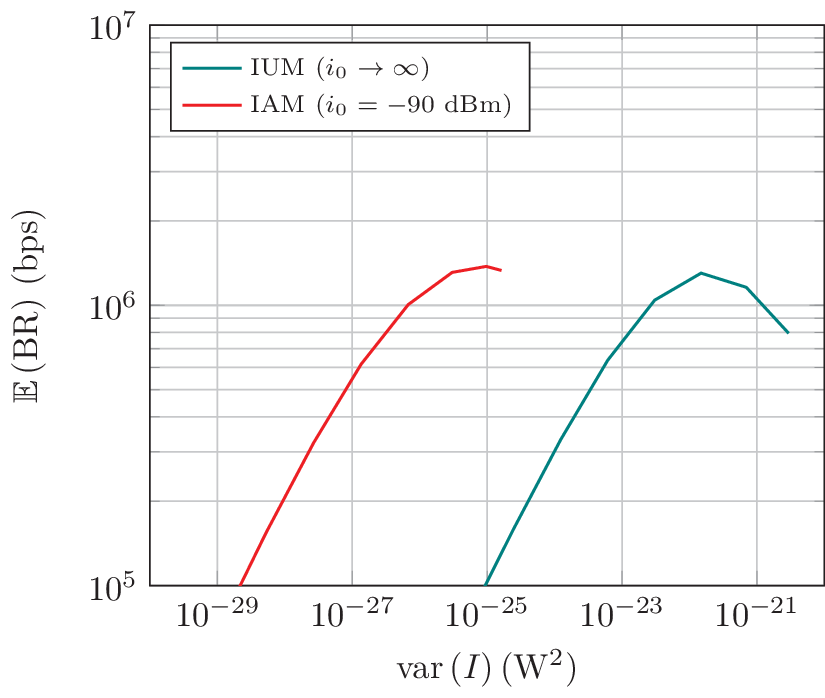}
    \caption{Average BR as a function of the variance of the interference for IAM
    ($i_{0} = -90$ dBm) and IUM ($i_{0} \to \infty$) schemes with $p_{\max} = 5$ dBm and
    $t^{(1)}/t^{(2)}=9$ dB.}
    \label{fig:avBR_vs_varI}
    \end{figure}

    \begin{figure}[t]
    \centering
    \includegraphics[width=\@figSize]{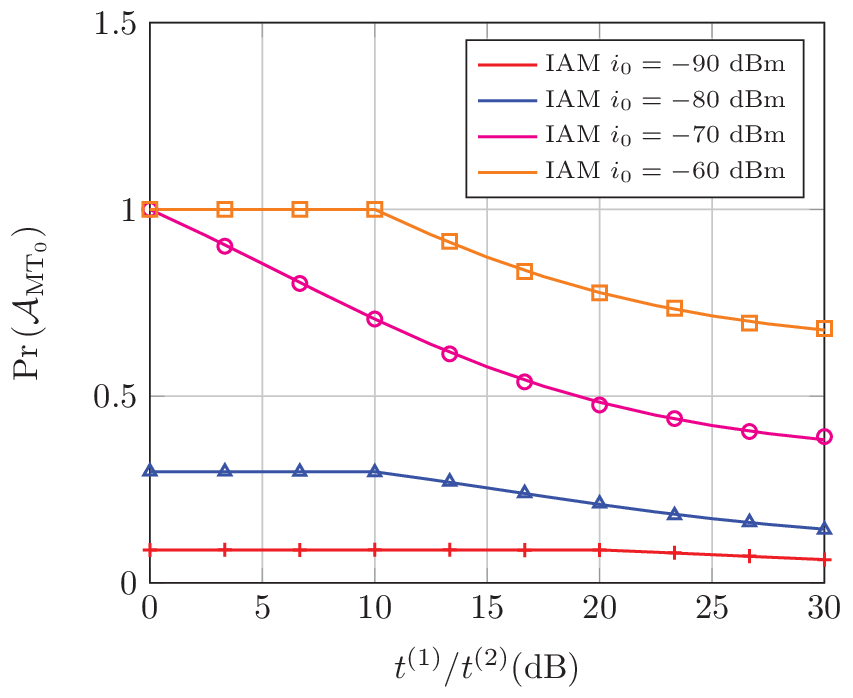}
    \caption{Probability the typical MT is active as a function of
    $t^{(1)}/t^{(2)}$ for IUFPC ($i_{0} \to \infty$) and IAM with $i_0=\{-90,-80,-70,-60\}$ dBm. $p_{\max} \to \infty$ for both schemes.}
    \label{fig:PrA_vs_t1t2}
    \end{figure}

    \begin{figure}[t]
    \centering
    \includegraphics[width=\@figSize]{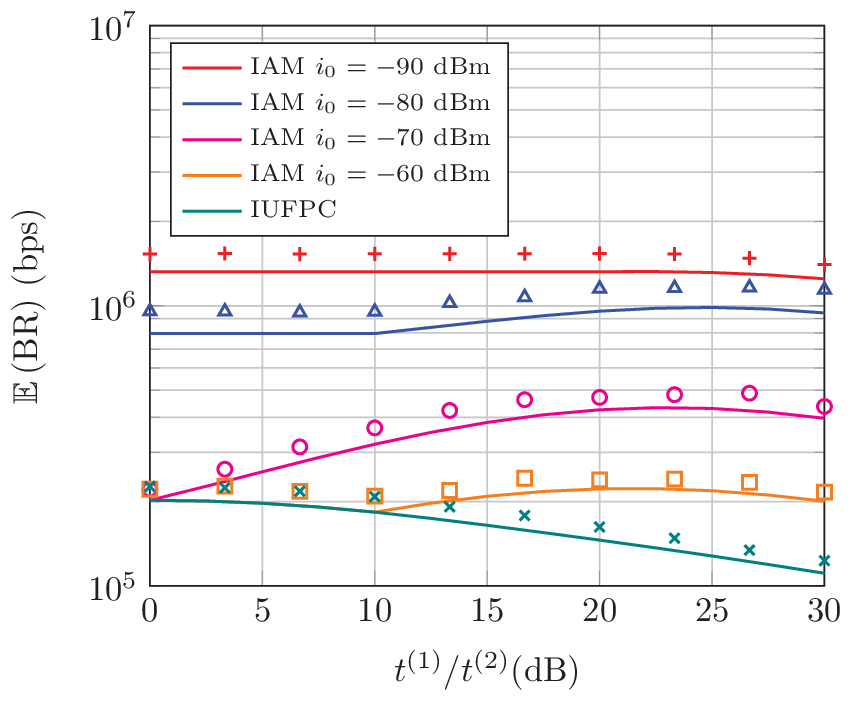}
    \caption{Average BR as a function of $t^{(1)}/t^{(2)}$ for IUFPC ($i_{0} \to \infty$) and IAM with $i_0=\{-90,-80,-70,-60\}$ dBm. $p_{\max} \to \infty$ for both schemes.}
    \label{fig:E_BR_vs_t1t2}
    \end{figure}
\else
\begin{figure}[!tbp]
  \centering
  \begin{minipage}[b]{0.32\textwidth}
    \includegraphics[width=\textwidth]{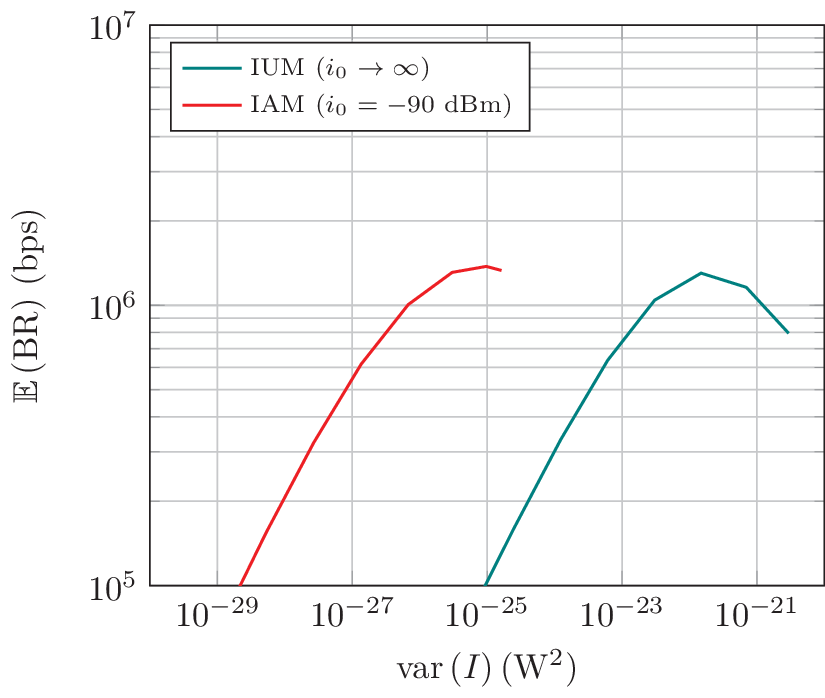}
    \caption{Average BR as a function of the variance of the interference for IAM
    ($i_{0} = -90$ dBm) and IUM ($i_{0} \to \infty$) schemes with $p_{\max} = 5$ dBm and
    $t^{(1)}/t^{(2)}=9$ dB.}
    \label{fig:avBR_vs_varI}
  \end{minipage}
  \hfill
  \begin{minipage}[b]{0.32\textwidth}
    \includegraphics[width=\textwidth]{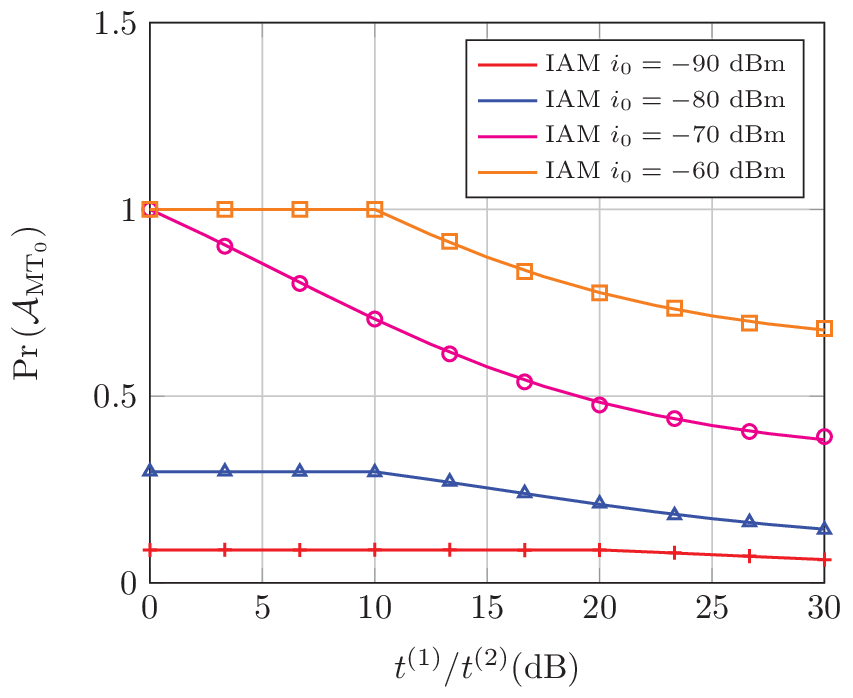}
    \caption{Probability the typical MT is active as a function of
    $t^{(1)}/t^{(2)}$ for IUFPC ($i_{0} \to \infty$) and IAM with $i_0=\{-90,-80,-70,-60\}$ dBm. $p_{\max} \to \infty$ for both schemes.}
    \label{fig:PrA_vs_t1t2}
  \end{minipage}
  \hfill
  \begin{minipage}[b]{0.32\textwidth}
    \includegraphics[width=\textwidth]{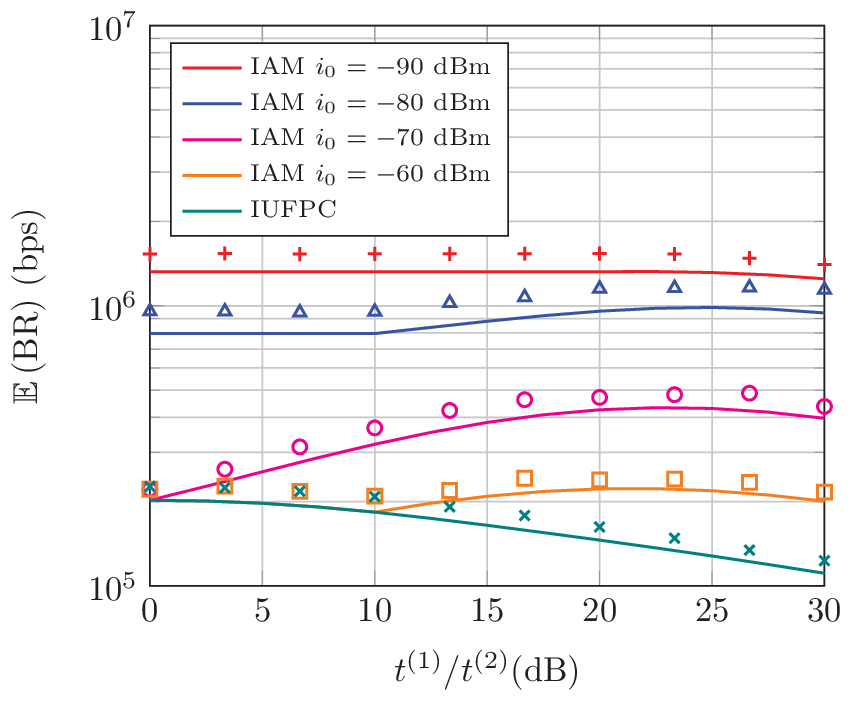}
    \caption{Average BR as a function of $t^{(1)}/t^{(2)}$ for IUFPC ($i_{0} \to \infty$) and IAM with $i_0=\{-90,-80,-70,-60\}$ dBm. $p_{\max} \to \infty$ for both schemes.}
    \label{fig:E_BR_vs_t1t2}
  \end{minipage}
\end{figure}
\fi

\subsection{Impact of the Association Weights: On UL-DL Decoupling}
\label{sec:Dependence with Association}
%
{\color{black} As shown in \cite{Singh15} and \cite{DiRenzo16b}, optimizing the performance of HCNs for DL transmission does not necessarily results in optimizing their performance in the UL. Based on the GCA criterion studied in Section IV, this implies that different cell association weights (i.e., a different ratio $t^{(1)}/t^{(2)}$ for two-tier HCNs) may be needed in the DL and in the UL. However, this approach, which is referred to as UL-DL decoupling, introduces additional implementation challenges, which require the modification of the existing network architecture and control plane.} 

In this section, motivated by these considerations, we analyze and compare IAM, IAFPC and IUFPC schemes as a function of $t^{(1)}/t^{(2)}$. The setup $t^{(1)}/t^{(2)}=0$ dB corresponds to the SPLA criterion. Some numerical illustrations are provided in Figs. \ref{fig:PrA_vs_t1t2} and \ref{fig:E_BR_vs_t1t2}, where the probability that the typical MT is active and the average BR are shown, respectively.

In Fig. \ref{fig:E_BR_vs_t1t2}, in particular, we compare the average BR of IUFPC and IAM schemes. The figure highlights important differences between these two interference management schemes for improving the performance of the UL of HCNs. First of all, we note that the average BR of the IUFPC scheme decreases as the ratio $t^{(1)}/t^{(2)}$ increases. More specifically, the best average BR is obtained if the SPLA criterion is used, which is in agreement with previously published papers \cite{Singh14b}. This originates from the fact that the larger $t^{(1)}/t^{(2)}$ is, the more MTs are associated with more distance BSs, which, due to the use of power control, results in increasing the interference in the UL. The performance trend is, on the other hand, different if the IAM scheme is used. In this case, there are several values of $i_0$ that provide a better average BR compared with IUFPC. In addition, the average BR increases as $t^{(1)}/t^{(2)}$ increases, since the excess interference that is generated under the IUFPC scheme is now kept under control by imposing the maximum interference constraint $i_0$. As observed in previous figures, Fig. \ref{fig:PrA_vs_t1t2} confirms that this gain is obtained since more MTs are turned off. 
%

Figures \ref{fig:PrA_vs_t1t2} and \ref{fig:E_BR_vs_t1t2} confirm the findings in \textbf{Remark \ref{rem:Operation regimes of IAM under GA}} and, in particular, the existence of an operating regime where the performance of IAM is independent of the association weights. Let us consider, for example, the setup for $i_0=-60$ dBm. In this case, $i_0 > p_0$ and hence, according to \textbf{Remark \ref{rem:Operation regimes of IAM under GA}}, the system is interference-unaware if $t^{(1)}/t^{(2)} \in [-10,+10]$ dB. Figure \ref{fig:E_BR_vs_t1t2}, more specifically, confirms that IAM is interference-unaware since it provides the same average BR as IUFPC for $t^{(1)}/t^{(2)} \in [-10,+10]$ dB\footnote{Only positive values (in dB) of the association weights $t^{(1)}/t^{(2)}$ are shown in Figs. \ref{fig:PrA_vs_t1t2} and \ref{fig:E_BR_vs_t1t2}.}. Similar conclusions can be drawn for other values of $i_0$, where different operating regimes can be identified as predicted in \textbf{Remark \ref{rem:Operation regimes of IAM under GA}}. If $i_0=-90$ dBm, in particular, then $i_0 < p_0$ and the system is independent of the cell association criterion for $t^{(1)}/t^{(2)} \in [-20,+20]$, which is confirmed in Figs. \ref{fig:PrA_vs_t1t2} and \ref{fig:E_BR_vs_t1t2}.
It is worth mentioning that the values of $t^{(1)}/t^{(2)}$ for which the considered system model is cell association independent are usually adopted in practical engineering applications. In particular, the authors of \cite{Singh15,Singh14} have shown that the optimal cell association ratio that optimizes the DL is usually less than $20$ dB. This is in agreement and compatible with the findings in Figs. \ref{fig:PrA_vs_t1t2} and \ref{fig:E_BR_vs_t1t2}.
%
%
%
{\color{black} In view of the numerical results and theoretical insights derived in this work, it is possible to state the following arguments in favor of such Interference-Aware Muting procedure:
\begin{enumerate}
\item Taking into account the periods where the typical MT is active and those where it is muted, the average BR is increased with IAM compared to IAFPC and IUFPC. 
%
\item Thanks to mobility and shadowing, MTs are only muted for a given period of time. 
\item Since muted MTs do not transmit, its average transmitted power is reduced compared to IAFPC and IUFPC. This has been studied with Fig. \ref{fig:comparativa_i0_iafpc_pmax_avPMT}. 
\item With IAM, there is a regime where the UL performance is independent of cell association, which eases joint design of UL and DL transmissions as it have been discussed above. 
\item It is straightforward to extend the developed model to consider other approaches where IAM take place only in a portion of the resources (e.g., bandwidth), leading to a higher system fairness. 
Let us consider, for instance, that the system bandwidth is split in two orthogonal sub-bands, e.g., 
$\mathcal{B}_A$ and $\mathcal{B}_M$.  $\mathcal{B}_A$ is restricted to active MTs, i.e., those whose interference is smaller than $i_0$, whereas the other sub-band is used to the rest of MTs\footnote{\color{black} Although it is possible to study more general frameworks for IAM, we have focused on the case considered in this paper due to space limitations and to study deeper the effect of MTs muting.}. 
%
\end{enumerate}
}
\section{Conclusion}
\label{sec:Conclusion}
In this paper, we have studied the performance of IAM: an interference management scheme for enhancing the throughput of HCNs. With the aid of stochastic geometry, we have developed a general mathematical approach for analyzing and optimizing its performance as a function of several system parameters. Simplified and insightful expressions of the throughput and other relevant performance indicators have been proposed for simplified but relevant case studies, such as in the presence of channel inversion power control and equal cell association weights. Among the many performance trends that have been identified, we have proved that, while optimizing the DL and the UL of HCNs necessitates, in general, to use different cell association weights, there exist some operating regimes where IAM is cell association independent. This is shown to simplify the design of HCNs, since no changes in their control plane is needed compared with conventional cellular networks. The mathematical frameworks and findings have been substantiated against Monte Carlo simulations, as well as the achievable performance of IAM has been compared against other IAFPC and IUFPC schemes, by highlighting several important trade-offs in terms of system fairness and system throughput.


\appendices

\section{Proof of Proposition \ref{prop:IAM Xjm}}
\label{app:Proof of Proposition IAM Xjm}
The probability that a MT is active is by definition as follows:
\begin{equation}
\label{eq:PrA}
\Pr \left( {{{\cal A}_{{\rm{M}}{{\rm{T}}_0}}}} \right) =\sum\nolimits_{ {j}  \in {\cal K}}  \sum\nolimits_{ {m}  \in {\cal K}} {\Pr \left( {{\cal X}_{{\rm{M}}{{\rm{T}}_0}}^{(j,m)},{{\cal A}_{{\rm{M}}{{\rm{T}}_0}}}} \right)}
\end{equation}
\noindent where $\Pr \left( {{\cal X}_{{\rm{M}}{{\rm{T}}_0}}^{(j,m)},{{\cal A}_{{\rm{M}}{{\rm{T}}_0}}}} \right)$ is the probability that the MT is active, is associated to tier $j$ and that the most interfered BS belongs to tier $m$. If $j \neq m$, it can be written as follows:
\ifTwoColumns
\begin{align}
\label{eq:IAM Xjm 1}
    \Pr & \left( \mathcal{X}^{(j,m)}_{\mathrm{MT}_0},
    \mathcal{A}_{\mathrm{MT}_0} \right)
    \overset{\mathrm{(a)}}{=} \mathbb{E}_{R^{(m)}_{\mathrm{MT}_0,(1)}}
    \mathbb{E}_{R^{(j)}_{\mathrm{MT}_0,(1)}, R^{(j)}_{\mathrm{MT}_0,(2)}}
    \Bigg[ \nonumber \\
    & \mathbf{1} \left(R^{(m)}_{\mathrm{MT}_0,(1)} >
    \left( \frac{t^{(m)}}{t^{(j)}} \right)^{\frac{1}{\alpha}}
    R^{(j)}_{\mathrm{MT}_0,(1)} \right) \times \nonumber \\
    & \mathbf{1} \left( R^{(m)}_{\mathrm{MT}_0,(1)} < R^{(j)}_{\mathrm{MT}_0,(2)} \right)
    \times \nonumber \\
    & \mathbf{1} \left( R^{(j)}_{\mathrm{MT}_0,(1)} < \frac{1}{\tau}
    \left( \frac{p_\mathrm{max}}{p_0}\right)^{\frac{1}{\alpha \epsilon}} \right) \times
    \nonumber \\
    & \mathbf{1} \left( R^{(m)}_{\mathrm{MT}_0,(1)} >
    \left( \frac{p_0}{i_0} \right)^{\frac{1}{\alpha}}
    \frac{\left( \tau R^{(j)}_{\mathrm{MT}_0,(1)} \right)^\epsilon}{ \tau } \right)
     \Bigg]
\end{align}
\else
\begin{align}
\label{eq:IAM Xjm 1}
    \Pr & \left( \mathcal{X}^{(j,m)}_{\mathrm{MT}_0},
    \mathcal{A}_{\mathrm{MT}_0} \right)
    \overset{\mathrm{(a)}}{=} \mathbb{E}_{R^{(m)}_{\mathrm{MT}_0,(1)}}
    \mathbb{E}_{R^{(j)}_{\mathrm{MT}_0,(1)}, R^{(j)}_{\mathrm{MT}_0,(2)}}
    \Bigg[
    \nonumber \\
    & \mathbf{1} \left(R^{(m)}_{\mathrm{MT}_0,(1)} >
    \left( \frac{t^{(m)}}{t^{(j)}} \right)^{\frac{1}{\alpha}}
    R^{(j)}_{\mathrm{MT}_0,(1)} \right) \times
    \mathbf{1} \left( R^{(m)}_{\mathrm{MT}_0,(1)} < R^{(j)}_{\mathrm{MT}_0,(2)} \right)
    \times
    \nonumber \\
    & \mathbf{1} \left( R^{(j)}_{\mathrm{MT}_0,(1)} < \frac{1}{\tau}
    \left( \frac{p_\mathrm{max}}{p_0}\right)^{\frac{1}{\alpha \epsilon}} \right) \times
     \mathbf{1} \left( R^{(m)}_{\mathrm{MT}_0,(1)} >
    \left( \frac{p_0}{i_0} \right)^{\frac{1}{\alpha}}
    \frac{\left( \tau R^{(j)}_{\mathrm{MT}_0,(1)} \right)^\epsilon}{ \tau } \right)
     \Bigg]
\end{align}
\fi
\noindent where (a) is obtained by definition of expectation formulated with the aid of indicator functions.

To compute this expectation, the PDF of the distance of the nearest BS and of the joint PDF of the distances of the nearest and second nearest BSs are needed. By definition of PPP, they are equal to $f_{R^{(j)}_{\mathrm{MT}_0,(1)}}(r) = 2 \pi \lambda^{(j)} r \mathrm{e}^{-\pi \lambda^{(j)} r^2}$ and $f_{R^{(j)}_{\mathrm{MT}_0,(1)}, R^{(j)}_{\mathrm{MT}_0,(2)}}(r_1,r_2) = 4 \left( \pi \lambda^{(j)} \right)^2 r_1 r_2  \mathrm{e}^{-\pi \lambda^{(j)} r_2^2}$ for $r_1 < r_2$, respectively, \cite{Martin-Vega14}. With the aid of these PDFs, we obtain:
%
%
%
\ifTwoColumns
\begin{align}
\label{eq:PrXjm part 2}
& \Pr  \left( \mathcal{X}^{(j,m)}_{\mathrm{MT}_0},
    \mathcal{A}_{\mathrm{MT}_0} \right)
    = \int\limits_{r^{(j)}_1=0}^{\frac{1}{\tau}
    \left( \frac{p_\mathrm{max}}{p_0} \right)^{\frac{1}{\alpha \epsilon}}}
    \times \nonumber \\
    &\int \limits^{\infty}_{r^{(m)}_1=\mathrm{max}
    \left( \left( \frac{p_0}{i_0}\right)^{\frac{1}{\alpha}}
    \frac{\left(\tau r^{(j)}_{1} \right)^\epsilon}{\tau},
    \left(\frac{t^{(m)}}{t^{(j)}} \right)^\frac{1}{\alpha} r^{(j)}_1 \right)}
    f_{R^{(m)}_{\mathrm{MT}_0,(1)}} \left( r^{(m)}_1 \right)
    \nonumber \\
    & \underbrace{
    \int\limits^{\infty}_{r^{(j)}_2 = \mathrm{max} \left( r^{(j)}_1, r^{(m)}_1\right)} \!\!
    f_{R^{(j)}_{\mathrm{MT}_0,(1)}, R^{(j)}_{\mathrm{MT}_0,(2)}}(r^{(j)}_1,r^{(j)}_2)
     \mathrm{d}r^{(j)}_2}
    _{\varphi \left( r^{(j)}_1, r^{(m)}_1 \right) }
    \mathrm{d}r^{(m)}_1 \mathrm{d}r^{(j)}_1
\end{align}
\else
\begin{align}
\label{eq:PrXjm part 2}
\Pr & \left( \mathcal{X}^{(j,m)}_{\mathrm{MT}_0},
    \mathcal{A}_{\mathrm{MT}_0} \right)
    = \int\limits_{r^{(j)}_1=0}^{\frac{1}{\tau}
    \left( \frac{p_\mathrm{max}}{p_0} \right)^{\frac{1}{\alpha \epsilon}}}
    \int \limits^{\infty}_{r^{(m)}_1=\mathrm{max}
    \left( \left( \frac{p_0}{i_0}\right)^{\frac{1}{\alpha}}
    \frac{\left(\tau r^{(j)}_{1} \right)^\epsilon}{\tau},
    \left(\frac{t^{(m)}}{t^{(j)}} \right)^\frac{1}{\alpha} r^{(j)}_1 \right)}
    f_{R^{(m)}_{\mathrm{MT}_0,(1)}} \left( r^{(m)}_1 \right)
    \nonumber \\
    &
    \int\limits^{\infty}_{r^{(j)}_2 = \mathrm{max} \left( r^{(j)}_1, r^{(m)}_1\right)} \!\!
    f_{R^{(j)}_{\mathrm{MT}_0,(1)}, R^{(j)}_{\mathrm{MT}_0,(2)}}(r^{(j)}_1,r^{(j)}_2)
     \mathrm{d}r^{(j)}_2
    \mathrm{d}r^{(m)}_1 \mathrm{d}r^{(j)}_1
\end{align}
\fi

The computation of the two-fold integral leads to the function $\nu^{(j)}(v)$ that is provided in (\ref{eq:nu}).

The case $j=m$ can be solved by using an approach similar to the previous case. The final result corresponds to the function $\eta^{(j)}(v)$ available in (\ref{eq:kappa}).

By combining both cases $j=m$ and $j \neq m$, $\Pr \left( \mathcal{X}_{\mathrm{MT}_{0}}^{(j,m)},\mathcal{A}_{\mathrm{MT}_{0}}
    \right)$ can be written as follows:
\ifTwoColumns
\begin{align}
    & \Pr  \left( \mathcal{X}^{(j,m)}_{\mathrm{MT}_0},  \mathcal{A}_{\mathrm{MT}_0} \right) =
    \int_{0}^{\frac{1}{\tau}\left( \frac{p_\mathrm{max}}{p_0} \right)^\frac{1}{\alpha \epsilon}}
    \nonumber \\ & \quad \times
    \left( \mathbf{1} \left( j \neq m \right) \nu^{(j,m)} (v) +  \right.
    \left. \mathbf{1} \left( j = m \right) \eta^{(j)} (v) \right) \mathrm{d}v
\end{align}
\else
\begin{align}
    \Pr  \left( \mathcal{X}^{(j,m)}_{\mathrm{MT}_0},  \mathcal{A}_{\mathrm{MT}_0} \right) =
    \int_{0}^{\frac{1}{\tau} \left( \frac{p_\mathrm{max}}{p_0} \right)^\frac{1}{\alpha \epsilon}}
    \left( \mathbf{1} \left( j \neq m \right) \nu^{(j,m)} (v) +  \right.
    \left. \mathbf{1} \left( j = m \right) \eta^{(j)} (v) \right) \mathrm{d}v
\end{align}
\fi

The proof follows by computing the summation over $m\in\mathcal{K}$ in (\ref{eq:PrA}).

\section{Proof of Lemma \ref{prop:IAM LI}}
\label{app:IAM LI}
The Laplace transform of the interference can be expressed as follows:
\ifTwoColumns
\begin{align}
\label{eq:IAM LI(s|Xj) Part 1}
    \mathcal{L}_{I} & \left( s|\mathcal{X}_{\mathrm{MT}_{0}}^{(j)} \right)
    = \mathbb{E}_{I}\left[ \mathrm{e}^{-sI}|\mathcal{X}_{\mathrm{MT}_{0}}^{(j)} \right]
    = \prod\limits_{k\in \mathcal{K}}{\mathbb{E}_{\Psi^{(k)}}}
    \nonumber \\
    & \prod\limits_{\mathrm{MT}_{i}\in \Psi ^{(k)}}{ \mathbb{E}_{R_{\mathrm{MT}_{i}}}}
    \left[
    \mathbb{E}_{H_{\mathrm{MT}_{i}}}\exp \left( -sH_{\mathrm{MT}_{i}}\left( \tau
    D_{\mathrm{MT}_{i}} \right)^{-\alpha }
    \right. \right.
    \nonumber \\
    & \left( \tau R_{\mathrm{MT}_{i}} \right)^{\alpha \epsilon }p_{0}
    \left. \left. \mathbf{1}\left( \mathcal{O}_{\mathrm{MT}_{i}}^{(j,k)}
    \right)\mathbf{1}\left( \mathcal{Z}_{\mathrm{MT}_{i}} \right) \right)
    |\mathcal{X}_{\mathrm{MT}_{i}}^{(k)},\mathcal{A}_{\mathrm{MT}_{i}} \right]
\end{align}
\else
\begin{align}
\label{eq:IAM LI(s|Xj) Part 1}
    \hspace{-0.5cm} \mathcal{L}_{I} & \left( s|\mathcal{X}_{\mathrm{MT}_{0}}^{(j)} \right)
    = \mathbb{E}_{I}\left[ \mathrm{e}^{-sI}|\mathcal{X}_{\mathrm{MT}_{0}}^{(j)} \right]
    = \prod\limits_{k\in \mathcal{K}}{\mathbb{E}_{\Psi^{(k)}}}
    \nonumber \\
    & \prod\limits_{\mathrm{MT}_{i}\in \Psi ^{(k)}}{ \mathbb{E}_{R_{\mathrm{MT}_{i}}}}
    \left[
    \mathbb{E}_{H_{\mathrm{MT}_{i}}}\exp \left( -sH_{\mathrm{MT}_{i}}\left( \tau
    D_{\mathrm{MT}_{i}} \right)^{-\alpha }
    \right. \right.
     \left( \tau R_{\mathrm{MT}_{i}} \right)^{\alpha \epsilon }p_{0}
    \left. \left. \mathbf{1}\left( \mathcal{O}_{\mathrm{MT}_{i}}^{(j,k)}
    \right)\mathbf{1}\left( \mathcal{Z}_{\mathrm{MT}_{i}} \right) \right)
    |\mathcal{X}_{\mathrm{MT}_{i}}^{(k)},\mathcal{A}_{\mathrm{MT}_{i}} \right]
\end{align}
\fi

By applying the Probability Generating Functional (PGF) theorem in \cite{Haenggi13} and computing the expectation with respect to the channel fading, $\mathcal{L}_{I}  \left( s|\mathcal{X}_{\mathrm{MT}_{0}}^{(j)} \right)$ is as follows:
\ifTwoColumns
\begin{align}
\label{eq:IAM LI(s|Xj) Part 1a}
    &
    \exp \left( -\sum\limits_{k\in \mathcal{K}}{2\pi \lambda ^{(k)}}
    \int\limits_{\rho=0}^{\infty }{\mathbb{E}_{R_{\mathrm{MT}_{i}}}}
    \left[
    \mathbf{1}\left( \mathcal{O}_{\mathrm{MT}_{i}}^{(j,k)} \right)\mathbf{1}\left(
    \mathcal{Z}_{\mathrm{MT}_{i}} \right)
    \right. \right. \nonumber \\
    & \left. \left. \frac{s\left( \tau \rho  \right)^{-\alpha }\left(
    \tau R_{\mathrm{MT}_{i}} \right)^{\alpha \epsilon }p_{0}}{1+s
    \left( \tau \rho  \right)^{-\alpha }\left( \tau R_{\mathrm{MT}_{i}} \right)
    ^{\alpha \epsilon }p_{0}}\rho
    |\mathcal{X}_{\mathrm{MT}_{i}}^{(k)},\mathcal{A}_{\mathrm{MT}_{i}}
    \right]
    \mathrm{d}\rho  \right)
\end{align}
\else
\begin{align}
\label{eq:IAM LI(s|Xj) Part 1a}
    &
    \exp \left( -\sum\limits_{k\in \mathcal{K}}{2\pi \lambda ^{(k)}}
    \int\limits_{\rho=0}^{\infty }{\mathbb{E}_{R_{\mathrm{MT}_{i}}}}
    \left[
    \mathbf{1}\left( \mathcal{O}_{\mathrm{MT}_{i}}^{(j,k)} \right)\mathbf{1}\left(
    \mathcal{Z}_{\mathrm{MT}_{i}} \right)
     \frac{s\left( \tau \rho  \right)^{-\alpha }\left(
    \tau R_{\mathrm{MT}_{i}} \right)^{\alpha \epsilon }p_{0}}{1+s
    \left( \tau \rho  \right)^{-\alpha }\left( \tau R_{\mathrm{MT}_{i}} \right)
    ^{\alpha \epsilon }p_{0}}\rho
    |\mathcal{X}_{\mathrm{MT}_{i}}^{(k)},\mathcal{A}_{\mathrm{MT}_{i}}
    \right]
    \mathrm{d}\rho  \right)
\end{align}
\fi

By conditioning on the event $\mathcal{Q}^{(n)}_{\mathrm{MT}_i}$ defined in (\ref{eq:Xjm}) and by using the total probability theorem, $\mathcal{L}_{I}  \left( s|\mathcal{X}_{\mathrm{MT}_{0}}^{(j)} \right)$ can be written as follows:
\ifTwoColumns
\begin{align}
\label{eq:IAM LI(s|Xj) Part 1a}
    &\exp \left(-\sum\limits_{k\in \mathcal{K}}{2\pi \lambda ^{(k)}}
    \sum\limits_{n\in \mathcal{K}}{\Pr \left( \mathcal{Q}_{
    \mathrm{MT}_{i}}^{(n)}|\mathcal{X}_{\mathrm{MT}_{i}}^{(k)},
    \mathcal{A}_{\mathrm{MT}_{i}} \right) } \right.
    \nonumber \\
    & \int\limits_{\rho=0}^{\infty }{\mathbb{E}_{R_{\mathrm{MT}_{i}}}}
    \left[
    \mathbf{1}\left( \mathcal{O}_{\mathrm{MT}_{i}}^{(j,k)} \right)\mathbf{1}
    \left( \mathcal{Z}_{\mathrm{MT}_{i}} \right)
    \right.
    \nonumber \\
    & \left. \frac{s\left( \tau \rho  \right)^{-\alpha }\left( \tau R_{\mathrm{MT}_{i}}
    \right)^{\alpha \epsilon }p_{0}}{1+s\left( \tau \rho  \right)^{-\alpha }
    \left( \tau R_{\mathrm{MT}_{i}} \right)^{\alpha \epsilon }p_{0}}\rho
    |\left. \mathcal{X}_{\mathrm{MT}_{i}}^{(k,n)},\mathcal{A}_{\mathrm{MT}_{i}} \right]
    \mathrm{d}\rho \right)
\end{align}
\else
\begin{align}
\label{eq:IAM LI(s|Xj) Part 1a}
    &\exp \left(-\sum\limits_{k\in \mathcal{K}}{2\pi \lambda ^{(k)}}
    \sum\limits_{n\in \mathcal{K}}{\Pr \left( \mathcal{Q}_{
    \mathrm{MT}_{i}}^{(n)}|\mathcal{X}_{\mathrm{MT}_{i}}^{(k)},
    \mathcal{A}_{\mathrm{MT}_{i}} \right) } \right.
    \nonumber \\
    & \int\limits_{\rho=0}^{\infty }{\mathbb{E}_{R_{\mathrm{MT}_{i}}}}
    \left[
    \mathbf{1}\left( \mathcal{O}_{\mathrm{MT}_{i}}^{(j,k)} \right)\mathbf{1}
    \left( \mathcal{Z}_{\mathrm{MT}_{i}} \right)
    \frac{s\left( \tau \rho  \right)^{-\alpha }\left( \tau R_{\mathrm{MT}_{i}}
    \right)^{\alpha \epsilon }p_{0}}{1+s\left( \tau \rho  \right)^{-\alpha }
    \left( \tau R_{\mathrm{MT}_{i}} \right)^{\alpha \epsilon }p_{0}}\rho
    |\left. \mathcal{X}_{\mathrm{MT}_{i}}^{(k,n)},\mathcal{A}_{\mathrm{MT}_{i}} \right]
    \mathrm{d}\rho \right)
\end{align}
\fi

The next step is the computation of the expectation with respect to $R_{\mathrm{MT}_i}$ by conditioning on $\mathcal{X}^{(k,n)}_{\mathrm{MT}_i} \cap \mathcal{A}_{\mathrm{MT}_i}$ and by applying the definition of the event $\mathcal{O}^{(j,k)}_{\mathrm{MT}_i}$ in (\ref{eq:O_MTi}) and of $\mathcal{Z}_{\mathrm{MT}_i}$ in (\ref{eq:ZMTi}). In particular, by conditioning on $\mathcal{X}^{(k,n)}_{\mathrm{MT}_i} \cap \mathcal{A}_{\mathrm{MT}_i}$ for $\mathrm{MT}_i \in \Psi^{(k)}$, the distances $R_{\mathrm{MT}_i}$ are independent and identically distributed random variables whose PDF is in (\ref{eq:f_RMT_0 cond XjmA IAM}). With some algebra, $\mathcal{L}_{I}  \left( s|\mathcal{X}_{\mathrm{MT}_{0}}^{(j)} \right)$ can be written as follows:
\ifTwoColumns
\begin{align}
\label{eq:IAM LI(s|Xj) Part 2}
    &\exp \left( -\sum\limits_{k\in \mathcal{K}}{2\pi \lambda ^{(k)}}\sum
    \limits_{n\in \mathcal{K}}{\Pr \left( \mathcal{Q}_{\mathrm{MT}_{i}}^{(n)}
    |\mathcal{X}_{\mathrm{MT}_{i}}^{(k)},\mathcal{A}_{\mathrm{MT}_{i}} \right)}
    \right. \nonumber \\
    & \int\limits_{r=0}^{\infty }{f_{R_{\mathrm{MT}_{i}}}\left( r
    |\mathcal{X}_{\mathrm{MT}_{i}}^{(k,n)},\mathcal{A}_{\mathrm{MT}_{i}} \right)}
    \int\limits_{\rho =\max \left( \left( \frac{t^{(j)}}{t^{(k)}} \right)
    ^{\frac{1}{\alpha }}r,\left( \frac{p_{0}}{i_{0}} \right)^{\frac{1}{\alpha }}
    \frac{\left( \tau r \right)^{\epsilon }}{\tau } \right)}^{\infty }
    \nonumber \\
    & \left.
    \frac{s
    \left( \tau \rho  \right)^{-\alpha }\left( \tau r \right)
    ^{\alpha \epsilon }p_{0}}{1+s\left( \tau \rho  \right)^{-\alpha }
    \left( \tau r \right)^{\alpha \epsilon }p_{0}}\rho
    \mathrm{d}\rho \mathrm{d}r
    \right)
\end{align}
\else
\begin{align}
\label{eq:IAM LI(s|Xj) Part 2}
    &\exp \left( -\sum\limits_{k\in \mathcal{K}}{2\pi \lambda ^{(k)}}\sum
    \limits_{n\in \mathcal{K}}{\Pr \left( \mathcal{Q}_{\mathrm{MT}_{i}}^{(n)}
    |\mathcal{X}_{\mathrm{MT}_{i}}^{(k)},\mathcal{A}_{\mathrm{MT}_{i}} \right)}
    \right. \nonumber \\
    & \int\limits_{r=0}^{\infty }{f_{R_{\mathrm{MT}_{i}}}\left( r
    |\mathcal{X}_{\mathrm{MT}_{i}}^{(k,n)},\mathcal{A}_{\mathrm{MT}_{i}} \right)}
    \int\limits_{\rho =\max \left( \left( \frac{t^{(j)}}{t^{(k)}} \right)
    ^{\frac{1}{\alpha }}r,\left( \frac{p_{0}}{i_{0}} \right)^{\frac{1}{\alpha }}
    \frac{\left( \tau r \right)^{\epsilon }}{\tau } \right)}^{\infty }
     \left.
    \frac{s
    \left( \tau \rho  \right)^{-\alpha }\left( \tau r \right)
    ^{\alpha \epsilon }p_{0}}{1+s\left( \tau \rho  \right)^{-\alpha }
    \left( \tau r \right)^{\alpha \epsilon }p_{0}}\rho
    \mathrm{d}\rho \mathrm{d}r
    \right)
\end{align}
\fi

The proof follows by computing the inner integral.

\ifCLASSOPTIONcaptionsoff
  \newpage
\fi


\bibliographystyle{IEEEtran}
\bibliography{UL_Power_Control}

\begin{thebibliography}{10}
\providecommand{\url}[1]{#1}
\csname url@samestyle\endcsname
\providecommand{\newblock}{\relax}
\providecommand{\bibinfo}[2]{#2}
\providecommand{\BIBentrySTDinterwordspacing}{\spaceskip=0pt\relax}
\providecommand{\BIBentryALTinterwordstretchfactor}{4}
\providecommand{\BIBentryALTinterwordspacing}{\spaceskip=\fontdimen2\font plus
\BIBentryALTinterwordstretchfactor\fontdimen3\font minus
  \fontdimen4\font\relax}
\providecommand{\BIBforeignlanguage}[2]{{%
\expandafter\ifx\csname l@#1\endcsname\relax
\typeout{** WARNING: IEEEtran.bst: No hyphenation pattern has been}%
\typeout{** loaded for the language `#1'. Using the pattern for}%
\typeout{** the default language instead.}%
\else
\language=\csname l@#1\endcsname
\fi
#2}}
\providecommand{\BIBdecl}{\relax}
\BIBdecl

\bibitem{Elshaer14}
H.~Elshaer, F.~Boccardi, M.~Dohler, and R.~Irmer, ``{Downlink and Uplink
  Decoupling: A disruptive architectural design for 5G networks},'' in
  \emph{2014 IEEE Global Communications Conference}, Dec 2014, pp. 1798--1803.

\bibitem{Singh15}
S.~Singh, X.~Zhang, and J.~G. Andrews, ``{Joint Rate and SINR Coverage Analysis
  for Decoupled Uplink-Downlink Biased Cell Associations in HetNets},''
  \emph{IEEE Transactions on Wireless Communications}, vol.~14, no.~10, pp.
  5360--5373, Oct 2015.

\bibitem{Novlan13}
T.~Novlan, H.~Dhillon, and J.~Andrews, ``{Analytical Modeling of Uplink
  Cellular Networks},'' \emph{Wireless Communications, IEEE Transactions on},
  vol.~12, no.~6, pp. 2669--2679, June 2013.

\bibitem{ElSawy14}
H.~ElSawy and E.~Hossain, ``{On Stochastic Geometry Modeling of Cellular Uplink
  Transmission With Truncated Channel Inversion Power Control},''
  \emph{Wireless Communications, IEEE Transactions on}, vol.~13, no.~8, pp.
  4454--4469, Aug 2014.

\bibitem{DiRenzo16b}
M.~D. Renzo and P.~Guan, ``{Stochastic Geometry Modeling and System-Level
  Analysis of Uplink Heterogeneous Cellular Networks With Multi-Antenna Base
  Stations},'' \emph{IEEE Trans. Communications}, vol.~64, no.~6, pp.
  2453--2476, June 2016.

\bibitem{Martin-Vega16}
F.~J. Martin-Vega and et~al., ``{Analytical Modeling of Interference Aware
  Power Control for the Uplink of Heterogeneous Cellular Networks},''
  \emph{IEEE Transactions on Wireless Communications}, vol.~PP, no.~99, pp.
  1--1, 2016.

\bibitem{Zhang12}
H.~Zhang, N.~Prasad, S.~Rangarajan, S.~Mekhail, S.~Said, and R.~Arnott,
  ``{Standards-compliant LTE and LTE-A uplink power control},'' in
  \emph{Communications (ICC), 2012 IEEE International Conference on}, June
  2012, pp. 5275--5279.

\bibitem{Shin02}
H.~Shin and J.~H. Lee, ``{Channel reliability estimation for turbo decoding in
  rayleigh fading channels with imperfect channel estimates},'' \emph{IEEE
  Communications Letters}, vol.~6, no.~11, pp. 503--505, Nov 2002.

\bibitem{wimo14}
\BIBentryALTinterwordspacing
M.~C. Aguayo-Torres and et~al., ``{WM-SIM LTE link simulator},'' University of
  Malaga, Tech. Rep., 2014. [Online]. Available:
  \url{http://riuma.uma.es/xmlui/handle/10630/7438}
\BIBentrySTDinterwordspacing

\bibitem{Martin-Vega13}
F.~Martin-Vega, I.~Delgado-Luque, F.~Blanquez-Casado, G.~Gomez,
  M.~Aguayo-Torres, and J.~Entrambasaguas, ``{LTE Performance over High Speed
  Railway Channel},'' in \emph{Vehicular Technology Conference (VTC Fall), IEEE
  78th}, Sept 2013.

\bibitem{Haenggi13}
M.~Haenggi, \emph{{Stochastic Geometry for Wireless Networks}}.\hskip 1em plus
  0.5em minus 0.4em\relax Cambridge University Press, 2013.

\bibitem{Dhillon14}
H.~Dhillon and J.~Andrews, ``{Downlink Rate Distribution in Heterogeneous
  Cellular Networks under Generalized Cell Selection},'' \emph{Wireless
  Communications Letters, IEEE}, vol.~3, no.~1, pp. 42--45, February 2014.

\bibitem{Lin15}
Y.~Lin, W.~Bao, W.~Yu, and B.~Liang, ``{Optimizing User Association and
  Spectrum Allocation in HetNets: A Utility Perspective},'' \emph{IEEE Journal
  on Selected Areas in Communications}, vol.~33, no.~6, pp. 1025--1039, June
  2015.

\bibitem{Andrews11}
J.~G. Andrews, F.~Baccelli, and R.~K. Ganti, ``{A Tractable Approach to
  Coverage and Rate in Cellular Networks},'' \emph{IEEE Transactions on
  Communications}, vol.~59, no.~11, pp. 3122--3134, Nov. 2011.

\bibitem{Sesia09}
S.~Sesia and et~al., \emph{{LTE, The UMTS Long Term Evolution: From Theory to
  Practice}}.\hskip 1em plus 0.5em minus 0.4em\relax Wiley Publishing, 2009.

\bibitem{DiRenzo16a}
M.~D. Renzo and et~al., ``{The Intensity Matching Approach: A Tractable
  Stochastic Geometry Approximation to System-Level Analysis of Cellular
  Networks},'' \emph{IEEE Transactions on Wireless Communications}, vol.~15,
  no.~9, Sept 2016.

\bibitem{Singh14b}
S.~Singh, F.~Baccelli, and J.~Andrews, ``{On Association Cells in Random
  Heterogeneous Networks},'' \emph{Wireless Communications Letters, IEEE},
  vol.~3, no.~1, pp. 70--73, February 2014.

\bibitem{Ferenc07}
J.-S. Ferenc and Z.~N\'eda, ``{On the size distribution of Poisson Voronoi
  cells},'' \emph{Physica A: Statistical Mechanics and its Applications}, vol.
  385, no.~2, p. 518–526, Nov. 2007.

\bibitem{3gpp36872}
3GPP, ``{Technical Specification Group Radio Access Network; Small cell
  enhancements for E-UTRA and E-UTRAN - Physical layer aspects},'' {3rd
  Generation Partnership Project (3GPP)}, TR {36.872}, Apr. 2013.

\bibitem{Singh14}
S.~Singh and J.~Andrews, ``{Joint Resource Partitioning and Offloading in
  Heterogeneous Cellular Networks},'' \emph{Wireless Communications, IEEE
  Transactions on}, vol.~13, no.~2, pp. 888--901, February 2014.

\bibitem{Martin-Vega14}
F.~Martin-Vega, F.~Lopez-Martinez, G.~Gomez, and M.~Aguayo-Torres,
  ``{Multi-user coverage probability of uplink cellular systems: A stochastic
  geometry approach},'' in \emph{Global Communications Conference (GLOBECOM),
  2014 IEEE}, Dec 2014.

\end{thebibliography}

\end{document}